%% file: Generalmodelexpmalliavin.tex
\documentclass[11pt, a4paper, oneside, reqno]{amsart}

\usepackage[utf8]{inputenc}
\usepackage{amsmath, amsthm, amssymb, amsfonts}
\usepackage[foot]{amsaddr}
\usepackage[english]{babel}
\usepackage{bbding}
\usepackage{booktabs}
\usepackage[unicode=true,pdfusetitle,bookmarks=true,bookmarksnumbered=true,bookmarksopen=true,bookmarksopenlevel=1,breaklinks=true,pdfborder={0 0 1},backref=false,colorlinks=true]
 {hyperref}
\hypersetup{
 linkcolor=blue, citecolor=blue, urlcolor=blue, filecolor=blue,pdfpagelayout=OneColumn, pdfnewwindow=true,pdfstartview=XYZ, plainpages=false}
\usepackage{crossreftools}
\usepackage{cleveref}
\usepackage{xcolor}
\usepackage{diagbox}
\usepackage{dsfont}
\usepackage[shortlabels]{enumitem}
\usepackage{float}
\usepackage{footmisc}
\usepackage[T1]{fontenc}
\usepackage[a4paper]{geometry} 
\usepackage{graphicx}
\usepackage{lmodern}
\usepackage{mathrsfs}
\usepackage[all]{xy} 
\usepackage[numbers]{natbib}
\usepackage{rotating}
\usepackage{setspace}
\usepackage{stmaryrd}
\usepackage{subcaption}
\usepackage{url}
\usepackage{xspace}
\usepackage{xcolor}

\setlist{leftmargin=*, topsep=0.5em, parsep=5pt, itemsep=1em, labelindent=0pt, align=left}

\theoremstyle{definition}
\newtheorem{theorem}{Theorem}[section]

\newtheorem{definition}{Definition}[section]

\newtheorem{proposition}{Proposition}[section]
\newtheorem{remark}{Remark}[section]
\newtheorem{assumption}{Assumption}

\newtheorem{lemma}{Lemma}[section]
\newtheorem{corollary}{Corollary}[section]

\numberwithin{equation}{section}

\newcommand{\Qro}{\mathbb{Q}}
\newcommand{\Ex}{\mathbb{E}}

\newcommand{\EE}{\mathcal{E}}
\newcommand{\FF}{\mathcal{F}}

\newcommand{\NN}{\mathcal{N}}

\newcommand{\rind}[1]{\textbf{1}_{#1}}

\newcommand{\borel}{\mathcal{B}(\mathbb{R})}
\newcommand{\reals}{\mathbb{R}}

\newcommand{\naturals}{\mathbb{N}}

\newcommand{\dd}{\mathrm {d}}

\newcommand{\vep}{\varepsilon}
\numberwithin{table}{section}

\definecolor{blue}{HTML}{1F77B4}
\definecolor{orange}{HTML}{FF7F0E}
\definecolor{green}{HTML}{2CA02C}
\definecolor{red}{HTML}{D62728}
\definecolor{purple}{HTML}{9467BD}
\definecolor{brown}{HTML}{8C564B}
\definecolor{pink}{HTML}{E377C2}
\definecolor{grey}{HTML}{7F7F7F}
\definecolor{yellow}{HTML}{BCBD22}
\definecolor{cyan}{HTML}{17BECF}
\definecolor{turquoise}{HTML}{3FE0D0}

\newcommand\red[1]{\color{red}#1}  
\newcommand\blue[1]{\color{blue}#1}

\newcommand\blfootnote[1]{%
  \begingroup
  \renewcommand\thefootnote{}\footnote{#1}%
  \addtocounter{footnote}{-1}%
  \endgroup
}

\crefformat{enumi}{(#2#1#3)}
\Crefformat{enumi}{(#2#1#3)}
\crefrangeformat{enumi}{(#3#1#4) to(#5#2#6)}
\Crefrangeformat{enumi}{(#3#1#4) to(#5#2#6)}
\crefrangeformat{equation}{eqs. (#3#1#4)--(#5#2#6)}

\begin{document}
\title[]{Explicit approximations of option prices via Malliavin calculus in a general stochastic volatility framework}
\author{Kaustav Das$^{\dagger \ddagger}$}
\address{$^\dagger$School of Mathematics, Monash University, Victoria, 3800 Australia.}
\address{$^\ddagger$Centre for Quantitative Finance and Investment Strategies, Monash University, Victoria, 3800 Australia.}
\author{Nicolas Langren\'e$^\S$}
\address{$^\S$Guangdong Provincial/Zhuhai Key Laboratory of Interdisciplinary Research and Application for Data Science, Beijing Normal-Hong Kong Baptist University, Zhuhai, 519087, China.}
\email{kaustav.das@monash.edu}
\email{nicolaslangrene@bnbu.edu.cn}
\date{}

\DeclareGraphicsExtensions{.pdf,jpg}
\vspace{12pt}

\onehalfspacing

\begin{abstract}

We establish an explicit approximation formula for European put option prices within a general stochastic volatility model with time-dependent parameters. Our methodology is based on expansions of the mixing representation of the put option price as an expectation of the Black-Scholes formula, in which the resulting terms are calculated explicitly by Malliavin calculus. We obtain an explicit representation of the error generated by the expansion procedure, and bound it in terms of moments of functionals of the underlying volatility process. Under the assumption of piecewise-constant parameters, our approximation formulas become closed-form, and compatible with a proposed fast calibration scheme. Finally, we perform a numerical sensitivity analysis to investigate the quality of our approximation formula in the so-called Stochastic Verhulst model, and show that the errors are well within the acceptable range for application purposes. \\[1em]

\end{abstract}

\maketitle

\section{Introduction}
\blfootnote{Kaustav Das has been supported by the Australian Research Council (Grant DP220103106). Nicolas Langren\'e acknowledges the partial support of the Guangdong Provincial/Zhuhai Key Laboratory IRADS (2022B1212010006) and the BNBU Start-up Research Fund UICR0700041-22. This project was started while Nicolas Langren\'e was affiliated with CSIRO's Data61, and whilst Kaustav Das was supported from an Australian Government Research Training Program (RTP) Scholarship.}
\label{sec:introduction3}

\noindent In this article, we consider the European put option pricing problem in a general stochastic volatility model with time-dependent parameters. Namely, the volatility process satisfies the stochastic differential equation $\dd V_t = \alpha(t,V_t) \dd t + \beta(t,V_t) \dd B_t$ where the drift and diffusion coefficients satisfy some regularity properties. The main contributions of this article are an explicit second-order approximation for the price of a European put option, an explicit form for the error induced in the approximation, as well as a fast calibration scheme. These approximation formulas are written in terms of certain iterated integral operators; under the assumption of piecewise-constant parameters, these operators are closed-form, yielding a closed-form approximation to the European put option price. We provide sufficient conditions regarding existence of an equivalent martingale measure in our general inhomogeneous stochastic volatility framework (\Cref{thm:existencemeasure}, which is an extension of \citep[][Theorem 2.4(i)]{lions2007correlations}). We then obtain a bound on the error term induced by the approximation procedure in terms of moments of functionals of the underlying volatility process. Lastly, we perform a numerical sensitivity analysis in the so-called Stochastic Verhulst model in order to assess our approximation procedure in application. Our approximation methodology combines the mixing solution approach \citep{hull1987pricing, Willard97,romano1997contingent} to reduce the dimensionality of the pricing problem, small volatility-of-volatility expansion techniques, as well as Malliavin calculus machinery, in the lines of \citep{timedepheston, IGa}.

It is well known that implied volatility is heavily dependent on the strike and maturity of European option contracts. This phenomenon is called the volatility smile, a feature the seminal Black-Scholes model fails to address due to its constant volatility assumption~\citep{bs}. In response, a number of frameworks have been proposed with the intention of accurately modelling this implied volatility surface. A popular approach is the class of stochastic volatility models: in a stochastic volatility model, the volatility (or variance) process is modelled as a stochastic process itself, possibly correlated with the underlying spot price. While stochastic volatility models are significantly more realistic than models with deterministic volatility, their greater complexity comes at a cost. Indeed, it is usually not possible to compute the prices of even the simplest contracts in a closed-form manner.

Affine stochastic volatility models are a subclass of stochastic volatility models which possess a certain amount of tractability. Specifically, affine models are those in which the characteristic function of the log-spot can be computed explicitly\footnote{More precisely, an affine model is such that the log of the characteristic function of the log-spot is an affine function.}. Examples include the Heston and Sch{\"o}bel-Zhu models \citep{heston, schobel1999stochastic}. In such models, it can be shown that it is possible to express European option prices in a quasi closed-form fashion. However, non-affine models have been empirically shown to be substantially more realistic than their affine counterparts, see for example \citep{gander2007stochastic, christoffersen2010volatility, kaeck2012volatility}. For this reason, there has recently been a significant push in the industry towards favouring non-affine models. 
Since option prices in non-affine models cannot be computed analytically, numerical procedures such as PDE and Monte-Carlo methods have been substantially developed in this literature, see \citep{zhang2015efficient, shi2016pricing}.

An interesting alternative methodology for option pricing is closed-form approximation, where the option price is approximated by a closed-form expression. The purpose of obtaining a closed-form approximation is to achieve a `best of both worlds' scenario; one can utilise a realistic and sophisticated model, yet still have a means to obtain prices of options rapidly. Moreover, since transform methods are usually not utilised, time-dependent parameters can often be handled well. One motivation for quick option pricing formulas is calibration, where the option price must be computed several times within an optimisation procedure.

Over the past few decades, closed-form approximations of option prices have been extensively studied in the literature. For example, \citep{sabr} use singular perturbation techniques to obtain an explicit approximation for the option price and implied volatility in their SABR model. \citep{lorigexplicit} derive a general closed-form expression for the price of an option via a PDE approach, as well as its corresponding implied volatility. A similar approach was utilised in \citep{lorig2022options} in order to describe the implied volatility of options on bonds in a general affine term structure model. Later on, this approach was extended to handle options on forward rates and quadratic term structure models \citep{lorig2023explicit}. \citep{alos2006generalization} show that, from the mixing solution formula, one can approximate the put option price by decomposing it into a sum of two terms, with only one of them depending on correlation. However, neither terms are explicit. \citep{alos2012decomposition} are able to achieve a decomposition of the option price via the mixing solution in terms of a classical Black-Scholes formula, with one term depending on correlation and one depending on vol-of-vol. In doing so, they are able to obtain explicit first and second-order approximations for option prices in the Heston model. \citep{Antonelli09,Antonelli10} show that under the assumption of small correlation, an expansion can be performed with respect to the mixing solution, where the resulting expectations can be computed using Malliavin calculus techniques. Similarly, in the case of the time-dependent Heston model, \citep{timedepheston} expand the mixing solution around vol-of-vol, using a combination of Taylor expansions and Malliavin calculus techniques. \citep{IGa} adapts the methodology of \citep{timedepheston} to the Inverse-Gamma model. \citep{das2018closedform} start from the mixing solution, and utilise a combination of change of measure and expansion techniques in order to obtain a second-order approximation of option prices in the Heston and GARCH diffusion models. \citep{guinea2024higher} adapt the methodology of \citep{das2018closedform} to obtain an explicit $N$-th order approximation in the case of the so-called Barndorff-Nielsen model. \citep{bergomi2012stochastic} combine PDE techniques alongside a small vol-of-vol expansion in order to deduce formulas for option prices and their associated implied volatilities. These formulas are in general written in terms of iterated integrals, which can be computed explicitly for the Heston model and Bergomi models. In the general case however, one would need to resort to numerical approximations of iterated integrals, such as Gaussian quadrature for example.

Stochastic volatility models usually either model the volatility directly, or indirectly via the variance process. A frequent assumption is that volatility or variance has some sort of mean reversion behaviour, as this is supported by empirical evidence \citep{gatheral2011volatility}. For modelling the variance process, a large class of one-factor stochastic models is given by
\begin{align*}
	\dd S_t &= S_t ( (r_t^d - r_t^f ) \dd t + \sqrt{V_t} \dd W_t ), \quad S_0,\\
	\dd V_t &=\kappa_t ( \theta_t V_t^{\hat \mu} - V_t^{\tilde \mu }) \dd t + \lambda_t V_t^{\mu} \dd B_t , \quad V_0 = v_0,\\
	\dd \langle W, B \rangle_t &= \rho_t \dd t,
\end{align*}
whereas for modelling the volatility, this class is of the form
\begin{align*}
	\dd S_t &= S_t ( (r_t^d - r_t^f ) \dd t + V_t \dd W_t ), \quad S_0,\\
	\dd V_t &=\kappa_t ( \theta_t V_t^{\hat \mu} - V_t^{\tilde \mu }) \dd t + \lambda_t V_t^{\mu} \dd B_t , \quad V_0 = v_0,\\
	\dd \langle W, B \rangle_t &= \rho_t \dd t,
\end{align*}	
for some $ \tilde \mu, \hat \mu$ and  $\mu \in \reals$.\footnote{There exist other classes of stochastic volatility models. For example, the exponential Ornstein-Uhlenbeck model \citep{wiggins1987option} is not included in either of these classes.} Some popular models in the literature include:\\
\begin{tabular}{ l  c  l  l  l  l  }
\toprule
Model &  Variance/Volatility  &Dynamics of $V$ & $\hat \mu$ & $\tilde \mu$ & $\mu$  \\[.01cm]
\midrule

Heston \text{\citep{heston}}  & Variance & $\dd V_t = \kappa_t( \theta_t - V_t) \dd t +  \lambda_t \sqrt {V_t} \dd B_t $ & 0  &1 &  1/2 \\[.01cm]

Sch{\"o}bel and Zhu \text{\citep{schobel1999stochastic}}  & Volatility &$ \dd V_t = \kappa_t( \theta_t - V_t) \dd t + \lambda_t \dd B_t$ & 0 & 1 & 0 \\[.01cm]
GARCH \text{\citep{nelson1990arch, Willard97}}  & Variance  & $\dd V_t = \kappa_t( \theta_t - V_t) \dd t +  \lambda_t V_t \dd B_t$  & 0 & 1 & 1\\[.01cm]

Inverse Gamma \text{\citep{IGa}}  & Volatility & $\dd V_t = \kappa_t( \theta_t - V_t) \dd t +  \lambda_t V_t \dd B_t $ & 0 & 1 & 1  \\[.01cm]
3/2 Model \text{\citep{lien2002option}}   & Variance & $\dd V_t = \kappa_t( \theta_t V_t - V_t^2) \dd t +  \lambda_t V_t^{3/2} \dd B_t$  & 1 & 2 & 3/2\\[.01cm]

Verhulst \text{\citep{lewis2019exact,carr2019lognormal}} & Volatility &  $\dd V_t = \kappa_t( \theta_t V_t - V_t^2) \dd t +  \lambda_t V_t \dd B_t $ & 1 & 2 & 1  \\[.01cm]
\bottomrule
\end{tabular}


We now introduce a more general class of stochastic volatility models. Let $(\Omega, \FF, (\FF_t)_{0 \leq t \leq T}, \Qro)$ be a filtered probability space, where $T$ is a finite time horizon, and $(\FF_t)_{0 \leq t \leq T}$ is the filtration, which satisfies the usual assumptions, generated by two Brownian motions $W$ and $B$ with deterministic, time-dependent instantaneous correlation $(\rho_t)_{0 \leq t \leq T}$. In particular, $\rho_t \in [-1,1]$ for any $t \in [0,T]$. Consider the following general model\footnote{Our drift formulation is for foreign exchange market purposes, but can easily be adapted to other markets such as equity or fixed income.}:
\begin{align}
	\begin{split}
	\dd S_t &= (r_t^d - r_t^f) S_t \dd t + V_t S_t \dd W_t, \quad S_0>0, \\
	\dd V_t &= \alpha(t, V_t) \dd t + \beta(t,V_t)\dd B_t, \quad V_0 = v_0, \\
	\dd \langle W, B \rangle_t &= \rho_t \dd t. \label{eqn:gen}
	\end{split}
\end{align}
where $(r_t^d)_{0\leq t \leq T}$ and $(r_t^f)_{0 \leq t \leq T}$ are the deterministic, time-dependent domestic and foreign interest rates respectively. In the following, $\Ex [\cdot]$ (sometimes $\Ex (\cdot )$ or simply $\Ex$) denotes the expectation under $\Qro$, where $\Qro$ is a domestic equivalent martingale measure which we assume to be chosen. We refer to \Cref{thm:existencemeasure} for sufficient conditions ensuring the existence of $\Qro$. In this article, partial derivatives of functions will often be written using subscript notation, for example, $\frac{\partial^3 f}{\partial x^2 \partial y} \equiv \partial_{xxy} f$.

The price of a European put option in the general model \cref{eqn:gen} with log-strike $k$ is given by
\begin{align}
	\text{Put}_\text{G} := e^{-\int_0^T r_t^d \dd t} \Ex \left [( e^k - S_T )_+ \right ] \label{eqn:priceput}
\end{align}
where $x_+ := \max \{x, 0 \}$. In this article, we obtain an explicit second-order expression for the price of a European put option. Without loss of generality, we choose to focus on put options, since our approach can be easily adapted to obtain an explicit second-order expression for European call option prices via the put-call parity. The methodology utilised in this article has been previously implemented for the subsequent models:
\begin{itemize}
\item 
For the Heston model
\begin{align*}
	\dd S_t &= S_t ( (r_t^d - r_t^f ) \dd t + \sqrt{V_t} \dd W_t ), \\
	\dd V_t &= \kappa_t ( \theta_t - V_t ) \dd t + \lambda_t \sqrt{V_t} \dd B_t, \\
	\dd \langle W, B \rangle_t &= \rho_t \dd t,
\end{align*}
this has been studied in \citep{timedepheston}.
\item 
For the Inverse-Gamma (IGa) model
\begin{align*}
	\dd S_t &= S_t ( (r_t^d - r_t^f ) \dd t + V_t \dd W_t ), \\
	\dd V_t &= \kappa_t ( \theta_t - V_t ) \dd t + \lambda_t V_t\dd B_t, \\
	\dd \langle W, B \rangle_t &= \rho_t \dd t,
\end{align*}
this has been tackled in \citep{IGa}.
\end{itemize}
The purpose of this article is to extend the methodology utilised in these aforementioned articles to option pricing in the general stochastic volatility framework \cref{eqn:gen}, where the volatility process possesses an arbitrary drift and diffusion coefficient satisfying some regularity conditions (\Cref{assregA} and \Cref{assregB}). The resulting price approximation is further explicited in the case of piecewise constant parameters, yielding a fast parameter calibration scheme. The sections are organised as follows: 

\begin{itemize}
	\item \Cref{sec:prelims3} details some preliminary calculations. First, we reparameterise the volatility process in terms of a small perturbation parameter, obtaining the process $\left (V_t^{(\vep)} \right )$. Then, we rewrite the expression for the price of a put option using the mixing solution formulation.
	\item In \Cref{sec:expprocedure3} we implement our expansion procedure. Namely, we combine a Taylor expansion of a Black-Scholes formula with a small vol-of-vol expansion of the function $\vep \mapsto V_t^{(\vep)}$ and its variants. This gives a second-order approximation to the price of a put option.
	\item  \Cref{sec:explicitprice3} is dedicated to the explicit calculation of terms induced by our expansion procedure from \Cref{expansionprocedure}. In particular, we utilise Malliavin calculus techniques in order to reduce the corresponding terms down into iterated integral expressions.
	\item In \Cref{sec:erroranalysis3} we present an explicit form for the error of our expansion methodology. In particular, we bound the error term in terms of moments of functionals of the underlying volatility process.
	\item \Cref{sec:fastcal3} computes the iterated integral operators in closed-form under the assumption of piecewise-constant parameters. This makes our pricing formulas closed-form, and allows us to establish a fast calibration scheme.
	\item \Cref{sec:numerical3} is dedicated to a numerical sensitivity analysis of our put option approximation formula in the Stochastic Verhulst model.
    \item In \Cref{appen:sabr} we examine the classical SABR model, and quickly demonstrate that our closed-form approximation formula works well in this familiar model.
\end{itemize} 
\Crefrange{appen:mixingsol}{append:expprice} contain various technical results and complex calculations required for this article.
\section{Preliminaries}
\label{sec:prelims3}
\noindent In order to ensure the well-posedness of $V$ in \cref{eqn:gen}, we will require certain assumptions on the regularity of its drift and diffusion coefficients. In addition, we will require certain restrictions on the drift and diffusion coefficients whenever our expansion procedure in \Cref{sec:expprocedure3} demands it. In anticipation of this, for the rest of this article we will enforce the following assumptions on the regularity of the drift and diffusion coefficients of $V$ in \cref{eqn:gen}.

\begin{assumption}
\label{assregA}
For $t \in [0,T]$:
\begin{enumerate}[label = (A\arabic*), ref = A\arabic*]
\item$\alpha$  is Lipschitz continuous in $x$, uniformly in $t$. \label{assregA1}
\item $\beta$ is H\"older continuous of order $\geq 1/2$ in $x$, uniformly in $t$. \label{assregA2}
\item There exists a weak solution of $V$.  \label{assregA3}
\item $\alpha$ and $\beta$ satisfy the growth bounds $x \alpha(t,x) \leq K(1+|x|^2)$ and $|\beta(t,x)|^2 \leq K(1+|x|^2)$ uniformly in $t$, where $K$ is a positive constant.
\end{enumerate}
\end{assumption}

\begin{assumption}
\label{assregB}
The following properties hold:
\begin{enumerate}[label = (B\arabic*), ref = B\arabic*]
\item The second derivative $\alpha_{xx}$ exists and is continuous a.e. in $x$ and $t \in [0,T]$. \label{assregB1}
\item The first derivative $\beta_x$ exists and is continuous a.e. in $x$ and $t \in [0,T]$. \label{assregB2}
\end{enumerate}
\end{assumption}
The purpose of \Cref{assregA} is to guarantee the existence of a pathwise unique strong solution to $V$ in \cref{eqn:gen}, see the Yamada-Watanabe theorem \citep{yamada1971uniqueness}, and to guarantee that solutions do not explode in finite time. Clearly \cref{assregB2} implies \cref{assregA1}; nonetheless we include \cref{assregA1} for the purpose of clarity.

Consider the general model \cref{eqn:gen} and let $X_t := \ln S_t$ denote the log-spot. We now perturb $(X, V)$ in the following way: for $\vep \in [0,1]$,
\begin{align}
\begin{split}
	\dd X_t^{(\vep)} &= \Big ( r_t^d - r_t^f - \frac{1}{2} (V_t^{(\vep)})^2 \Big )  \dd t +V_t^{(\vep)} \dd W_t, \quad  X^{(\vep)}_0 = \ln S_0 =: x_0,  \\
	\dd V_t^{(\vep)} &= \alpha(t,V^{(\vep)}_t)\dd t + \vep \beta(t,V^{(\vep)}_t) \dd B_t, \quad V_0^{(\vep)} = v_0, \label{SDElogvep} \\
	\dd \langle W, B \rangle_t &= \rho_t\dd t.
\end{split}
\end{align}
We can recover the original diffusion from \cref{SDElogvep} as $(S, V) = (\exp ( X^{(1)}), V^{(1)})$.

Denote the filtration generated by $B$ as $(\FF_t^B)_{0\leq t \leq T}$ satisfying the usual assumptions. Let $\tilde X_t^{(\vep)} := X_t^{(\vep)} - \int_0^t (r_u^d - r_u^f)\dd u $. By writing $W_t = \int_0^t \rho_u \dd B_u + \int_0^t \sqrt{1 - \rho_u^2} \dd Z_u $,  where $Z$ is a Brownian motion independent of $B$, it can be seen that 
\begin{align*}
	\tilde X_T^{(\vep)} | \mathcal{F}_T^B \stackrel{d}{=} \mathcal{N} ( \hat \mu_\vep (T), \hat \sigma_\vep^2(T)),
\end{align*}
with 
\begin{align*}
	\hat \mu_\vep (T) &:= x_0 - \int_0^T  \frac{1}{2} (V_t^{(\vep)})^2  \dd t + \int_0^T \rho_t V_t^{(\vep)}  \dd B_t, \\
	\hat \sigma_\vep^2 (T) &:= \int_0^T   ( 1- \rho_t^2 ) (V_t^{(\vep)})^2 \dd t.
\end{align*}
Let
\begin{align*}
	g(\vep) &:= e^{-\int_0^T\! r_u^d \dd u}\, \Ex\Big [(e^k - e^{X_T^{(\vep)}})_+ \Big ].
\end{align*}
Then $g(1)$ is the price of a put option in the general model \cref{eqn:gen}. That is, $g(1) = \text{Put}_{\text{G}}$. 
\begin{proposition}
\label{mixingproposition}
The function $g$ can be expressed as
\begin{align*}
	g(\vep) &= \Ex \left [ e^{-\int_0^T\! r_u^d \dd u}\, \Ex \Big [(e^k - e^{X_T^{(\vep)}})_+ | \mathcal{F}_T^B \Big ] \right ] 
                =  \Ex \left [P_{\text{BS}} \left ( \hat \mu_\vep(T) + \frac{1}{2} \hat \sigma_\vep^2 (T), \hat \sigma_\vep^2 (T) \right ) \right ],
\end{align*}
where
\begin{align*}
	\hat \mu_\vep(T) + \frac{1}{2}  \hat \sigma_\vep^2(T) &= x_0 - \int_0^T \frac{1}{2} \rho^2_t  (V_t^{(\vep)})^2 \dd t + \int_0^T  \rho_t  V_t^{(\vep)} \dd B_t, \\
	\hat \sigma^2_\vep(T) &= \int_0^T  (1 - \rho^2_t) ( V^{(\vep)}_t)^2 \dd t,
\end{align*}
and
\begin{align}
	P_{\text{BS}}(x, y) &:= e^k e^{-\int_0^T r_t^d \dd t }  \NN ( -d^{\ln}_-) -   e^x  e^{-\int_0^T r_t^f \dd t } \NN (-d^{\ln}_+ ), \label{eqn:PBS} \\
	d^{\ln}_{\pm}:= d^{\ln}_{\pm}(x,y) &:= \frac{ x - k + \int_0^T ( r_t^d - r_t^f ) \dd t}{\sqrt{y}} \pm \frac{1}{2} \sqrt{y}, \label{eqn:dln}
\end{align}
where $\NN(\cdot)$ denotes the standard normal distribution function.
\end{proposition}
\begin{proof}
This result is a consequence of the mixing solution methodology. A derivation can be found in \Cref{appen:mixingsol}.
\end{proof}


\section{Expansion procedure}
\label{sec:expprocedure3}
\noindent In this section, we detail our expansion procedure. The notation is similar to that in \citep{timedepheston}. The main idea is to approximate the function $g$ up to second-order through a number of Taylor expansions. The expansion procedure can be briefly summarised by two main steps: 
\begin{itemize}
\item First, we expand the function $P_{\text{BS}}$ up to second-order. This step is given in \Cref{sec:expandingPBS}.
\item Then, we expand the functions $\vep \mapsto V_t^{(\vep)}$ and $\vep \mapsto \left (V_t^{(\vep)}\right)^2$ up to second-order. This step is given in \Cref{sec:expandingvee,sec:expandingPQ}.
\end{itemize}
We then combine both expansions in order to obtain a second-order approximation for the put option price, given in \Cref{thm:largeresult}. However, this approximation is still expressed in terms of expectations. In \Cref{thm:expprice}, we will obtain the explicit second-order approximation in terms of iterated integral operators, defined in \Cref{defn:integraloperator}. This explicit approximation will lead to a fast calibration scheme described in \Cref{sec:fastcal3}.
\label{expansionprocedure}
\begin{remark}
\label{tayrep}
Let $(t, \vep) \mapsto \xi_{t}^{(\vep)}$ be a $C([0,T] \times [0,1] ; \reals)$ function smooth in $\vep$. Denote by $\xi_{i,t}^{(\vep)} := \frac{ \partial^i \xi }{ \partial \vep^i}$ its $i$-th derivative in $\vep$, and let $\xi_{i,t} := \xi_{i,t}^{(\vep)}|_{\vep = 0}$. Then by a second-order Taylor expansion around $\vep = 0$, we have the representation
\begin{align*}
	\xi_t^{(\vep)} = \xi_{0,t} + \vep \xi_{1,t} + \frac{1}{2} \vep^2  \xi_{2,t} + \Theta_{2,t}^{(\vep)}(\xi)
\end{align*}
where $\Theta$ is the second-order error term given by Taylor's theorem. Specifically, for $i \geq 0$,
\begin{align*}
	 \Theta_{i,t}^{(\vep)}(\xi) := \int_0^\vep \frac{1}{i!}  (\vep - u)^i \xi_{i+1,t}^{(u)} \dd u.
\end{align*}
\end{remark}

\subsection{Expanding processes $ \vep \mapsto V_t^{(\vep)}$  and $\vep \mapsto \left (V_t^{(\vep)}\right)^2$}
\label{sec:expandingvee}
Using the notation from \Cref{tayrep}, we can now represent the functions $ \vep \mapsto V_t^{(\vep)}$  and $\vep \mapsto \left (V_t^{(\vep)}\right)^2$ via a Taylor expansion around $\vep = 0$ to second-order:
\begin{align}
\begin{split}
	V_t^{(\vep)} &= v_{0,t} + \vep V_{1,t}  + \frac{1}{2} \vep^2 V_{2,t} + \Theta^{(\vep)}_{2,t}(V),   \\
	(V_t^{(\vep)})^2 &= v_{0,t}^2 + 2  \vep v_{0,t} V_{1,t}+   \vep^2 \left ( V^2_{1,t}+ v_{0,t} V_{2,t}\right ) + \Theta^{(\vep)}_{2,t}(V^2), \label{Vtaylor}
\end{split}
\end{align}
where $v_{0,t} := V_{0,t}$.

\begin{lemma} 
\label{lem:Vderivativeprocesses}
The processes $(V_{1, t})$ and $(V_{2,t})$ satisfy the SDEs
\begin{align}
	\dd V_{1,t} &= \alpha_x (t, v_{0,t}) V_{1,t} \dd t + \beta(t, v_{0,t}) \dd B_t, \quad V_{1,0} = 0, \label{SDEV1} \\
	\dd V_{2,t} &= \Big (\alpha_{xx} (t, v_{0,t}) (V_{1,t})^2  +\alpha_x(t, v_{0,t}) V_{2,t} \Big ) \dd t + 2 \beta_x(t, v_{0,t}) V_{1,t}\dd B_t, \quad V_{2,0} = 0 \label{SDEV2},
\end{align}
with explicit solutions
\begin{align}
	V_{1,t} &= e^{\int_0^t \alpha_x (z,v_{0,z}) \dd z } \int_0^t \beta(s, v_{0,s}) e^{-\int_0^s \alpha_x (z, v_{0,z}) \dd z }\dd B_s,\label{V1}\\
	V_{2,t} &= e^{\int_0^t \alpha_x (z,v_{0,z}) \dd z } \Bigg \{  \int_0^t \alpha_{xx} (s, v_{0,s}) (V_{1,s})^2 e^{-\int_0^s \alpha_x (z, v_{0,z}) \dd z} \nonumber \dd s 
	\\&\quad + \int_0^t 2 \beta_x (s, v_{0,s}) V_{1,s} e^{-\int_0^s \alpha_x (z, v_{0,z}) \dd z }\dd B_s \Bigg \}. \label{V2}
\end{align}
\end{lemma}
\begin{proof}
We give a sketch of the proof for $(V_{1,t})$. First, we write 
\begin{align*}
	\dd V_{1,t}^{(\vep)} = \dd \left ( \partial_\vep V_t^{(\vep)} \right ) = \partial_\vep (\dd V_t^{(\vep)}).
\end{align*}
The SDE for $V_t^{(\vep)}$ is given in \cref{SDElogvep}. By differentiating, we obtain
\begin{align*}
	\dd V_{1,t}^{(\vep)} = \alpha_x (t, V_t^{(\vep)}) V_{1,t}^{(\vep)} \dd t + \left [\vep \beta_x (t, V_t^{(\vep)}) V_{1,t}^{(\vep)} + \beta(t, V_t^{(\vep)}) \right ] \dd B_t, \quad V^{(\vep)}_{1,0} = 0. 
\end{align*}
By \Cref{assregB}, $\alpha_{xx}$ and $\beta_{x}$ exist and are continuous a.e. in $x$ and $t \in [0,T]$. Then letting $\vep = 0$ yields the SDE \cref{SDEV1}. Since the SDE is linear, it can be solved explicitly. This gives the result \cref{V1}. The calculations for $(V_{2,t})$ are similar.
\end{proof}

\subsection{Expanding $P_{\textnormal{BS}}$}
\label{sec:expandingPBS}
Let 
\begin{align}
\label{eqn:PQveptilde}
\begin{split}
	\tilde P_T^{(\vep)} &:= x_0 - \int_0^T \frac{1}{2} \rho_t^2 \left (V_t^{(\vep)} \right)^2 \dd t + \int_0^T \rho_t V_t^{(\vep)}  \dd B_t,  \\
	\tilde Q_T^{(\vep)} &:=  \int_0^T (1-\rho^2_t) \left (V_t^{(\vep)}\right)^2\dd t.
\end{split}
\end{align}
Additionally, introduce the functions
\begin{align}
\label{eqn:PQvep}
\begin{split}
	P_T^{(\vep)} &:=  \tilde P_T^{(\vep)} - \tilde P_T^{(0)} \\&= \int_0^T \rho_t (V_t^{(\vep)} - v_{0, t} ) \dd B_t - \frac{1}{2} \int_0^T \rho_t^2  \left ( \left (V_t^{(\vep)}\right )^2 - v_{0,t}^2 \right ) \dd t, \\
	Q_T^{(\vep)} &:=  \tilde Q_T^{(\vep)} - \tilde Q_T^{(0)} \\&= \int_0^T (1 - \rho_t^2) \left ( \left (V_t^{(\vep)}\right )^2 - v_{0,t}^2 \right ) \dd t,
\end{split}
\end{align}
and the short-hand notations
\begin{align}
\label{eqn:PBStilde}
	\tilde P_{\text{BS}} &:= P_{\text{BS}} \left (\tilde P_T^{(0)}, \tilde Q_T^{(0)} \right), &\frac{ \partial^{i+j} \tilde P_{\text{BS}}}{\partial x^i \partial y^j } &:= \frac{ \partial^{i+j} P_{\text{BS}}\left (\tilde P_T^{(0)}, \tilde Q_T^{(0)} \right) }{\partial x^i \partial y^j }.
\end{align}
We immediately have $ \tilde P_T^{(\vep)} = \hat \mu_\vep(T) + \frac{1}{2} \hat \sigma_\vep^2(T)$ and $\tilde Q_T^{(\vep)} = \hat \sigma_\vep^2(T)$. Hence from \Cref{mixingproposition},
\begin{align}
	g(\vep) = \Ex \left [P_{\text{BS}} \left (\tilde P_T^{(\vep)}, \tilde Q_T^{(\vep)}\right )\right].
\end{align}
As $g(1)$ corresponds to the price of a put option $\text{Put}_{\text{G}}$, we are interested in approximating the function $P_{\text{BS}}$ evaluated at $ \left (\tilde P_T^{(1)}, \tilde Q_T^{(1)}\right )$.
\begin{proposition}
\label{prop:PBSsecondorder}
By a second-order Taylor expansion, the expression $P_{\text{BS}} \left (\tilde P_T^{(1)}, \tilde Q_T^{(1)} \right)$ can be approximated to second-order as
\begin{align*}
	P_{\text{BS}} \left (\tilde P_T^{(1)}, \tilde Q_T^{(1)} \right) &\approx \tilde P_{\text{BS}}  +  \left (\partial_{x} \tilde P_{\text{BS}} \right ) P_T^{(1)} + \left ( \partial_{y} \tilde P_{\text{BS}} \right ) Q_T^{(1)} \\
												&\quad +\frac{1}{2}  \left (\partial_{xx} \tilde P_{\text{BS}} \right )  \left ( P_T^{(1)}\right )^2 + \frac{1}{2} \left ( \partial_{yy} \tilde P_{\text{BS}} \right ) \left ( Q_T^{(1)} \right )^2 +  \left (\partial_{xy} \tilde P_{\text{BS}} \right ) P_T^{(1)} Q_T^{(1)}.
\end{align*}
\end{proposition}
\begin{proof}
Simply expand $P_{\text{BS}}$ around the point 
\begin{align*}
	\left (\tilde P_T^{(0)}, \tilde Q_T^{(0)} \right) = \Big (x_0 -\int_0^T \frac{1}{2} \rho^2_t v_{0,t}^2 \dd t + \int_0^T \rho_t v_{0,t} \dd B_t , \int_0^T (1-\rho_t^2) v^2_{0,t} \dd t \Big )
\end{align*}
and evaluate at $\left (\tilde P_T^{(1)}, \tilde Q_T^{(1)} \right)$.
\end{proof}

\subsection{Expanding functions $\vep \mapsto P_T^{(\vep)}$, $\vep \mapsto Q_T^{(\vep)}$ and its variants}
\label{sec:expandingPQ}
The next step in our expansion procedure is to approximate the functions $\vep \mapsto P_T^{(\vep)}, \vep \mapsto \left (P_T^{(\vep)}\right )^2, \vep \mapsto Q_T^{(\vep)},\vep \mapsto \left (Q_T^{(\vep)} \right )^2$ and $\vep \mapsto P_T^{(\vep)} Q_T^{(\vep)}$. This is summarised by the following lemma.
\begin{lemma}
\label{PQlemma}
The following second-order expansions hold:
\begin{align}
	P_T^{(\vep)} &=  \vep P_{1,T}  + \frac{1}{2} \vep^2 P_{2,T}   + \Theta^{(\vep)}_{2,T}(P), \label{eqn:PTvep}   \\
	(P_T^{(\vep)})^2 &= \vep^2 P^2_{1,T} + \Theta^{(\vep)}_{2,T}(P^2), \label{eqn:PTvep2}
\end{align}
\begin{align}
	Q_T^{(\vep)} &= \vep Q_{1,T}  + \frac{1}{2} \vep^2 Q_{2,T} + \Theta^{(\vep)}_{2,T}(Q), \label{eqn:QTvep}   \\
	(Q_T^{(\vep)})^2 &= \vep^2 Q^2_{1,T} + \Theta^{(\vep)}_{2,T}(Q^2), \label{eqn:QTvep2}
\end{align}
and 
\begin{align}
	P_T^{(\vep)} Q_T^{(\vep)} &= \vep^2 P_{1,T} Q_{1,T}+ \Theta^{(\vep)}_{2,T}(PQ), \label{eqn:PQTvep}
\end{align}
where
\begin{align*}
	 P_{1, T} &= \int_0^T \rho_t V_{1,t}  \dd B_t - \int_0^T \rho_t^2 v_{0,t} V_{1,t} \dd t,
	 &P_{2,T}  &= \int_0^T \rho_t V_{2,t}  \dd B_t - \int_0^T \rho_t^2 \left ( V_{1,t}^2 + v_{0,t} V_{2,t} \right ) \dd t,   \\ 
	 Q_{1,T} &=  2 \int_0^T(1 - \rho_t^2 ) v_{0,t} V_{1,t} \dd t,
	 &Q_{2,T} &=  2 \int_0^T (1 - \rho_t^2) \left ( V_{1, t}^2 + v_{0,t} V_{2,t} \right ) \dd t.
\end{align*}
\end{lemma}
\begin{proof}
To obtain \crefrange{eqn:PTvep}{eqn:PQTvep} we appeal to \Cref{tayrep}, noting that $P_{0,T} = P_T^{(0)} = \tilde P_T^{(0)} - \tilde P_T^{(0)} = 0$, and $Q_{0,T} = 0$. We will show how to obtain the form of $P_{1, T}$, the rest being similar. By definition
\begin{align*}
	P^{(\vep)}_{1, T} = \partial_{\vep} \left ( P_T^{(\vep)} \right ) &= \partial_{\vep} \left ( x_0 - \int_0^T \frac{1}{2} \rho_t^2 \left ( V_t^{(\vep)}\right )^2 \dd t + \int_0^T \rho_t V_t^{(\vep)} \dd B_t  \right ) \\
	&= \int_0^T \rho_t V_{1,t}^{(\vep)} \dd B_t - \int_0^T \rho_t^2  V_t^{(\vep)} V_{1,t}^{(\vep)}  \dd t.
\end{align*}
By putting $\vep = 0$ we obtain $P_{1, T}$.
\end{proof}

\begin{theorem}[Second-order put option price approximation]
\label{thm:largeresult}
Denote by $\text{Put}_{\text{G}}^{(2)}$ the second-order approximation to the price of a put option in the general model \cref{eqn:gen}. Then
\begin{align*}
&&\text{Put}_{\text{G}}^{(2)} &=  \Ex [\tilde P_{\text{BS}}] \\
(C_x :=&) 	    &	 &\quad + \Ex \Bigg [ \partial_x \tilde P_{\text{BS}} \Bigg(\int_0^T \rho_t  \left ( V_{1,t} + \frac{1}{2} V_{2,t}  \right )\dd B_t  \\
& & &\qquad - \frac{1}{2}  \int_0^T   \rho_t^2 \left ( 2v_{0,t} V_{1,t} + \left (V^2_{1,t} + v_{0,t} V_{2,t}  \right ) \right )\dd t  \Bigg ) \Bigg ] \\
(C_y :=&)	    &      &\quad + \Ex \Bigg [ \partial_y \tilde P_{\text{BS}}  \left ( \int_0^T (1 - \rho^2_t) \left ( 2v_{0,t} V_{1,t} + \left (V^2_{1,t} + v_{0,t} V_{2,t} \right )  \right )\dd t\right ) \Bigg ]\\ 
(C_{xx} :=&) &	&\quad + \frac{1}{2}  \Ex \Bigg [ \partial_{xx} \tilde P_{\text{BS}} \left ( \int_0^T \rho_t  V_{1,t} \dd B_t  - \int_0^T  \rho_t^2 v_{0,t} V_{1,t} \dd t  \right )^2 \Bigg ] \\
(C_{yy} :=&)  &	&\quad + \frac{1}{2} \Ex \Bigg [ \partial_{yy}  \tilde P_{\text{BS}} \left ( \int_0^T (1 - \rho^2_t) ( 2v_{0,t} V_{1,t})\dd t \right )^2 \Bigg ] \\
(C_{xy} :=&)  &    &\quad + \Ex \Bigg [ \partial_{xy} \tilde P_{\text{BS}}  \left ( \int_0^T \rho_t  V_{1,t} \dd B_t  - \int_0^T  \rho_t^2 v_{0,t} V_{1,t} \dd t\right )  \\ 
& &    &\qquad \cdot \left ( \int_0^T (1 - \rho^2_t) ( 2v_{0,t} V_{1,t}) \dd t\right ) \Bigg ].
\end{align*}
Additionally, $\text{Put}_{\text{G}} = \text{Put}_{\text{G}}^{(2)} + \Ex [\mathcal{E}]$, where $\mathcal{E}$ denotes the expansion error. 
\end{theorem}
\begin{proof}
From \Cref{prop:PBSsecondorder}, consider the two-dimensional Taylor expansion of $P_{\text{BS}}$ around $\left (\tilde P_T^{(0)}, \tilde Q_T^{(0)} \right)$ evaluated at $\left (\tilde P_T^{(1)}, \tilde Q_T^{(1)} \right)$.
Then, substitute in the second-order expressions of $P_T^{(1)}, \left (P_T^{(1)}\right )^2, Q_T^{(1)}, \left ( Q_T^{(1)} \right )^2$ and $P_T^{(1)} Q_T^{(1)}$ from \Cref{PQlemma}. As this is a second-order expression, the remainder terms $\Theta$ are neglected. Taking expectation yields $\text{Put}_{\text{G}}^{(2)}$.
\end{proof}
\begin{remark}
The explicit expression for $\mathcal{E}$ and its analysis is left for \Cref{sec:erroranalysis3}.
\end{remark}

\section{Explicit price}
\noindent The goal is now to express $C_x, C_y, C_{xx}, C_{yy}, C_{xy}$ in terms of the following integral operators.
\label{sec:explicitprice3}
\begin{definition}[Integral operator]
\label{defn:integraloperator}
Let $k$ and $\ell$ be real-valued functions defined on $[0, T]$ such that the following integral operator
\begin{align*}
	\omega_{t, T}^{(k, \ell )} := \int_t^T \ell_u e^{\int_0^u k_z \dd z } \dd u
\end{align*}
exists. In addition, for sequences of real valued functions $(k^{(n)})_n$ and $(\ell^{(n)})_n$ defined on $[0, T]$, we define the $n$-fold iterated integral operator through the following recurrence:
\begin{align*}
	\omega_{t,T}^{(k^{(n)}, \ell^{(n)}), (k^{(n-1)}, \ell^{(n-1)}), \dots , (k^{(1)}, \ell^{(1)})}  := \omega_{t,T}^{\big (k^{(n)}, \ell^{(n)} w_{\cdot,T} ^{(k^{(n-1)}, \ell^{(n-1)}) ,\dots, (k^{(1)}, \ell^{(1)})}\big ) }, \quad n \in \mathbb{N}, 
\end{align*}
whenever it exists.
\end{definition}

\begin{assumption}
\label{assumpbeta}
$\beta(t,x) = \lambda_t x^{\mu}$ for $\mu \in [1/2, 1]$, where $\lambda$ is bounded over $[0, T]$.
\end{assumption}
We will comment on the reasoning behind \Cref{assumpbeta} in \Cref{rem:assumpbeta}.

\noindent The main result of this article is the following theorem.

\begin{theorem}[Explicit second-order put option price]
\label{thm:expprice}
Enforcing \Cref{assumpbeta}, the explicit second-order price of a put option in the general model \cref{eqn:gen} is given by
\begin{align*}
\text{Put}^{(2)}_{\text{G}} &= P_{\text{BS}} \left ( x_0, \int_0^T v_{0,t}^2  \dd t \right )  \\
&\quad + 2 \omega_{0,T} ^{(- \alpha_x, \rho \lambda v_{0, \cdot}^{\mu + 1}), (\alpha_x, v_{0, \cdot}) } \partial_{xy} P_{\text{BS}} \left ( x_0, \int_0^T v_{0,t}^2  \dd t \right ) \\
&\quad + \omega_{0, T}^{(-2 \alpha_x, \lambda^2 v_{0, \cdot}^{2\mu } ), ( 2\alpha_x, 1)} \partial_{y} P_{\text{BS}} \left ( x_0, \int_0^T v_{0,t}^2  \dd t \right ) \\
&\quad + 2\omega_{0,T}^{(-\alpha_x, \rho \lambda v_{0, \cdot }^{\mu + 1}), (- \alpha_x, \rho \lambda v_{0, \cdot}^{\mu + 1}), (2 \alpha_x,1 ) } \partial_{xxy} P_{\text{BS}} \left ( x_0, \int_0^T v_{0,t}^2  \dd t \right )  \\
&\quad + \omega_{0,T}^{(- 2\alpha_x, \lambda^2 v_{0, \cdot}^{2 \mu}),( \alpha_x, \alpha_{xx}) , (\alpha_x, v_{0, \cdot} ) }  \partial_y P_{\text{BS}} \left ( x_0, \int_0^T v_{0,t}^2  \dd t \right ) \\
&\quad + \Bigg \{ 2\omega_{0,T}^{(-\alpha_x, \rho \lambda v_{0, \cdot}^{\mu + 1}), ( -\alpha_x, \rho \lambda v_{0, \cdot}^{\mu + 1} ) , ( \alpha_x, \alpha_{xx}), (\alpha_x, v_{0, \cdot}) } \\&\qquad + 2\mu \omega_{0,T}^{(-\alpha_x, \rho \lambda v_{0, \cdot}^{\mu + 1} ), (0, \rho \lambda v_{0, \cdot}^{2 \mu - 1 }), (\alpha_x, v_{0,\cdot} ) } \Bigg \}  \partial_{xxy} P_{\text{BS}} \left ( x_0, \int_0^T v_{0,t}^2  \dd t \right ) \\
&\quad +2 \omega_{0,T}^{(- \alpha_x, \rho \lambda v_{0,\cdot}^{\mu + 1}), (0, \rho \lambda v_{0,\cdot}^\mu), (\alpha_x, v_{0, \cdot }) } \partial_{xxy} P_{\text{BS}} \left ( x_0, \int_0^T v_{0,t}^2  \dd t \right ) \\
&\quad + 4 \omega_{0,T}^{(- 2 \alpha_x, \lambda^2 v_{0,\cdot}^{2\mu}),(\alpha_x, v_{0,\cdot}), (\alpha_x, v_{0,\cdot}) }  \partial_{yy} P_{\text{BS}} \left ( x_0, \int_0^T v_{0,t}^2  \dd t \right ) \\
&\quad +2\left ( \omega_{0,T}^{(-\alpha_x, \rho \lambda v_{0, \cdot}^{\mu + 1}),(\alpha_x, v_{0,\cdot}) } \right )^2   \partial_{xxyy} P_{\text{BS}} \left ( x_0, \int_0^T v_{0,t}^2  \dd t \right )
\end{align*}
where the partial derivatives of $P_{\text{BS}}$ are given in \Cref{appen:partialderivativesPBS}.
\end{theorem}
\begin{proof}
The proof is given in \Cref{append:expprice}.
\end{proof}

\begin{remark}
In \Cref{thm:expprice}, we enforce \Cref{assumpbeta}. This means we can obtain a second-order pricing formula for various stochastic volatility models by setting $\mu \in [1/2,1]$, choosing a specific $\alpha(t,x)$ satisfying \Cref{assregA} and \Cref{assregB}, and applying \Cref{thm:expprice}. For instance, if we set $\alpha(t,x) = \kappa_t(\theta_t - x)$, then this drift satisfies both \Cref{assregA} and \Cref{assregB}. By choosing some $\mu \in [1/2,1]$, we then obtain an explicit second-order price approximation of a European put option when the stochastic volatility obeys the dynamics
\begin{align*}
	\dd V_t = \kappa_t(\theta_t - V_t)\dd t + \lambda_t V_t^{\mu} \dd B_t, \quad V_0 = v_0.
\end{align*}
In particular, to obtain an explicit second-order put option price in the Inverse-Gamma model, choose $\alpha(t,x) = \kappa_t(\theta_t - x)$ and $\mu = 1$, so that $\alpha_x(t,x) = - \kappa_t$ and $\alpha_{xx}(t, x) = 0$. Indeed, this gives the desired result for the second-order put option price in the Inverse-Gamma model as seen in \citep{IGa}. Additionally, by taking $\kappa_t \equiv 0$ and $\mu = 1$, we can recover the classical SABR model with skewness parameter 1. We tackle this model in \Cref{appen:sabr}, to demonstrate that our closed-form approximation formula works well in this familiar model.
\end{remark}

\begin{remark}
The second-order approximation $\text{Put}^{(2)}_{\text{G}}$ from \Cref{thm:expprice} is expressed in terms of iterated integral operators from \Cref{defn:integraloperator} and partial derivatives of $P_{\text{BS}}$. \Cref{sec:fastcal3} shows how to express these iterated integral operators in a closed-form manner when parameters are assumed to be piecewise-constant.
\end{remark}

\section{Error analysis}
\label{sec:erroranalysis3}
\noindent This section is dedicated to the explicit representation and analysis of the error induced by our expansion procedure in \Cref{sec:expprocedure3}. The section is divided into two parts:
\begin{itemize}
\item \Cref{sec:expliciterror3} provides an explicit representation of the error term induced by the expansion procedure.
\item \Cref{sec:boundverhulsterror} provides a bound on the error term induced by the expansion procedure in terms of the remainder terms generated by the approximation of the underlying volatility process.
\end{itemize}
In the following, $L^p := L^p(\Omega, \FF, \Qro)$ denotes the vector space of random variables (identified $\Qro$ a.s.) with finite $L^p$ norm, given by $\| \cdot \|_p = \left [\Ex| \cdot |^p \right ]^{1/p}$. Moreover, for an $n$-tuple $\alpha := (\alpha_1, \dots , \alpha_n) \in \naturals^n$, then $| \alpha| = \sum_{i = 1}^n \alpha_i$ denotes its $1$-norm.

\subsection{Explicit expression for error}
\label{sec:expliciterror3}
Recall from \Cref{thm:largeresult} that the price of a put option in the general model \cref{eqn:gen} is $\text{Put}_{\text{G}} = \text{Put}^{(2)}_{\text{G}} + \Ex [\mathcal{E}]$, where $\text{Put}_{\text{G}}^{(2)}$ is the second-order closed-form price. The error term $\mathcal{E}$ is generated by the truncation of a Taylor series. It can be expressed explicitly, as recalled in \Cref{thm:tay2} below for the second-order error.

\begin{theorem}[Taylor's theorem]
\label{thm:tay2}
Let $A \subseteq \reals^2$ and $g : A \to \reals$ be a $C^3$ function in a closed ball about the point $(a, b) \in A$. Then
\begin{align*}
	g(x, y) &= g(a, b) + g_x (a, b)(x-a) + g_y(a, b) (y-b) 
	\\& \quad + \frac{1}{2} g_{xx} (a,b) (x-a)^2 + \frac{1}{2} g_{yy}(a, b) (y-b)^2  + g_{xy} (a, b)(x-a)(y-b) + R(x,y),
\end{align*}
where
\begin{align*}
	R(x,y) &= \sum_{|\alpha| = 3} \frac{|\alpha|}{\alpha_1! \alpha_2 ! }E_{\alpha}(x, y) ( x- a)^{\alpha_1} (y - b)^{\alpha_2}, \\
	E_\alpha (x, y) &= \int_0^1 ( 1 - u)^{2} \frac{\partial^{3} }{\partial x^{\alpha_1} \partial y^{\alpha_2}} g(a + u(x - a), b  + u(y - b)) \dd u.
\end{align*}
\end{theorem}
Now recall from \Cref{expansionprocedure} the functions $\tilde P_T^{(\vep)}, \tilde Q_T^{(\vep)}, P_T^{(\vep)}, Q_T^{(\vep)}, \tilde P_{\text{BS}}, \frac{ \partial^{i+j} \tilde P_{\text{BS}}}{\partial x^i \partial y^j }$ given in \crefrange{eqn:PQveptilde}{eqn:PBStilde}.

\begin{theorem}[Explicit error term]
\label{expliciterrorterm}
The error term $\EE$ in \Cref{thm:largeresult} induced from the expansion procedure can be decomposed as
\begin{align*}
	\mathcal{E} = \mathcal{E}_P + \mathcal{E}_V
\end{align*}
where
\begin{align*}
	\mathcal{E}_P &= \sum_{|\alpha| = 3} \frac{|\alpha| }{\alpha_1! \alpha_2!} E_{\alpha}\! \left (\tilde P_T^{(1)}, \tilde Q_T^{(1)} \right)  \left (P_T^{(1)} \right )^{\alpha_1} \left (Q_T^{(1)} \right )^{\alpha_2}, \\
	E_{\alpha}\! \left (\tilde P_T^{(1)}, \tilde Q_T^{(1)} \right) &= \int_0^1 (1 - u)^2 \frac{\partial^{\alpha} P_{\text{BS}}}{\partial x^{\alpha_1}\partial y^{\alpha_2}}\left ( (1-u) \tilde P_T^{(0)} + u \tilde P_T^{(1)}, (1-u) \tilde Q_T^{(0)} + u \tilde Q_T^{(1)} \right ) \dd u,
\end{align*}
and 
\begin{align*}
	\EE_V = \sum_{|\alpha| = 1} \frac{\partial \tilde P_{\text{BS}}}{\partial x^{\alpha_1} \partial y^{\alpha_2}}  \Theta_{2, T}^{(1)}(P^{\alpha_1} Q^{\alpha_2})  + \frac{1}{2}  \sum_{|\alpha| = 2 } \frac{|\alpha|}{\alpha_1!\alpha_2!} \frac{\partial^{2} \tilde P_{\text{BS}}}{\partial x^{\alpha_1} \partial y^{\alpha_2}}  \Theta_{2, T}^{(1)}(P^{\alpha_1} Q^{\alpha_2}),
\end{align*}
where $\Theta_{2, T}^{(1)}(P^0 Q^x) := \Theta_{2, T}^{(1)}( Q^x) $ and $\Theta_{2, T}^{(1)}(P^y Q^0):= \Theta_{2, T}^{(1)}(P^y)$, for $x, y \in \{1, 2 \}$. Here, $\mathcal{E}_P$ corresponds to the error in the approximation of the function $P_{\text{BS}}$, and $\mathcal{E}_V$ corresponds to the error in the approximation of the functions $\vep \mapsto V_t^{(\vep)}$ and $\vep \mapsto \left (V_t^{(\vep)} \right )^2$. 
\end{theorem}
\begin{proof} 
The decomposition $\EE = \EE_P + \EE_V$ is a clear consequence of Taylor's theorem. The next two subsections are dedicated to representing $\EE_P$ and $\EE_V$ explicitly.

\subsubsection{Explicit $\mathcal{E}_P$}
First we explicitly derive $\mathcal{E}_P$, the error term corresponding to the second-order approximation of $P_{\text{BS}}$. In our expansion procedure, we expand $P_{\text{BS}}$ up to second-order around the point 
\begin{align*}
	\left (\tilde P_T^{(0)}, \tilde Q_T^{(0)} \right) = \Big (x_0 -\int_0^T \frac{1}{2} \rho^2_t v_{0,t}^2 \dd t + \int_0^T \rho_t v_{0,t} \dd B_t , \int_0^T (1-\rho_t^2) v^2_{0,t} \dd t \Big )
\end{align*}
and evaluate at $\left (\tilde P_T^{(1)}, \tilde Q_T^{(1)} \right)$. Thus in the Taylor expansion of $P_{\text{BS}}$, the terms will be of the form
\begin{align*}
	\frac{\partial^{|\alpha|} \tilde P_{\text{BS}}}{\partial x^{\alpha_1}\partial y^{\alpha_2}} \left (P_T^{(1)} \right )^{\alpha_1} \left (Q_T^{(1)} \right )^{\alpha_2}
\end{align*}
for $|\alpha| = 0, 1, 2$. 
By \Cref{thm:tay2} (Taylor's theorem) we can write the second-order Taylor polynomial of $P_{\text{BS}} \left (\tilde P_T^{(1)}, \tilde Q_T^{(1)} \right)$ with error term as 
\begin{align}
\begin{split}
	P_{\text{BS}} \left (\tilde P_T^{(1)}, \tilde Q_T^{(1)} \right) &= \tilde P_{\text{BS}}  +  \left (\partial_{x} \tilde P_{\text{BS}} \right ) P_T^{(1)} + \left ( \partial_{y} \tilde P_{\text{BS}} \right ) Q_T^{(1)} \\
												&\quad +\frac{1}{2}  \left (\partial_{xx} \tilde P_{\text{BS}} \right )  \left ( P_T^{(1)}\right )^2 + \frac{1}{2} \left ( \partial_{yy} \tilde P_{\text{BS}} \right ) \left ( Q_T^{(1)} \right )^2 +  \left (\partial_{xy} \tilde P_{\text{BS}} \right ) P_T^{(1)} Q_T^{(1)} \\
												&\quad + \underbrace{\sum_{|\alpha| = 3} \frac{|\alpha| }{\alpha_1! \alpha_2!} E_{\alpha}\! \left (\tilde P_T^{(1)}, \tilde Q_T^{(1)} \right)  \left (P_T^{(1)} \right )^{\alpha_1} \left (Q_T^{(1)} \right )^{\alpha_2}}_{\text{Error term}} \label{PBSexpwitherror}
\end{split}												
\end{align}
with
\begin{align*}
				E_{\alpha}\! \left (\tilde P_T^{(1)}, \tilde Q_T^{(1)} \right) &= \int_0^1 (1 - u)^2 \frac{\partial^{3} P_{\text{BS}}}{\partial x^{\alpha_1}\partial y^{\alpha_2}}\left ( \tilde P_T^{(0)} + u P_T^{(1)}, \tilde Q_T^{(0)} + u Q_T^{(1)} \right ) \dd u \\
				&= \int_0^1 (1 - u)^2 \frac{\partial^{3} P_{\text{BS}}}{\partial x^{\alpha_1}\partial y^{\alpha_2}}\left ( (1-u) \tilde P_T^{(0)} + u \tilde P_T^{(1)}, (1-u) \tilde Q_T^{(0)} + u \tilde Q_T^{(1)} \right ) \dd u.
\end{align*}
Taking expectation gives $\text{Put}_{\text{G}}$. Thus the explicit form for the error term $\mathcal{E}_P$ is 
\begin{align*}
	\mathcal{E}_P = \sum_{|\alpha| = 3} \frac{|\alpha| }{\alpha_1! \alpha_2!} E_{\alpha}\! \left (\tilde P_T^{(1)}, \tilde Q_T^{(1)} \right)  \left (P_T^{(1)} \right )^{\alpha_1} \left (Q_T^{(1)} \right )^{\alpha_2}.
\end{align*}

\subsubsection{Explicit $\EE_V$}
Now we explicitly derive $\EE_V$, the error corresponding to the second-order approximation of the functions  $\vep \mapsto V_t^{(\vep)}$ and $\vep \mapsto \left (V_t^{(\vep)} \right )^2$. The idea is to approximate the functions $\vep \mapsto P_T^{(\vep)}$, $\vep \mapsto Q_T^{(\vep)}$ and their variants by their second-order expansions given in \Cref{PQlemma}. For example, in the case of the first-order derivative of $P_{\text{BS}}$ in its second argument in \cref{PBSexpwitherror}, we have
\begin{align*}
	( \partial_y \tilde P_{\text{BS}} ) Q_T^{(1)} &= ( \partial_y \tilde P_{\text{BS}} ) ( Q_{1, T} + \frac{1}{2} Q_{2, T} + \Theta_{2, T}^{(1)}(Q) ) \\
								&=  ( \partial_y \tilde P_{\text{BS}} ) (Q_{1, T} + \frac{1}{2} Q_{2, T}) + \underbrace{( \partial_y \tilde P_{\text{BS}} )( \Theta_{2, T}^{(1)}(Q) }_{\text{Error term}}).
\end{align*} 
For the term corresponding to the second-order derivative of $P_{\text{BS}}$ in its second argument, we would have 
\begin{align*}
	\frac{1}{2} (\partial_{yy}  \tilde P_{\text{BS}} ) \left (Q_T^{(1)} \right )^2 &= \frac{1}{2} ( \partial_{yy} \tilde P_{\text{BS}} ) \left ( Q_{1,T}^2 +  \Theta_{2, T}^{(1)}(Q^2) \right ) \\
	&= \frac{1}{2}  ( \partial_{yy} \tilde P_{\text{BS}} ) \left (Q_{1,T}^2 \right )  + \underbrace{ \frac{1}{2} ( \partial_{yy} \tilde P_{\text{BS}} ) \left (\Theta_{2, T}^{(1)}(Q^2) \right )}_{\text{Error term}}.
\end{align*}
Following this pattern, we can see that the error term $\EE_V$ can be written explicitly as 
\begin{align*}
	\EE_V = \sum_{|\alpha| = 1} \frac{\partial \tilde P_{\text{BS}}}{\partial x^{\alpha_1} y^{\alpha_2}}  \Theta_{2, T}^{(1)}(P^{\alpha_1} Q^{\alpha_2})  + \frac{1}{2}  \sum_{|\alpha| = 2 } \frac{|\alpha|}{\alpha_1!\alpha_2!} \frac{\partial^2 \tilde P_{\text{BS}}}{\partial x^{\alpha_1} y^{\alpha_2}}  \Theta_{2, T}^{(1)}(P^{\alpha_1} Q^{\alpha_2}).
\end{align*}
\end{proof} 
As the second-order price of a put option is the expectation of our expansion, the goal is to bound $\EE$ in $L^1$ for a generic volatility process $V$.

\subsection{Bounding error term}
\label{sec:boundverhulsterror}
Our objective is to obtain a bound in $L^1$ of the explicit representation of the error term $\EE$ from \Cref{expliciterrorterm} under the general stochastic volatility model \cref{eqn:gen}. As we are working in a general framework, the resulting bound will be expressed in terms of expectations of functionals of the volatility process $V$.
\begin{proposition}
\label{prop:ingredientsbound}
One can bound $\EE$ in $L^1$ in terms of the following quantities:
\begin{enumerate}[label = (\roman*), ref = \roman*]
\item Bounds on $\| \Theta^{(1)}_{2,T}(P^{\alpha_1} Q^{\alpha_2}) \|$, where $|\alpha| = 1,2$. \label{item1ingredients}
\item Bounds on $\|P_T^{(1)}\|_p$ and $\|Q_T^{(1)}\|_p$ for $ p \geq 4$. \label{item2ingredients} 
\end{enumerate}
\end{proposition}

The purpose of the next part of this section is to prove \Cref{prop:ingredientsbound}. This requires a number of technical results.
\begin{lemma}[{\citet[Lemma 5.2]{das2018closedform}}]
\label{lem:parblowu}
Define
\begin{align}
	\text{Put}_{\text{BS}}(x, y) &:= K e^{-\int_0^T \!r_t^d \dd t } \NN (-d_- ) - x e^{-\int_0^T \!r_t^f \dd t } \NN ( -d_+), \label{eqn:PutBS}    \\
	d_{\pm}(x,y) := d_{\pm} &:= \frac{ \ln(x/K) + \int_0^T ( r_t^d - r_t^f ) \dd t}{\sqrt{y}} \pm \frac{1}{2} \sqrt{y}. \label{eqn:d}
\end{align}
Consider the third-order partial derivatives of $\text{Put}_{\text{BS}}$, $\frac{\partial^3 \text{Put}_{\text{BS}}}{\partial x^{\alpha_1} \partial y^{\alpha_2}}$, where $\alpha_1 + \alpha_2 = 3$ as well as the linear functions $h_1, h_2 :[0,1] \to \reals_+$ such that $h_1(u) = u(d_1 - c_1) + c_1$ and $h_2(u) = u(d_2 - c_2) + c_2$. Assume there exists no point $a \in (0,1)$ such that 
\begin{align*}
	\lim_{u \to a} \frac{ \ln(h_1(u)/K) + \int_0^T (r_t^d - r_t^f) \dd t  }{\sqrt {h_2(u)}} = 0  \quad \text{and} \quad \lim_{u \to a} h_2(u) = 0.
\end{align*} 
Then there exist functions $M_{\alpha}$ bounded on $\reals_+^2$ such that
\begin{align*}
	\sup_{u \in (0,1)} \left | \frac{\partial^3 \text{Put}_{\text{BS}}}{\partial x^{\alpha_1} \partial y^{\alpha_2}} (h_1(u), h_2(u)) \right | =M_{\alpha}(T, K).
\end{align*}
Furthermore, the behaviour of $M_{\alpha}$ for fixed $K$ and $T$ is characterised by the functions $\zeta$ and $\eta$ respectively, where
\begin{align*}
	\zeta(T) =  \hat A  e^{-\int_0^T r_t^f \dd t } e^{ -E_2 \tilde r^2(T) } e^{- E_1 \tilde r(T)} \sum_{i=0}^n c_i \tilde r^i (T),
\end{align*}
with $\tilde r (T) :=  \int_0^T (r_t^d - r_t^f) \dd t$ and $E_2 > 0$, $E_1 \in \reals, \hat A \in \reals$, $n \in \mathbb{N}$ and $c_0, \dots, c_n$ are constants, and
\begin{align*}
	\eta(K) =\tilde A K^{-D_2 \ln(K) + D_1} \sum_{i = 0}^N C_i (-1)^i \ln^i(K), 
\end{align*}
with $D_2 > 0, D_1 \in \reals, \tilde A \in \reals, N \in \mathbb{N}$ and $C_0, \dots, C_N$ are constants.
\end{lemma}

\begin{lemma}
\label{lem:parblowupart3}
Consider the third-order partial derivatives of $P_{\text{BS}}$, $\frac{\partial^3 P_{\text{BS}}}{\partial x^{\alpha_1} \partial y^{\alpha_2}}$, where $\alpha_1 + \alpha_2 = 3$ as well as the linear functions $h_1, h_2 :[0,1] \to \reals_+$ such that $h_1(u) = u(d_1 - c_1) + c_1$ and $h_2(u) = u(d_2 - c_2) + c_2$. Assume there exists no point $a \in (0,1)$ such that 
\begin{align}
	\lim_{u \to a} \frac{ h_1(u) - k + \int_0^T (r_t^d - r_t^f) \dd t  }{\sqrt {h_2(u)}} = 0  \quad \text{and} \quad \lim_{u \to a} h_2(u) = 0. \label{asslimitpart3}
\end{align} 
Then there exist functions $B_{\alpha}$ bounded on $\reals_+ \times \reals$ such that
\begin{align*}
	\sup_{u \in (0,1)} \left | \frac{\partial^3 P_{\text{BS}}}{\partial x^{\alpha_1} \partial y^{\alpha_2}} (h_1(u), h_2(u)) \right | =B_{\alpha}(T, k).
\end{align*}
Furthermore, the behaviour of $B_{\alpha}$ for fixed $k$ and $T$ is characterised by the functions $\zeta$ and $\nu$ respectively, where
\begin{align*}
	\zeta(T) =  \hat A  e^{-\int_0^T r_t^f \dd t } e^{ -E_2 \tilde r^2(T) } e^{- E_1 \tilde r(T)} \sum_{i=0}^n c_i \tilde r^i (T),
\end{align*}
with $\tilde r (T) :=  \int_0^T (r_t^d - r_t^f) \dd t$ and $E_2 > 0$, $E_1 \in \reals, \hat A \in \reals$, $n \in \mathbb{N}$ and $c_0, \dots, c_n$ are constants, and
\begin{align*}
	\nu(k) =\tilde A e^{-D_2 k^2 + D_1 k} \sum_{i = 0}^N C_i (-1)^i k^i, 
\end{align*}
with $D_2 > 0, D_1 \in \reals, \tilde A \in \reals, N \in \mathbb{N}$ and $C_0, \dots, C_N$ are constants.
\end{lemma}
\begin{proof}
\Cref{lem:parblowupart3} is an adaptation of \Cref{lem:parblowu} where the strike $K$ has been replaced by the log-strike $k$ (function $P_{\text{BS}}$ \eqref{eqn:PBS} instead of function Put$_{\text{BS}}$ \eqref{eqn:PutBS}). In fact, we will show that \Cref{lem:parblowu} implies \Cref{lem:parblowupart3}. In the following proof, $F$ and $G$ denote arbitrary polynomials of some degree, and $A$ is an arbitrary constant. That is, they may be different on each use.

Recall functions $d_{\pm}$ and $d^{\ln}_{\pm}$ given by \cref{eqn:d} and \cref{eqn:dln} respectively, and denote by $\phi$ the standard normal density function. First, as a function of $x$ and $y$, notice from \Cref{appen:partialderivativesPBS} that the third-order partial derivatives $\frac{\partial^3 P_{\text{BS}}}{\partial x^{\alpha_1} \partial y^{\alpha_2}}$ where $\alpha_1 + \alpha_2 = 3$ can be written as 
\begin{align}
	A \frac{e^x \phi(d_+^{\ln})}{y^{m/2} } G(d_+^{\ln}, d_-^{\ln}, \sqrt{y}), \quad m \in \naturals  \label{PBSthirdorderpart3}
\end{align}
except when $\alpha = (3, 0)$, in which case the partial derivative can be written as
\begin{align}
	A \frac{e^x \phi(d_+^{\ln})}{y^{m/2} } G(d_+^{\ln}, d_-^{\ln}, \sqrt{y}) + \overline A e^x \phi(d_+^{\ln}) (\NN(d_+^{\ln}) - 1), \quad m \in \naturals.
\end{align}
Similarly, it can seen that as a function of $x$ and $y$, the third-order partial derivatives of $\text{Put}_{\text{BS}}$ can be written as
\begin{align}
	A \frac{\phi(d_+)}{x^n y^{m/2}}F(d_+, d_-, \sqrt{y}),  \quad n \in \mathbb{Z}, m \in \mathbb{N}. \label{PutBSthirdorderpart3}
\end{align}
Let us consider the cases for which $\alpha \neq (3, 0)$. Without loss of generality, set $k = \ln(K)$. Notice that $d_{\pm}^{\ln}(x, y) = d_{\pm}(e^x, y)$. Take $n = -1$ in \cref{PutBSthirdorderpart3}. Roughly speaking, we will say that two functions $f$ and $g$ are `of the same form', denoted by $f \stackrel{C}{\sim} g$, if they are equal up to constant values. Then, one can see from \cref{PBSthirdorderpart3,PutBSthirdorderpart3} that the partial derivatives of $P_{\text{BS}}$ are of the same form as the partial derivatives of $\text{Put}_{\text{BS}}$ composed with the function $e^{x}$ in its first argument:
\begin{align*}
	 \frac{\partial^3 P_{\text{BS}}}{\partial x^{\alpha_1} \partial y^{\alpha_2}}(x,y) \stackrel{C}{\sim} \frac{\partial^3 \text{Put}_{\text{BS}}}{\partial x^{\alpha_1} \partial y^{\alpha_2}}(e^x,y).
\end{align*}
Now, consider arbitrary functions $f,b: \reals \to \reals$ such that 
\begin{align*}
	\sup_{x \in \reals} |f(x)| = L < \infty.
\end{align*}
Then it is true that
\begin{align*}
	\sup_{x \in \reals} |f(b(x))| = \tilde L \leq L < \infty.
\end{align*}
Thus 
\begin{align*}
	 \sup_{u \in (0, 1)} \left | \frac{\partial^3 P_{\text{BS}}}{\partial x^{\alpha_1} \partial y^{\alpha_2}}(h_1(u), h_2(u)) \right | \stackrel{C}{\sim} \sup_{u \in (0, 1)} \left | \frac{\partial^3 \text{Put}_{\text{BS}}}{\partial x^{\alpha_1} \partial y^{\alpha_2}}(e^{h_1(u)},h_2(u)) \right |.
\end{align*}
Under the assumption in \cref{asslimitpart3}, and then using \Cref{lem:parblowu}, this supremum will not blow up. Clearly, $\sup_{u \in (0, 1)} \frac{\partial^3 \text{Put}_{\text{BS}}}{\partial x^{\alpha_1} \partial y^{\alpha_2}}(e^{h_1(u)},h_2(u))$ is a function of $T$ and $K$. By substituting $k = \ln(K)$ in the result of \Cref{lem:parblowu}, we obtain the form of $\zeta$ and $\nu$.

Now for the case $\alpha = (3, 0)$, we have that 
\begin{align*}
	 \frac{\partial^3 P_{\text{BS}}}{\partial x^{3}}(x,y) \stackrel{C}{\sim} \frac{\partial^3 \text{Put}_{\text{BS}}}{\partial x^{3}}(e^x,y) + A \underbrace {e^x \phi(d_+^{\ln}) (\NN(d_+^{\ln}) - 1)}_{=:H(x,y)}.
\end{align*}
Now 
\begin{align*}
	|H(x,y)| = |e^x \phi(d_+^{\ln}) (\NN(d_+^{\ln}) - 1)| \leq e^x \phi(d_+^{\ln}). 
\end{align*}
Thus
\begin{align*}
	 \sup_{x \in \reals} |H(x,y)| =  \sup_{x \in \reals} |e^x \phi(d_+^{\ln}) (\NN(d_+^{\ln}) - 1)| \leq \sup_{x \in \reals} e^x \phi(d_+^{\ln}) < \infty
\end{align*}
and also 
\begin{align*}
	 \sup_{y \in \reals_+} |H(x,y)|  = \sup_{y \in \reals_+} | \phi(d_+^{\ln}) (\NN(d_+^{\ln}) - 1)| \leq  \sup_{y \in \reals_+} \phi(d_+^{\ln}) < \infty.
\end{align*}
Hence
\begin{align*}
	\sup_{u \in (0, 1)} H(h_1(u),h_2(u)) \leq \sup_{u \in (0, 1)} e^{h_1(u)} \phi(d^{\ln}_+(h_1(u),h_2(u)))  =  \hat m(T, k),
\end{align*}
where $\hat m$ is a bounded function on $\reals_+ \times \reals$. By direct computation, it is clear that, for fixed $T$, the form of $\hat m$ is given by
\begin{align*}
	A e^{-\hat D_2 k^2} e^{\hat D_1 k},
\end{align*}
where $\hat D_2 > 0$ and $\hat D_1 \in \reals$. For fixed $k$, it is given by 
\begin{align*}
	A e^{-\hat E_2 \tilde r^2(T)}e^{\hat E_1 \tilde r (T)},
\end{align*}
where $\hat E_2 > 0$ and $\hat E_1 \in \reals$. Thus
\begin{align*}
	 \sup_{u \in (0, 1)} \left | \frac{\partial^3 P_{\text{BS}}}{\partial x^3}(h_1(u), h_2(u)) \right | & \stackrel{C}{\sim}  \sup_{u \in (0,1)}\left  |\frac{\partial^3 \text{Put}_{\text{BS}}}{\partial x^{3}}(e^{h_1(u)},h_2(u)) + A H(h_1(u),h_2(u)) \right | \\
	 &\leq  \sup_{u \in (0,1)} \left |\frac{\partial^3 \text{Put}_{\text{BS}}}{\partial x^{3}}(e^{h_1(u)},h_2(u)) \right |+ A \sup_{u \in (0,1)} \left |H(h_1(u),h_2(u)) \right | \\
	 &\stackrel{C}{\sim} B_{(3,0)}(T,k) + A \hat m(T, k).
\end{align*}
Since the form of $\hat m$ is exactly that of $B_{\alpha}$ without the polynomial expression, the sum of them is again of the form of $B_{\alpha}$.
\end{proof}

\subsubsection{Bounding $\EE_V$ in $L^1$}
\label{sec:boundev}
We first bound the term $\EE_V$ from \Cref{expliciterrorterm} in $L^1$. Analysing $\EE_V$, we note that the terms of interest to bound are 
\begin{align*}
	 \frac{\partial^{|\alpha|} \tilde P_{\text{BS}}}{\partial x^{\alpha_1} \partial y^{\alpha_2}}  \Theta_{2, T}^{(1)}(P^{\alpha_1} Q^{\alpha_2}), \quad |\alpha| = 1, 2.
\end{align*}
The second argument of $\tilde P_{\text{BS}}$ is $\tilde Q_T^{(0)} $, which is strictly positive. Then, apply \Cref{lem:parblowupart3} with the trivial linear function $u \mapsto (1-u) Q_T^{(0)} + uQ_T^{(0)}$ to obtain $\frac{\partial^{|\alpha|} \tilde P_{\text{BS}}}{\partial x^{\alpha_1} \partial y^{\alpha_2}} \leq B_{\alpha}(T,k)$. Thus
\begin{align}
	\left \|  \frac{\partial^{|\alpha|} \tilde P_{\text{BS}}}{\partial x^{\alpha_1} \partial y^{\alpha_2}}  \Theta_{2, T}^{(1)}(P^{\alpha_1} Q^{\alpha_2})  \right \|  \leq B_{\alpha}(T, k) \left \| \Theta_{2, T}^{(1)} (P^{\alpha_1} Q^{\alpha_2})  \right \|. \label{cauchy1}
\end{align}
\Cref{cauchy1} suggests that obtaining an $L^1$ bound on the remainder term $\Theta_{2, T}^{(1)} (P^{\alpha_1} Q^{\alpha_2})$ for $|\alpha| = 1, 2$ is sufficient. This validates \cref{item1ingredients} in \Cref{prop:ingredientsbound}.

\subsubsection{Bounding $\EE_P$ in $L^1$}
\label{sec:boundep}
The terms of interest are
\begin{align*}
	 E_{\alpha}\! \left (\tilde P_T^{(1)}, \tilde Q_T^{(1)} \right)  \left (P_T^{(1)} \right )^{\alpha_1} \left (Q_T^{(1)} \right )^{\alpha_2}, \quad |\alpha| = 3.
\end{align*}
We define the linear functions
\begin{align*}
	J(u) :=  (1-u) \tilde P_T^{(0)} + u \tilde P_T^{(1)}, \\
	K(u) := (1-u) \tilde Q_T^{(0)} + u \tilde Q_T^{(1)},
\end{align*}
so that 
\begin{align*}
	\frac{\partial^{3} P_{\text{BS}}}{\partial x^{\alpha_1}\partial y^{\alpha_2}}\left ( J(u), K(u) \right ) = \frac{\partial^{3} P_{\text{BS}}}{\partial x^{\alpha_1}\partial y^{\alpha_2}}\left ( (1-u) \tilde P_T^{(0)} + u \tilde P_T^{(1)}, (1-u) \tilde Q_T^{(0)} + u \tilde Q_T^{(1)} \right ). 
\end{align*} 
Since the arguments $J(u)$ and $K(u)$ are random, the application of \Cref{lem:parblowupart3} must be handled with care. This is the purpose of the following proposition.
\begin{proposition}
\label{parboundproppart3}
There exists functions $B_{\alpha}$ with $\alpha_1 + \alpha_2 = 3$ as in \Cref{lem:parblowupart3}, where the constants in the definitions of $\zeta$ and $\nu$ are possibly replaced with random variables, such that 
\begin{align*}
	\sup_{u \in (0,1)} \left | \frac{\partial^{3}  P_{\text{BS}} }{\partial x^{\alpha_1} \partial y^{\alpha_2}}\left (J(u) , K(u) \right )  \right | \leq B_{\alpha}(T, k) \quad \Qro \text{   a.s.}
\end{align*}
\end{proposition}

\begin{proof}
Since $J$ and $K$ are linear functions, then from \Cref{lem:parblowupart3} this claim is immediately true if we can show that $K$ is strictly positive $\Qro$ a.s. Recall
\begin{align*}
	K(u) &=  (1-u) \left ( \int_0^T (1 - \rho_t^2) v_{0,t}^2 \dd t \right ) + u\int_0^T ( 1 - \rho_t^2) V_t^2 \dd t.
\end{align*}
$K$ corresponds to the linear interpolation between $\int_0^T ( 1- \rho_t^2) v_{0,t}^2 \dd t$ and $\int_0^T (1 - \rho_t^2) V_t^2 \dd t$. It is clear that $\sup_{t \in [0, T]} (1 - \rho_t^2) > 0$. As $V^2$ corresponds to the variance process, this is always chosen to be a non-negative process such that the set $ \{t \in [0,T] : V_t^2 > 0 \}$ has non-zero Lebesgue measure. Thus these integrals are strictly positive and hence $K$ is strictly positive $\Qro$ a.s. 
\end{proof}
\noindent By \Cref{parboundproppart3},
\begin{align*}
	\left |E_{\alpha}\! \left (\tilde P_T^{(1)}, \tilde Q_T^{(1)} \right) \right | &= \left | \int_0^1 (1-u)^2 \frac{\partial^{|\alpha|} P_{\text{BS}}}{\partial x^{\alpha_1}\partial y^{\alpha_2}}\left ( J(u), K(u) \right ) \dd u \right | \\
	&\leq \frac{1}{3} B_{\alpha}(T, k).
\end{align*}
Thus
\begin{align}
	\left \|  E_\alpha\!  \left (\tilde P_T^{(1)}, \tilde Q_T^{(1)} \right)   \left (P_T^{(1)} \right )^{\alpha_1} \left (Q_T^{(1)} \right )^{\alpha_2} \right \| \leq \frac{1}{3} \left \| B_{\alpha}(T, k) \right \|_2 \left \| \left (P_T^{(1)} \right )^{\alpha_1} \right \|_4 \left \|  \left ( Q_T^{(1)} \right )^{\alpha_2} \right \|_4. \label{cauchy2}
\end{align}
Looking at the second and third term on the RHS of \cref{cauchy2}, it is clear one of our objectives is to bound $P_T^{(1)}$ and $Q_T^{(1)}$ in $L^p$ for $p \geq 4$. This validates \cref{item2ingredients} in \Cref{prop:ingredientsbound}.

\begin{lemma}
\label{lemma:boundremaindertermsV}
The terms from \Cref{prop:ingredientsbound} can be bounded in terms of the following quantities:
\begin{enumerate}[label = (\roman*), ref = \roman*]
\item $\| \Theta_{0,t}^{(1)} (V) \|_p$ and $\| \Theta_{0,t}^{(1)} (V^2) \|_p$ for $p \geq 2$. \label{item1remainder}
\item $\| \Theta_{1, t}^{(1)} (V) \|_p$ and $ \| \Theta_{1,t}^{(1)}(V^2)  \|_p $ for $p \geq 2$. \label{item2remainder} 
\item $\| \Theta_{2, t}^{(1)}(V) \|_p$ and $ \| \Theta_{2,t}^{(1)}(V^2)  \|_p $ for $p \geq 2$. \label{item3remainder}
\end{enumerate}
\end{lemma}
\begin{proof}
We will make extensive use of the following form of Jensen's inequality:
\begin{align}
	\left ( \int_0^T |f(u) | \dd u \right )^p \leq T^{p-1} \int_0^T |f(u)|^p \dd u, \quad p \geq 1. \label{intinequality}
\end{align}
For the rest of this proof, assume that $p \geq 2$. We will denote by $C_p$ and $D_p$ generic constants that solely depend on $p$. They may be different on each use. Notice
\begin{align*}
	P_T^{(1)} &= \int_0^T \rho_t \Theta_{0,t}^{(1)}(V) \dd B_t - \frac{1}{2} \int_0^T \rho_t^2 \Theta_{0,t}^{(1)}(V^2) \dd t,
	&Q_T^{(1)} &= \int_0^T (1 - \rho_t^2)  \Theta_{0,t}^{(1)}(V^2) \dd t.
\end{align*}
Applying the Minkowski and Burkholder-Davis-Gundy inequalities, as well as Jensen's inequality \cref{intinequality}, we obtain 
\begin{align*}
	\| P_T^{(1)} \|_p \leq C_p T^{\frac{1}{2}- \frac{1}{p}} \left ( \int_0^T \rho_t^p \| \Theta_{0,t}^{(1)} (V) \|_p^p \dd t \right )^{1/p} + \frac{1}{2} D_p T^{1 - \frac{1}{p}} \left (\int_0^T \rho_t^{2p} \| \Theta_{0,t}^{(1)}(V^2) \|_p^p \dd t \right )^{1/p}
\end{align*}
and 
\begin{align*}
	\| Q_T^{(1)} \|_p \leq C_p T^{1- \frac{1}{p}} \left ( \int_0^T (1- \rho_t^2)^p \| \Theta_{0,t}^{(1)} (V^2) \|_p^p \dd t \right )^{1/p}.
\end{align*}
Now $\left (V_t^{(1)} \right )^2 = \left (v_{0,t} + \Theta_{0,t}^{(1)}(V) \right )^2 = v_{0,t}^2 + 2v_{0,t} \Theta^{(1)}_{0,t}(V) +  \left ( \Theta^{(1)}_{0,t}(V) \right )^2$, so that 
\begin{align*}
	\Theta_{0,t}^{(1)}(V^2) = 2v_{0,t} \Theta^{(1)}_{0,t}(V) +  \left ( \Theta^{(1)}_{0,t}(V) \right )^2.
\end{align*}
This suggests that finding an $L^p$ bound on the remainder term $\Theta_{0,t}^{(1)}(V)$ is sufficient in order to bound $P_T^{(1)}$ and $Q_T^{(1)}$ in $L^p$. This validates \cref{item1remainder}.

We can write the following remainder terms of $P$ and $Q$ as
\begin{align}
\label{PQremainders}
\begin{split}
	\Theta_{2,T}^{(1)}(P) &= \int_0^T \rho_t \Theta_{2,t}^{(1)} (V) \dd B_t - \frac{1}{2} \int_0^T \rho_t^2 \Theta_{2,t}^{(1)}(V^2) \dd t, \\
	\Theta_{2,T}^{(1)}(Q) &= \int_0^T (1 - \rho_t^2) \Theta_{2,t}^{(1)} (V^2) \dd t, \\
	\Theta_{2,T}^{(1)}(P^2) &= \left (P_T^{(1)} \right)^2 - P_{1, T}^2 = (P_T^{(1)} - P_{1,T})(P_T^{(1)} + P_{1,T}) \\
					        &=  \left ( \int_0^T \rho_t \Theta_{1,t}^{(1)} (V) \dd B_t - \frac{1}{2} \int_0^T \rho_t^2 \Theta_{1,t}^{(1)}(V^2) \dd t \right ) \\
					    &\quad \cdot \left (  \int_0^T \rho_t  (2 \Theta_{0,t}^{(1)}(V) - \Theta_{1,t}^{(1)}(V) )\dd B_t - \frac{1}{2} \int_0^T \rho_t^2 (2 \Theta_{0,t}^{(1)}(V^2) - \Theta_{1,t}^{(1)}(V^2)   ) \dd t\right ), \\
	\Theta_{2, T}^{(1)}(Q^2) &= \left (Q_T^{(1)} \right)^2 - Q_{1, T}^2 = (Q_T^{(1)} - Q_{1,T})(Q_T^{(1)} + Q_{1,T}) \\
					    &=  \left ( \int_0^T (1-\rho_t^2) \Theta_{1,t}^{(1)}(V^2) \dd t \right ) \left (  \int_0^T (1-\rho_t^2) \left [ 2\Theta_{0,t}^{(1)}(V^2) - \Theta_{1,t}^{(1)}(V^2) \right ] \dd t \right ). 
\end{split}
\end{align}
Furthermore, notice
\begin{align*}
	\Theta_{1,t}^{(1)}(V^2) &= \Theta_{0,t}^{(1)}(V^2) + 2 v_{0,t} \left ( \Theta_{1,t}^{(1)}(V) - \Theta_{0,t}^{(1)}(V)\right ), \\
	\Theta_{2,t}^{(1)}(V^2) &= \Theta_{1,t}^{(1)}(V^2) - 2 v_{0,t} \left ( \Theta_{1,t}^{(1)}(V) -  \Theta_{2,t}^{(1)}(V) \right ) - \left ( \Theta_{0,t}^{(1)}(V) -  \Theta_{1,t}^{(1)}(V)  \right )^2.
\end{align*}
Then, by application of the Minkowski, Burkholder-Davis-Gundy and Cauchy-Schwarz inequalities, it is sufficient to obtain $L^p$ bounds on $\Theta_{1,t}^{(1)}(V)$ and $\Theta_{2,t}^{(1)}(V)$ in order to obtain $L^p$ bounds on the remainders of $P$ and $Q$ from \cref{PQremainders}. For the cross remainder term, we have
\begin{align*}
	\| \Theta_{2,T}^{(1)}(PQ) \|_p  \leq \| P_T^{(1)} \|_{2p} \| Q_T^{(1)} \|_{2p} + \| P_{1,T}^{(1)} \|_{2p} \| Q_{1,T}^{(1)} \|_{2p}.
\end{align*}
We just need to check how to obtain $L^p$ bounds on $P_{1,T}^{(1)}$ and $Q_{1,T}^{(1)}$. Notice
\begin{align*}
	\| P_{1,T} \|_p &\leq  C_p T^{\frac{1}{2}- \frac{1}{p}} \left ( \int_0^T \rho_t^p \| \Theta_{0,t}^{(1)} (V) - \Theta_{1,t}^{(1)} (V)\|_p^p \dd t \right )^{1/p} \\&\quad + \frac{1}{2} D_p T^{1 - \frac{1}{p}} \left (\int_0^T \rho_t^{2p} \| \Theta_{0,t}^{(1)}(V^2) - \Theta_{1,t}^{(1)}(V^2)\|_p^p \dd t \right )^{1/p}
\end{align*}
and 
\begin{align*}
	\| Q_{1,T} \|_p &\leq  D_p T^{1 - \frac{1}{p}} \left (\int_0^T (1-\rho_t^2)^{p} \| \Theta_{0,t}^{(1)}(V^2) - \Theta_{1,t}^{(1)}(V^2)\|_p^p \dd t \right )^{1/p}.
\end{align*}
Again, all we need to obtain $L^p$ bounds on the cross remainder term are $L^p$ bounds on $\Theta_{1,t}^{(1)}(V)$ and $\Theta_{2,t}^{(1)}(V)$. This validates \cref{item2remainder} and \cref{item3remainder}.
\end{proof}

\input{FastcalibrationGeneralmodelexpmalliavin_JCAM}
\input{NumericsGeneralmodelexpmalliavin_JCAM}
\input{SABRGeneralmodelexpmalliavin}
\section{Conclusion}
\noindent We have established a second-order approximation for the price of a put option in a general stochastic volatility framework. The general drift and power-type diffusion coefficients of the stochastic volatility model satisfy some regularity conditions, for which we provided sufficient conditions regarding existence of an equivalent martingale measure. When parameters are assumed to be piecewise-constant, our approximation formula becomes closed-form. In addition, this assumption allows us to devise a fast calibration scheme by exploiting recursive properties of the iterated integral operators in terms of which our approximation formulas are expressed. We established the explicit form of the error term induced by the expansion, and bounded it in terms of moments of functionals of the underlying volatility process. We performed a numerical sensitivity analysis for the approximation formula in the SABR-$\mu$ and Stochastic Verhulst models, and showed that the error is small, behaves as we expect with respect to parameter changes, and is within an acceptable range for application purposes. It is of our opinion that such general second-order approximation formula will be very useful for practitioners since: obtaining the pricing formula for different stochastic volatility models only requires basic differentiation, the formula is essentially instantaneous to compute, the fast calibration scheme can be used to calibrate models rapidly, and the error is sufficiently low for the purposes of application.


\bibliographystyle{plainnat_lastnamefirst.bst}
\bibliography{References.bib}

\appendix

\section{Mixing solution}
\label{appen:mixingsol}
In this appendix, we present a derivation of the result referred to as the mixing solution by \citep{hull1987pricing}. This result is crucial for the expansion methodology implemented in \Cref{sec:expprocedure3}. Hull and White first established the expression for the case of independent Brownian motions $W$ and $B$. Later on, this was extended to the correlated Brownian motions case, see \citep{Willard97, romano1997contingent}. 
\begin{theorem}[Mixing solution]
\label{thm:mixingsol}
Under a chosen domestic equivalent martingale measure $\Qro$, suppose that the spot $S$ with volatility $V$ are given as the solution to the general model \cref{eqn:gen}. Define $X$ as the log-spot and $k$ the log-strike. Namely, $X_t = \ln S_t$ and $k = \ln K$. Then
\begin{align*}
	\text{Put}_{\text{G}} = e^{-\int_0^T\! r_t^d \dd t } \Ex [( e^k- e^{X_T})_+] &= \Ex \left [  e^{-\int_0^T\! r_t^d \dd t } \Ex \big  [ (e^k - e^{X_T})_+ | \FF_T^B \big ] \right ] \\
			&= \Ex \left [ P_{\text{BS}} \left ( x_0 - \int_0^T \frac{1}{2} \rho_t ^2 V_t^2 \dd t + \int_0^T \rho_t V_t  \dd B_t , \int_0^T V_t^2 ( 1 - \rho_t^2) \dd t \right ) \right ],
\end{align*}
where $P_{\text{BS}}$ is given in \cref{eqn:PBS}.
\end{theorem}
\begin{proof}
By writing the driving Brownian motion of the spot as $W_t = \int_0^t \rho_u  \dd B_u +  \int_0^t \sqrt{1 - \rho_u^2} \dd Z_u$, where $Z$ is a Brownian motion under $\Qro$ which is independent of $B$, this yields the pathwise unique strong solution of $X$ as 
\begin{align*}
	X_T &= x_0 + \int_0^T \left ( r_t^d - r_t^f - \frac{1}{2} V_t^2 \right ) \dd t + \int_0^T \rho_t V_t \dd B_t + \int_0^T V_t \sqrt{1 - \rho_t^2 } \dd Z_t.
\end{align*}
First, notice that $V$ is adapted to the filtration $(\FF_t^B)_{0 \leq t \leq T}$. Thus, it is evident that $X_T|\FF_T^B$ will have a normal distribution. Namely,
\begin{align*}	
	X_T | \FF_T^B &\sim \mathcal{N}\left (\hat \mu (T) , \tilde \sigma^2(T) \right ), \\
	\hat \mu(T) &:=  x_0  + \int_0^T \left ( r_t^d - r_t^f  \right )\dd t- \frac{1}{2} \int_0^T  V_t^2 \dd t + \int_0^T \rho_t V_t \dd B_t,\\
	\tilde \sigma^2(T) &:=  \int_0^T  V_t^2  (1 - \rho_t^2) \dd t.
\end{align*}
Also, let $\tilde \mu(T) := \hat \mu(T) - \int_0^T (r_t^d - r_t^f) \dd t$. Hence the calculation of $e^{-\int_0^T r_t^d \dd t } \Ex \big [(e^k - e^{X_T})_+ | \FF_T^B \big ]$ will result in a Black-Scholes-like formula.
\begin{align*}
&e^{-\int_0^T r_t^d \dd t } \Ex \big [(e^k - e^{X_T})_+ | \FF_T^B \big ] \\ &= e^k e^{ -\int_0^T r_t^d \dd t } \,\NN\! \left  (\frac{ k - \hat \mu(T) }{\tilde \sigma(T)}\right ) - e^{-\int_0^T r_t^d \dd t } e^{ \hat \mu(T) + \frac{1}{2} \tilde \sigma^2(T) } \,\NN\! \left  ( \frac{ k - \hat \mu(T)-  \tilde \sigma^2(T)}{\tilde \sigma(T)}\right ) \\
&= e^k e^{ -\int_0^T r_t^d \dd t }  \,\NN\! \left  ( \frac{k - \hat \mu(T) - \frac{1}{2}   \tilde \sigma^2(T)}{\tilde \sigma(T)} + \frac{1}{2} \tilde \sigma(T)\right ) \\&\quad - e^{\tilde \mu(T) + \frac{1}{2} \tilde \sigma^2 (T) } e^{-\int_0^T r_t^f \dd t }\,\NN\! \left  (\frac{k - \hat \mu(T)  - \frac{1}{2} \tilde \sigma^2(T)}{\tilde \sigma(T)} - \frac{1}{2} \tilde \sigma(T)\right ) \\
&= e^k e^{ -\int_0^T r_t^d \dd t } \,\NN\! \left  ( \frac{k - (\tilde \mu(T) + \frac{1}{2} \tilde \sigma^2(T) ) - \int_0^T (r_t^d - r_t^f) \dd t}{\tilde \sigma(T)} + \frac{1}{2} \tilde \sigma(T)\right ) \\&\quad - e^{ \tilde \mu(T) + \frac{1}{2} \tilde \sigma^2 (T) } e^{-\int_0^T r_t^f \dd t }\,\NN\! \left  ( \frac{k - (\tilde \mu(T) + \frac{1}{2} \tilde \sigma^2(T) ) -\int_0^T (r_t^d - r_t^f) \dd t}{\tilde \sigma(T)} - \frac{1}{2} \tilde \sigma(T)\right ).
\end{align*}
It is now immediate that $ e^{-\int_0^T\! r_t^d \dd t } \,\Ex \big [(e^k- e^{X_T})_+ | \FF_T^B \big ]= P_{\text{BS}}\! \left (\tilde \mu(T) + \frac{1}{2} \tilde \sigma^2(T), \tilde \sigma^2(T) \right )$.
\end{proof}

\section{Existence of a domestic equivalent martingale measure}
\label{appen:existenceriskneutral}
\noindent 
In this appendix, we consider the existence of a domestic equivalent martingale measure in the Verhulst model. To elaborate, first consider the general model \cref{eqn:gen} under the probability measure $\Qro$, equivalent to the real world measure. Let 
\begin{align*}
    F_t := S_t e^{\int_0^t r_u^f \dd u}/e^{\int_0^t r_u^d \dd u }
\end{align*}
be the foreign bank account denominated in units of domestic currency and discounted in units of domestic currency. If $(F_t)$ is a local martingale under $\Qro$, then $\Qro$ is called a domestic equivalent local martingale measure. If $(F_t)$ is a martingale under $\Qro$, then $\Qro$ is called an equivalent martingale measure. 

By the fundamental theorem of asset pricing, the existence of an equivalent local martingale measure is equivalent to the absence of free lunch with vanishing risk. However, when an equivalent local martingale measure exists, the discounted asset price may be a strict local martingale, implying the failure of pricing via expectation and the presence of bubble phenomena. For this reason, we focus on establishing conditions ensuring the existence of a domestic equivalent martingale measure.

The issue pertaining to the existence of an equivalent martingale measure in stochastic volatility models has been identified and studied by \citep{lions2007correlations} as well as \citep{andersen2007moment}. Both articles provide conditions for when a (time-homogeneous) stochastic volatility model is legally specified under an equivalent martingale measure. A comprehensive survey on this issue is given in \citep{bernard2017martingale}. However, to our knowledge, the case of time-inhomogeneous stochastic volatility models (that is, with time-dependent parameters) has not been studied in the literature. Thus, we extend these results to the time-inhomogeneous setting. First, notice that the general model \cref{eqn:gen} can be reexpressed as
\begin{align}
\begin{split}
	\dd F_t &= F_t V_t \dd W_t, \quad F_0 = S_0,\\ \label{genF}
	\dd V_t &= \alpha(t, V_t) \dd t + \beta(t, V_t) \dd B_t, \quad V_0 = v_0, \\
	\dd \langle W, B \rangle_t &= \rho_t \dd t,
\end{split}
\end{align}
where we stress that $W$ and $B$ are Brownian motions with deterministic, time-dependent instantaneous correlation $(\rho_t)_{0 \leq t \leq T}$, defined on the filtered probability space $(\Omega, \FF, (\FF_t)_{0 \leq t \leq T}, \Qro)$. The question of whether or not \cref{eqn:gen} is truly specified under a domestic equivalent martingale measure is equivalent to the question of whether or not $(F_t)$ is a martingale under $\Qro$.

\begin{theorem}
\label{thm:existencemeasure}
Let \Cref{assregA} hold. Suppose the following condition is true:
\begin{align}
	\sup_{t \in [0,T]} \limsup_{x \to \infty} \frac{ \rho_t \beta(t,x) x + \alpha(t, x) }{x} < \infty. \label{condition1}
\end{align}
Then $(F_t)$ is a martingale.
\end{theorem}
\begin{proof}[Proof]
We essentially follow in a similar vein to \citep[][Theorem 2.4 (i)]{lions2007correlations}. Let $\tau_n := \inf \{t \geq 0: V_t > n \}$ be the first time $V$ crosses the level $n$. Clearly $\tau_n \uparrow \infty$ $\Qro$ a.s. Define $F_t^n := F_{t \wedge \tau_n}$. It is well known that $(F_t^n)_t$ and $(F_t)$ possess the pathwise unique strong solutions
\begin{align*}
	F_t^n = F_0 \exp \left (\int_0^t V_u \rind{\{u \leq \tau_n\} } \dd W_u - \frac{1}{2}\int_0^t V_u^2 \rind{\{u \leq \tau_n\} } \dd u \right )
\end{align*}
and
\begin{align*}
	F_t = F_0 \exp \left (\int_0^t V_u \dd W_u - \frac{1}{2} \int_0^t V_u^2 \dd u \right )
\end{align*}
respectively. Now $(F_t^n)_t$ is a martingale for each $n \in \naturals$. Furthermore, $(F_t)$ is a non-negative local-martingale. Thus, by Fatou's lemma, $(F_t)$ is a non-negative supermartingale.
Utilising the condition \cref{condition1}, this implies that there exists some constant $M>0$ such that
\begin{align}
	\rho_t \beta(t,x) x + \alpha(t,x) \leq M(1 + x), \label{importantbound1}
\end{align}
for all $x \geq 0$, uniformly in $t \leq T$. Now if we show that 
\begin{align*}
	\sup_n \Ex[ F_t^n \log(F_t^n) ] < \infty
\end{align*}
then by the Vall\'ee Poussin theorem, $(F_t^n)_n$ is uniformly integrable and thus $\Ex[F_t] = F_0$ for each $t \leq T$. This, combined with the fact that $(F_t)$ is a non-negative supermartingale, will ensure $(F_t)$ is a martingale. Now,
\begin{align*}
	F_t^n \log (F_t^n) = F_t^n \left ( \log(F_0) + \int_0^t V_u \rind{\{u \leq \tau_n\} } \dd W_u - \frac{1}{2} \int_0^t V_u^2  \rind{\{u \leq \tau_n\} } \dd u \right ).
\end{align*}
Notice $F_T^n/F_0$ is a Radon-Nikodym derivative which defines a measure $\hat \Qro$. And by Girsanov's theorem, $\hat W_t = W_t - \int_0^t V_u \rind{ \{ u \leq \tau_n \} }\dd u$ is a Brownian motion under $\hat \Qro$. Denote the expectation under $\hat \Qro$ by $\hat \Ex$. Then
\begin{align*}
	\Ex[F_T^n \log (F_T^n)] &= F_0 \log(F_0) + \Ex \left [F_T^n \left (\int_0^T V_t  \rind{\{t \leq \tau_n\} } \dd W_t - \frac{1}{2} \int_0^T V_t^2  \rind{\{t \leq \tau_n\} } \dd t \right )\right ] \\
	&= F_0 \log(F_0) + F_0 \hat \Ex \left [\int_0^T V_t \rind{\{t \leq \tau_n\} } \dd \hat W_t + \frac{1}{2} \int_0^T V_t^2  \rind{\{t \leq \tau_n\} } \dd t \right ] \\
		&=F_0 \log(F_0) + \frac{F_0}{2} \int_0^T \hat \Ex[ V_t^2  \rind{\{t \leq \tau_n\} } ] \dd t \\
		&\leq F_0 \log (F_0) + \frac{F_0}{2} \int_0^T \hat \Ex[V_t^2] \dd t.
\end{align*}
So it suffices to determine a bound on $\hat \Ex[V_t^2]$ for $t \leq T$. Now write $B_t = \int_0^t \rho_u \dd W_u + \int_0^t \sqrt{1-\rho_u^2 } \dd Z_u$ where $Z$ is a $\Qro$ Brownian motion independent of $W$. Then 
\begin{align*}
	B_t = \int_0^t \rho_u V_u \rind{\{u \leq \tau_n \}} \dd u + \int_0^t \rho_u \dd \hat W_u + \int_0^t \sqrt{1-\rho_u^2 } \dd Z_u.
\end{align*}
Hence,
\begin{align*}
	\dd V_t = \Big [\rho_t \beta(t, V_t) V_t \rind{\{t \leq \tau_n \}}  + \alpha(t, V_t) \Big] \dd t + \beta(t, V_t) \dd \hat B_t
\end{align*}
where $\hat B_t = \int_0^t \rho_u \dd \hat W_u + \int_0^t \sqrt{1-\rho_u^2} \dd Z_u$ is a Brownian motion under $\hat \Qro$. Utilising \cref{importantbound1} yields
\begin{align*}
	\dd V_t \leq M (1 + V_t) \dd t + \beta(t, V_t) \dd \hat B_t.
\end{align*}
Then for the second moment:
\begin{align*}
	\hat \Ex [V_t^2] &\leq  \hat \Ex \left [ \left ( v_0 +M \int_0^t (1+V_u) \dd u + \int_0^t \beta(u, V_u) \dd \hat B_u \right )^2 \right ] \\
				&\leq 3 \left (v_0^2 + M^2 \hat \Ex  \left [ \left (\int_0^t (1+V_u) \dd u\right )^2 \right ] + \hat \Ex \left [ \left (\int_0^t \beta(u,V_u) \dd \hat B_u \right )^2 \right ] \right ) \\
				&\leq 3 \left (v_0^2 + M^2 T \int_0^t \hat\Ex [(1 + V_u)^2] \dd u + \int_0^t \hat \Ex \big [(\beta(u, V_u)^2\big ] \dd u \right )  \\
				&\leq 3 \left (v_0^2 + M^2 T \int_0^t \hat \Ex \Big [ (2 + 2V^2_u) + \frac{K}{M^2 T}(1 + V_u^2) \Big ] \dd u \right ),
\end{align*}
where we have used that $(\sum_{i = 1}^n |a_i|)^2 \leq n \sum_{i =1}^n a_i^2$, $\left (\int_0^t f(u) \dd u \right)^2 \leq t \int_0^t f^2(u) \dd u$ (Jensen's inequality) and \cref{assregA3} in \Cref{assregA}. Define $m_{2,t} := \hat \Ex [V_t^2]$. Then (after redefining constants), this suggests we study the integral inequality
\begin{align*}
	m_{2,t} \leq c_0(1 + t) + c\int_0^t m_{2,u} \dd u.
\end{align*} 
Utilising Gronwall's inequality, we obtain
\begin{align*}
	m_{2,t} \leq c_0(1 + t) e^{c t}.
\end{align*} 
\end{proof}

\begin{proposition}
\label{prop:existencemeasurexgbm}
Consider the Verhulst model \cref{eqn:xgbm} specified under the probability measure $\Qro$. Then $\Qro$ is a domestic equivalent martingale measure if
\begin{align*}
	\rho_t  \lambda_t - \kappa_t\leq 0
\end{align*}
holds for all $t \in [0,T]$.
\end{proposition}
\begin{proof}
We utilise \Cref{thm:existencemeasure} with $\alpha(t,x) = \kappa_t(\theta_t - x)x$ and $\beta(t,x) = \lambda_t x$. Then
\begin{align*}
	\frac{\rho_t \lambda_t x^2 + \kappa_t(\theta_t -x )x }{x} = x (\rho_t \lambda_t - \kappa_t) + \kappa_t \theta_t.
\end{align*}
Clearly \cref{condition1} is satisfied if and only if $\rho_t \lambda_t - \kappa_t \leq 0$ for all $t \in [0,T]$. So in this case, $\Qro$ is a domestic equivalent martingale measure.
\end{proof}

\section{Malliavin calculus machinery}
\label{appen:malliavincalculusmachinery}
\noindent 
This appendix provides a short excerpt on Malliavin calculus. This is predominantly to fix notation. We point the reader towards \citep{nualart2006malliavin} for a complete and accessible source on Malliavin calculus.

The underlying framework of Malliavin calculus involves a so-called isonormal Gaussian process $\tilde W$. Specifically, $\tilde W = \{ \tilde W(h) : h \in H \}$ is a zero-mean Gaussian process induced by an underlying real, separable Hilbert space $H$ such that $\Ex(\tilde W(h) \tilde W(g)) = \langle h, g \rangle_H$. We need only make use of Malliavin calculus when the underlying Hilbert space is
\begin{align*}
	H = L^2([0, T]) \equiv L^2([0, T], \mathcal{B}([0,T]), \lambda^*),
\end{align*}
where $\lambda^*$ is the one-dimensional Lebesgue measure. Thus the inner product on $H$ is 
\begin{align*}
	\langle h, g \rangle_H = \int_0^T h_t g_t \lambda^*(\dd t) = \int_0^T h_t g_t \dd t.
\end{align*}
Our Gaussian process $\tilde W$ will be explicitly given as $\tilde W(h) := \int_0^T h_t \dd \tilde B_t$, where $\tilde B$ is a Brownian motion with natural filtration $(\FF_t^{\tilde B})_{0 \leq t \leq T}$ and $h \in L^2([0, T])$. By use of the zero-mean and It\^o isometry properties of the It\^o integral, it can be seen that such a Hilbert space $H$ and Gaussian process $\tilde W$ satisfy the framework for Malliavin calculus.

\begin{definition}[Malliavin derivative]
Let 
\begin{align*}
	\mathcal{S}_n := \left \{ F = f\left ( \int_0^T h_{1, t} \dd \tilde B_t , \dots, \int_0^T h_{n, t} \dd \tilde B_t  \right ) : f \in C_p^\infty (\reals^n ; \reals),  h_{i, \cdot} \in  H  \right \}
\end{align*}
and $\mathcal{S} := \bigcup_{n \geq 1} \mathcal{S}_n$. 
Here $C_p^\infty(\reals^n ; \reals)$ is the space of smooth Borel measurable functions $f: (\reals^n, \mathcal{B}(\reals^n)) \to (\reals, \borel)$ which have at most polynomial growth.  Thus, the elements of $\mathcal{S}_n$ are random variables. For $F \in \mathcal{S}_n$, the Malliavin derivative $D$ is an unbounded operator from $\mathcal{S}_n \subseteq L^p(\Omega) \to L^p([0, T] \times \Omega )$ for $p \geq 1$ and is given by 
\begin{align*}
	D_tF := \sum_{i = 1}^n \partial_i f\left ( \int_0^T h_{1, u} \dd \tilde B_u , \dots, \int_0^T h_{n, u} \dd \tilde B_u  \right ) h_{i, t}.
\end{align*}
It can be shown that the Malliavin derivative $D$ is a closable operator on $L^p(\Omega)$ into $L^p([0, T] \times \Omega)$. We denote the closed extension of it again by $D$ and moreover, its domain is denoted by $\mathbb{D}^{1, p}$. Another way to think about the domain $\mathbb{D}^{1, p}$ is as the completion of $\mathcal{S}$ with respect to the seminorm 
\begin{align*}
	\|F\|_{1, p} := \left ( \Ex|F|^p + \Ex \left [ \left (\int_0^T (D_t F)^2 \dd t \right )^{p/2} \right ] \right )^{1/p}
\end{align*}
where $F \in \mathcal{S}$ and $p \geq 1$.
\end{definition}

The Malliavin derivative satisfies a duality relationship. In the case of adapted processes, the duality relationship reads as follows.
\begin{proposition}[Malliavin duality relationship] 
\label{prop:duality}
Let $G \in \mathbb{D}^{1, 2}$ and $\alpha \in L^2([0, T] \times \Omega)$ such that $\alpha$ is adapted to the filtration $(\FF_t^{\tilde B})_{0 \leq t \leq T}$. Then 
\begin{align*}
	\Ex \left [\int_0^t \alpha_s (D_sG) \dd s \right ] = \Ex \left [G \int_0^t \alpha_s \dd \tilde B_s \right ]
\end{align*}
for any $t < T$. 
\end{proposition}
\begin{proof}
See \citep{nualart2006malliavin}.
\end{proof}

\begin{lemma}[Malliavin integration by parts] 
\label{lem:mintbyparts}
Let $\tilde T \leq T$ and $\hat T \leq T $. Also, let $\alpha \in L^2([0, T] \times \Omega)$ such that $\alpha$ is adapted to $(\FF_t^{\tilde B})_{0 \leq t \leq T}$. Then,
\begin{align*}
	\Ex \left [ \ell\!\left (\int_0^{\tilde T} h_u \dd \tilde B_u \right ) \left ( \int_0^{\hat T} \alpha_u \dd \tilde B_u \right ) \right ] = \Ex \left [ \ell'\!\left (\int_0^{\tilde T} h_u \dd \tilde B_u \right ) \left ( \int_0^{\tilde T \wedge \hat T}h_u \alpha_u \dd u \right ) \right ].
\end{align*}
In particular, for $\tilde T = T$ and $\hat T = t < T$,
\begin{align*}
	\Ex \left [ \ell\!\left (\int_0^{T} h_u \dd \tilde B_u \right ) \left ( \int_0^{t} \alpha_u \dd \tilde B_u \right ) \right ] = \Ex \left [\ell'\!\left (\int_0^{T}h_u \dd \tilde B_u \right ) \left ( \int_0^{t} h_u  \alpha_u \dd u \right ) \right ].
\end{align*}
\end{lemma}
\begin{proof}
Let $G = \ell\!\left (\int_0^{\tilde T} h_u \dd \tilde B_u \right )$. Then $G \in \mathcal{S}_1 \subseteq \mathbb{D}^{1, 2}$ and $D_tG = \ell'\!\left ( \int_0^{\tilde T} h_u \dd \tilde B_u \right ) h_t \textbf{1}_{ \{ t \leq \tilde T \} }$. The result follows by a consequence of \Cref{prop:duality}.
\end{proof}

\section{$P_{\text{BS}}$ partial derivatives}
\label{appen:partialderivativesPBS}
\noindent 
This appendix contains some partial derivatives for the Black-Scholes put option formula $P_{\text{BS}}$ \cref{eqn:PBS}. One can think of these partial derivatives as being analogous to the Black-Scholes Greeks. However, these are slightly different as our Black-Scholes formulas are parameterised with respect to log-spot and integrated variance rather than spot and volatility respectively. Recall that $\NN$ and $\phi$ denote the standard normal distribution and density functions respectively.

\subsection{First-order $P_{\textnormal{BS}}$}
\begin{align*}
	\partial_x P_{\text{BS}} &= e^x e^{-\int_0^T\! r_u^f \dd u } \left ( \NN (d_+^{\ln} ) - 1 \right ), \\ 
	\partial_y P_{\text{BS}} &= \frac{e^x e^{-\int_0^T\! r_u^f \dd u} \phi(d^{\ln}_+)}{2 \sqrt{y}}.
\end{align*}

\subsection{Second-order $P_{\textnormal{BS}}$}
\begin{align*}
	\partial_{xx} P_{\text{BS}} 	&=  \frac{e^x e^{-\int_0^T\! r^f_u \dd u } \phi(d^{\ln}_+)}{\sqrt{y}} + \partial_x P_{\text{BS}} \\
	&=   \frac{e^x e^{-\int_0^T\! r^f_u \dd u } \phi(d^{\ln}_+)}{\sqrt{y}} + e^x e^{-\int_0^T\! r^f_u \dd u }\left ( \NN(d_+^{\ln}) - 1 \right ), \\
	\partial_{xy} P_{\text{BS}} &= (-1) \frac{e^x e^{-\int_0^T\! r^f_u \dd u } \phi(d^{\ln}_+) d^{\ln}_-}{2y}, \\
	\partial_{yy} P_{\text{BS}} &= \frac{e^x e^{-\int_0^T\! r_u^f \dd u} \phi(d^{\ln}_+)}{4 y^{3/2}} (d^{\ln}_- d^{\ln}_+ - 1).
\end{align*}

\subsection{Third-order $P_{\textnormal{BS}}$}
\begin{align*}
	\partial_{xxx} P_{\text{BS}} &=  \frac{e^x e^{-\int_0^T\! r^f_u \dd u } \phi(d^{\ln}_+)}{y} ( \sqrt{y} - d_+^{\ln})  + \partial_{xx} P_{\text{BS}} \\
						  &= \frac{e^x e^{-\int_0^T\! r^f_u \dd u } \phi(d^{\ln}_+)}{y} ( 2 \sqrt{y} - d_+^{\ln} ) +  e^x e^{-\int_0^T\! r^f_u \dd u } (\NN(d_+^{\ln}) -1), \\
	\partial_{xxy} P_{\text{BS}} &= (-1) \frac{e^x e^{-\int_0^T\! r_u^f \dd u }\phi(d^{\ln}_+)}{2y^{3/2}} \Big (d^{\ln}_- \sqrt{y} + (1 - d_-^{\ln} d^{\ln}_+)\Big ),  \\
	\partial_{xyy} P_{\text{BS}} &=  \frac{e^x e^{-\int_0^T\! r_u^f \dd u }\phi(d^{\ln}_+)}{4y^{2}} \Big ((2 d^{\ln}_+ - \sqrt{y}) + (1 - d_-^{\ln} d^{\ln}_+) ( d_+^{\ln} - \sqrt{y} ) \Big ), \\
	\partial_{yyy} P_{\text{BS}} &= \frac{ e^x e^{ - \int_0^T  r_u^f \dd u } \phi( d^{\ln}_+)}{8 y^{5/2}} \left ( (d^{\ln}_- d^{\ln}_+ - 1)^2 - (d^{\ln}_- + d^{\ln}_+)^2 + 2  \right ).
\end{align*}

\subsection{Fourth-order $P_{\textnormal{BS}}$}
\begin{align*}
	\partial_{xxxx} P_{\text{BS}} &= (-1)  \frac{e^x e^{-\int_0^T\! r^f_u \dd u } \phi(d^{\ln}_+)}{y^{3/2}} ( 1 - (d_+^{\ln} - \sqrt{y} )^2)  + \partial_{xxx} P_{\text{BS}} \\
						&=   \frac{e^x e^{-\int_0^T\! r^f_u \dd u } \phi(d^{\ln}_+)}{y^{3/2}}  \left [ (d_+^{\ln} - \sqrt{y} )^2 + 2y - d_+^{\ln} \sqrt{y} - 1\right ] + e^x e^{-\int_0^T\! r^f_u \dd u } (\NN(d_+^{\ln}) -1), \\
	\partial_{xxxy} P_{\text{BS}} &= \frac{e^x e^{-\int_0^T\! r^f_u \dd u } \phi(d^{\ln}_+)}{2y^{2}} \Bigg [ (\sqrt{y} - d_+^{\ln} ) (d_-^{\ln} d_+^{\ln} - 2) + (\sqrt{y} + d_-^{\ln}) - d^{\ln}_- y - \sqrt{y} (1 - d_-^{\ln} d^{\ln}_+)  \Bigg ], \\
	\partial_{xxyy} P_{\text{BS}} 
	&= (-1) \frac{e^x e^{-\int_0^T\! r_u^f \dd u}\phi(d^{\ln}_+)}{2y^{5/2}} \Bigg [ 3d^{\ln}_-d^{\ln}_+ + \frac{1}{2} (d^{\ln}_-)^2 d^{\ln}_+ \sqrt{y} - \frac{1}{2} (d_-^{\ln})^2 (d^{\ln}_+)^2 \\&\quad + \frac{1}{2} y - \frac{1}{2} \sqrt{y} \left ( 2 d^{\ln}_- + d^{\ln}_+ \right ) - \frac{3}{2} \Bigg ], \\
	\partial_{xyyy} P_{\text{BS}} &=  \frac{e^x e^{-\int_0^T\! r^f_u \dd u } \phi(d^{\ln}_+)}{8y^{7/2}} \Bigg [ 2y^{3/2}(d_-^{\ln} d_+^{\ln} -1 )(2 d_+^{\ln} - \sqrt{y} ) - 4\sqrt{y}(d_-^{\ln} + d_+^{\ln} ) \\&\quad + \sqrt{y} (\sqrt{y} - d_+^{\ln} ) \Big ( (d_-^{\ln} d_+^{\ln} - 1)^2 - (d_-^{\ln} + d_+^{\ln} )^2 + 2 \Big )\Bigg ], \\
	\partial_{yyyy} P_{\text{BS}} &= \frac{ e^x e^{ - \int_0^T  r_u^f \dd u } \phi( d^{\ln}_+)}{8 y^{7/2}} \Bigg ( \frac{1}{2} (d_-^{\ln} d_+^{\ln} - 1)^2(d_-^{\ln} d_+^{\ln} - 5) - (d_-^{\ln}d_+^{\ln} - 1)(d_-^{\ln} + d_+^{\ln}) \\&\quad - \frac{1}{2}(d_-^{\ln} + d_+^{\ln})^2 (d_-^{\ln} d_+^{\ln} - 7) + (d_-^{\ln} d_+^{\ln} - 1)   \Bigg ).
\end{align*}


\section{Proof of \Cref{thm:expprice}}
\label{append:expprice}
\noindent 
In this appendix, we provide the proof of \Cref{thm:expprice}. In order to do so, we will utilise results from Malliavin calculus extensively. A short treatment of Malliavin calculus is presented in \Cref{appen:malliavincalculusmachinery}. In addition to Malliavin calculus machinery, we will require the following ingredients.

\begin{proposition}[$P_{\text{BS}}$ partial derivative relationship]
\label{prop:partial}
\begin{align*}
	\partial_y P_{\text{BS}} (x, y) = \frac{1}{2} \left (\partial_{xx} P_{\text{BS}}(x,y) - \partial_x P_{\text{BS}}(x,y) \right ).
\end{align*}
\end{proposition}
\begin{proof}
A simple application of differentiation yields the result.
\end{proof}

In addition, we will make extensive use of the stochastic integration by parts formula, which we will list here for convenience. 
\begin{proposition}[Stochastic integration by parts]
\label{rem:sintbyparts}
Let $X$ and $Y$ be semimartingales with respect to a filtration $(\tilde \FF_t)$. Then we have 
\begin{align*}
	X_T Y_T = \int_0^T X_t \dd Y_t + \int_0^T Y_t \dd X_t + \int_0^T \dd \langle X, Y \rangle_t,
\end{align*}
given that the above It\^o integrals exist. 
In particular, if $X_t = \int_0^t x_u \dd \tilde X_u$ and $Y_t = \int_0^t y_u \dd \tilde Y_u$, where $\tilde X$ and  $\tilde Y$ are semimartingales and $x$ and $y$ are stochastic processes adapted to the underlying filtration $(\tilde \FF_t)$ such that $X$ and $Y$ exist, then the stochastic integration by parts formula reads as 
\begin{align*}
	\int_0^T x_t \dd \tilde X_t \int_0^T y_t \dd \tilde Y_t = \int_0^T \left (\int_0^t x_u \dd \tilde X_u \right ) y_t \dd \tilde Y_t +\int_0^T \left (\int_0^t y_u \dd \tilde Y_u \right ) x_t \dd \tilde X_t + \int_0^T x_t y_t \dd \langle \tilde X, \tilde Y \rangle_t.
\end{align*}
\end{proposition}

\begin{lemma}
\label{lem:detfunc}
Let $Z$ be a semimartingale such that $Z_0 = 0$ and let $f$ be a Lebesgue integrable deterministic function. Then 
\begin{align*}
	\int_0^T f_t Z_t \dd t = \int_0^T \omega_{t,T}^{(0,f)} \dd Z_t
\end{align*}
where $\omega_{t, T}^{(0, f)}$ is the integral operator given in \Cref{defn:integraloperator}.
\end{lemma}
\begin{proof}
A simple application of \Cref{rem:sintbyparts} (stochastic integration by parts) gives the desired result.
\end{proof}

Now we can proceed with the proof of \Cref{thm:expprice}. The proof is quite long and arduous, and hence will be broken up over a number of subsections, dedicated respectively to the calculation of $\Ex \tilde P_{\textnormal{BS}}, C_x, C_y, C_{xx}, C_{yy}$ and $C_{xy}$. In order to reduce the number of brackets, we will extensively utilise $\Ex$ without brackets to denote expectation of everything to the right of it. For the time being we will not yet enforce \Cref{assumpbeta}.


\subsection{$\Ex \tilde P_{\textnormal{BS}} $}
\label{sec:C0}
Notice that $\Ex \tilde P_{\text{BS}} = g(0) = \Ex(e^k - e^{X_T^{(0)}})_+.$ Since the perturbed volatility process $V_t^{(\vep)}$ is deterministic when $\vep = 0$, then $g(0)$ will just be a Black-Scholes formula. Thus we have 
\begin{align*}
	\Ex \tilde P_{\text{BS}} = P_{\text{BS}}\!\left ( x_0 , \int_0^T\! v_{0,t}^2 \dd t  \right ).
\end{align*} 

\subsection{$C_x$}
\label{sec:Cx}
Using \Cref{lem:mintbyparts} (Malliavin integration by parts),
\begin{align*}
	\Ex  \partial_x \tilde P_{\text{BS}} \int_0^T \rho_t  \left (V_{1,t} + \frac{1}{2} V_{2,t}  \right )\dd B_t  = \Ex  \partial_{xx} \tilde P_{\text{BS}} \int_0^T \rho_t^2 v_{0,t}   \left (V_{1,t} + \frac{1}{2} V_{2,t} \right )\dd t .
\end{align*}
Furthermore, using  \Cref{prop:partial} ($P_{\text{BS}}$ partial derivative relationship),
\begin{align*}
	\Ex  \partial_{xx} \tilde P_{\text{BS}} \int_0^T\rho_t^2 v_{0,t}  \left ( V_{1,t} + \frac{1}{2} V_{2,t} \right ) \dd t   = \Ex (2\partial_y + \partial_x ) \tilde P_{\text{BS}} \int_0^T\rho_t^2 v_{0,t}  \left (V_{1,t} + \frac{1}{2} V_{2,t} \right ) \dd t.
\end{align*}
Thus
\begin{align*}
	C_x = 2 \Ex \partial_y \tilde P_{\text{BS}} \int_0^T\rho_t^2 v_{0,t} \left ( V_{1,t} + \frac{1}{2} V_{2,t} \right ) \dd t -  \frac{1}{2}  \Ex \partial_x \tilde P_{\text{BS}} \int_0^T \rho_t^2 V_{1,t}^2 \dd t.
\end{align*}

\subsection{$C_{xx}$}
\label{sec:Cxx}
For $C_{xx}$ we first use \Cref{rem:sintbyparts} (stochastic integration by parts) to reduce this expression. 
\begin{align*}
C_{xx} &= \frac{1}{2}  \Ex \partial_{xx} \tilde P_{\text{BS}} \left ( \int_0^T \rho_t  V_{1,t} \dd B_t  - \int_0^T  \rho_t^2 v_{0,t} V_{1,t} \dd t  \right )^2 \\
&= \frac{1}{2} \Ex \partial_{xx} \tilde P_{\text{BS}} \left ( \int_0^T \rho_t   V_{1,t}\dd B_t \right )^2 - \Ex \partial_{xx} \tilde P_{\text{BS}} \left (\int_0^T   \rho_t^2 v_{0,t} V_{1,t}  \dd t \right ) \left ( \int_0^T \rho_t   V_{1,t} \dd B_t  \right ) \\
&\quad + \frac{1}{2} \Ex \partial_{xx} \tilde P_{\text{BS}} \left (   \int_0^T \rho_t^2 v_{0,t}  V_{1,t} \dd t \right ) ^2 \\
&= \Ex \partial_{xx} \tilde P_{\text{BS}} \left ( \int_0^T {\left \{ \int_0^t\rho^2_s v_{0,s} V_{1,s} \dd s \right \}} \rho_t^2 v_{0,t} V_{1,t}  \dd t  \right ) \\ 
&\quad - \Ex \partial_{xx} \tilde P_{\text{BS}}  \left (\int_0^T { \left \{ \int_0^t  \rho_s^2 v_{0,s} V_{1,s} \dd s \right \} }  \rho_t V_{1,t} \dd B_t  + \int_0^T { \left \{  \int_0^t \rho_s V_{1,s} \dd B_s \right \}} \rho_t^2 v_{0,t} V_{1,t} \dd t  \right ) \\
&\quad + \Ex \partial_{xx} \tilde P_{\text{BS}} \left ( \int_0^T  {\left \{ \int_0^t  \rho_s V_{1,s} \dd B_s \right \} } \rho_t V_{1,t} \dd B_t \right ) + \frac{1}{2} \Ex \partial_{xx} \tilde P_{\text{BS}} \left ( \int_0^T \rho_t^2 V_{1,t}^2 \dd t \right ) \\
&= \Ex \partial_{xx} \tilde P_{\text{BS}} \left ( \int_0^T \left \{{\Big ( \int_0^t \rho_s V_{1,s} \dd B_s - \int_0^t \rho_s^2 v_{0,s} V_{1,s} \dd s \Big )} \Big (\rho_t V_{1,t} \dd B_t - \rho_t^2 v_{0,t} V_{1,t}  \dd t \Big )\right \} \right )\\&\quad + \frac{1}{2} \Ex \partial_{xx} \tilde P_{\text{BS}} \left ( \int_0^T \rho^2_t V_{1,t}^2 \dd t \right ).
\end{align*}
Using \Cref{lem:mintbyparts} (Malliavin integration by parts),
\begin{align*}
C_{xx} &= -\Ex \partial_{xx} \tilde P_{\text{BS}} \left ( \int_0^T \left \{{\Big ( \int_0^t \rho_s V_{1,s} \dd B_s - \int_0^t \rho_s^2 v_{0,s} V_{1,s} \dd s \Big )} \rho_t^2 v_{0,t} V_{1,t}  \dd t \right \} \right ) \\ &\quad + \Ex \partial_{xxx} \tilde P_{\text{BS}} \left ( \int_0^T \Big ( \int_0^t \rho_s V_{1,s} \dd B_s - \int_0^t \rho_s^2 v_{0,s} V_{1,s} \dd s \Big )\rho_t^2 v_{0,t} V_{1,t} \dd t  \right ) + \frac{1}{2} \Ex \partial_{xx} \tilde P_{\text{BS}} \left ( \int_0^T \rho^2_t V_{1,t}^2 \dd t  \right ) \\
&=  \Ex ( \partial_{xxx} \tilde P_{\text{BS}} - \partial_{xx} \tilde P_{\text{BS}}) \left ( \int_0^T \Big ( \int_0^t \rho_s V_{1,s} \dd B_s - \int_0^t \rho_s^2 v_{0,s} V_{1,s} \dd s \Big )\rho_t^2 v_{0,t} V_{1,t} \dd t  \right )\\ 
&\quad + \frac{1}{2} \Ex \partial_{xx} \tilde P_{\text{BS}} \left ( \int_0^T \rho^2_t V_{1,t}^2 \dd t  \right ).
\end{align*} 
Then using \Cref{prop:partial} ($P_{\text{BS}}$ partial derivative relationship),
\begin{align*}
C_{xx} &= 2 \Ex \partial_{xy} \tilde P_{\text{BS}}  \left ( \int_0^T \Big ( \int_0^t \rho_s V_{1,s} \dd B_s - \int_0^t \rho_s^2 v_{0,s} V_{1,s} \dd s \Big )\rho_t^2 v_{0,t} V_{1,t} \dd t  \right ) \\ &\quad +\frac{1}{2} \Ex \partial_{xx} \tilde P_{\text{BS}} \left ( \int_0^T \rho^2_t V_{1,t}^2 \dd t  \right ).
\end{align*}
Adding the terms $C_x$, $C_{xx}$ and $C_y$ yields
\begin{align*}
C_x + C_{xx} + C_y &= 
\Ex \partial_y \tilde P_{\text{BS}} \left ( \int_0^T 2v_{0,t} V_{1,t} + V_{1,t}^2 + v_{0,t} V_{2,t} \dd t \right ) \\&\quad + 2 \Ex \partial_{xy} \tilde P_{\text{BS}} \left (\int_0^T \left (\int_0^t \rho_s V_{1,s} \dd B_s - \int_0^t \rho_s^2 v_{0,s} V_{1,s} \dd s \right ) \rho_t^2 v_{0,t} V_{1,t} \dd t \right ).
\end{align*} 

\subsection{$C_{xy}$}
\label{sec:Cxy}
For $C_{xy}$ we use \Cref{rem:sintbyparts} (stochastic integration by parts) to obtain 
\begin{align*}
&\Ex \partial_{xy} \tilde P_{\text{BS}}  \left ( \int_0^T \rho_t  V_{1,t} \dd B_t  - \int_0^T  \rho_t^2 v_{0,t} V_{1,t} \dd t\right )\left ( \int_0^T (1 - \rho^2_t) ( 2v_{0,t} V_{1,t}) \dd t\right ) \\&=2\Ex \partial_{xy} \tilde P_{\text{BS}}  \left ( \int_0^T \left (\int_0^t (1 - \rho_s^2) v_{0,s} V_{1,s} \dd s \right ) \left ( \rho_t V_{1,t} \dd B_t - \rho_t^2 v_{0,t} V_{1,t} \dd t\right )\right )  \\
&\quad +2 \Ex  \partial_{xy} \tilde P_{\text{BS}}   \int_0^T  \left ( \int_0^t \rho_s V_{1,s} \dd B_s -  \int_0^t \rho_s^2 v_{0,s} V_{1,s} \dd s \right ) (1-\rho_t^2) v_{0,t} V_{1,t} \dd t \\
&= 2\Ex \partial_{xy} \tilde P_{\text{BS}}  \left ( \int_0^T \left (\int_0^t (1 - \rho_s^2) v_{0,s} V_{1,s} \dd s \right ) \left ( \rho_t V_{1,t} \dd B_t - \rho_t^2 v_{0,t} V_{1,t} \dd t\right )\right )  \\
&\quad -2 \Ex  \partial_{xy} \tilde P_{\text{BS}}   \int_0^T  \left ( \int_0^t \rho_s^2 v_{0,s} V_{1,s} \dd s  \right ) v_{0,t} V_{1,t} \dd t  \\
&\quad + 2 \Ex \partial_{xy} \tilde P_{\text{BS}}  \left ( \int_0^T  \left  ( \int_0^t \rho_s V_{1,s} \dd B_s \right ) v_{0,t} V_{1,t}  \dd t\right ) \\
&\quad -  2 \Ex \partial_{xy} \tilde P_{\text{BS}} \left (\int_0^T \left (\int_0^t \rho_s V_{1,s} \dd B_s - \int_0^t \rho_s^2 v_{0,s} V_{1,s} \dd s \right ) \rho_t^2 v_{0,t} V_{1,t} \dd t \right ).
\end{align*}
Furthermore, using \Cref{prop:duality} (Malliavin duality relationship),
\begin{align*}
\hat C_{xy} &:=  2 \Ex \partial_{xy} \tilde P_{\text{BS}}  \left ( \int_0^T v_{0,t} V_{1,t} \left  ( \int_0^t \rho_s V_{1,s} \dd B_s \right ) \dd t\right ) 
&= 2 \int_0^T \Ex \partial_{xy} \tilde P_{\text{BS}} v_{0,t} V_{1,t} \left  ( \int_0^t \rho_s V_{1,s} \dd B_s \right ) \dd t \\
&= 2 \int_0^T \Ex \left ( \int_0^t \rho_s V_{1,s} D_s( \partial_{xy} \tilde P_{\text{BS}} v_{0,\cdot} V_{1,\cdot}  ) \dd s \right ) \dd t.
\end{align*}
Using the definition of the Malliavin derivative, we obtain
\begin{align*}
D_s( \partial_{xy} \tilde P_{\text{BS}} v_{0,\cdot} V_{1,\cdot}) &=  \partial_{xx y} \tilde P_{\text{BS}} v_{0,t} V_{1,t} \rho_s v_{0,s} \textbf{1}_{\{ s \leq T \} }+ \partial_{xy} \tilde P_{\text{BS}}  D_s ( v_{0, \cdot} V_{1, \cdot} )  \\
&=   \partial_{xx y} \tilde P_{\text{BS}} v_{0,t} V_{1,t} \rho_s v_{0,s} \textbf{1}_{\{ s \leq T \} }   
\\&\quad + \partial_{xy} \tilde P_{\text{BS}} v_{0,t} \left ( e^{\int_0^t \alpha_x (u, v_{0,u}) \dd u } \beta (s, v_{0,s}) e^{-\int_0^s \alpha_x (z, v_{0,z} ) \dd z } \textbf{1}_{\{ s \leq t \} }  \right ),
\end{align*}
where we have used the explicit form for $V_{1, t}$ from \cref{V1}. Thus using \Cref{prop:duality} (Malliavin duality relationship),
\begin{align*}
	&2 \int_0^T \Ex \left ( \int_0^t \rho_s V_{1,s} D_s( \partial_{xy} \tilde P_{\text{BS}} v_{0,\cdot} V_{1,\cdot}  ) \dd s \right ) \dd t \\
	&= 2  \int_0^T \Ex  \partial_{xx y} \tilde P_{\text{BS}} \left ( \int_0^t \rho_s^2 v_{0,s} V_{1, s} \dd s \right ) v_{0,t} V_{1, t} \dd t \\
	&\quad + 2  \int_0^T \Ex \partial_{xy} \tilde P_{\text{BS}}  \left ( \int_0^t \rho_s V_{1, s}   \beta(s, v_{0,s}) e^{-\int_0^s \alpha_x (z, v_{0,z}) \dd z } \dd s \right ) v_{0,t}  e^{ \int_0^t \alpha_x (z, v_{0,z}) \dd z }\dd t \\
	&= 2 \Ex  \partial_{xx y} \tilde P_{\text{BS}}  \int_0^T  \left ( \int_0^t \rho_s^2 v_{0,s} V_{1, s} \dd s \right ) v_{0,t} V_{1, t} \dd t \\
	&\quad +2 \Ex \partial_y \tilde P_{\text{BS}} \int_0^T e^{\int_0^t \alpha_x(z, v_{0,z}) \dd z } v_{0,t} \left (\int_0^t v_{0,s}^{-1} \beta(s, v_{0,s} ) V_{1,s} e^{-\int_0^s \alpha_x(z, v_{0,z} ) \dd z } \dd B_s \right )\dd t.
\end{align*}

We will now enforce \Cref{assumpbeta}. 
\begin{remark}
\label{rem:assumpbeta}
The purpose of \Cref{assumpbeta} is:
\begin{enumerate}[label = (\arabic*), ref = \arabic*]
\item $\beta(t,x) = \lambda_t x^{\mu}$ for $\mu \geq 1/2$ is H\"older continuous of order $\geq 1/2$ in $x$ uniformly in $t \in [0, T]$, and $\beta_x(t,x) = \lambda_t \mu x^{\mu -1}$ is continuous a.e. in $x$ and $t \in [0,T]$. Thus, \Cref{assregA} and \Cref{assregB} are satisfied.
\item Such a diffusion coefficient is common in application, see for example the SABR model \citep{sabr} and CEV model \citep{CEV}.
\end{enumerate}
Truthfully, we could leave $\beta$ as an arbitrary diffusion coefficient that solely obeys the conditions in \Cref{assregA} and \Cref{assregB}. However, for the purposes of application and also for our fast calibration scheme in \Cref{sec:fastcal3}, it will be more insightful to have this form for $\beta$. For the interested reader, all the following calculations still remain valid solely under \Cref{assregA} and \Cref{assregB}.
\end{remark}

Due to \Cref{assumpbeta}, we can rewrite $V_{1,t}$ and $V_{2,t}$ from \Cref{lem:Vderivativeprocesses} as 
\begin{align} 
	V_{1,t} &= e^{\int_0^t \alpha_x (z,v_{0,z}) \dd z } \int_0^t \lambda_s v_{0,s}^\mu e^{-\int_0^s \alpha_x (z, v_{0,z}) \dd z }\dd B_s, \label{v1bexp}\\
	V_{2,t} &= e^{\int_0^t \alpha_x (z,v_{0,z}) \dd z } \Bigg \{  \int_0^t \alpha_{xx} (s, v_{0,s}) (V_{1,s})^2 e^{-\int_0^s \alpha_x (z, v_{0,z}) \dd z} \dd s \nonumber \\&\quad + \int_0^t 2 \mu \lambda_s v_{0,s}^{\mu - 1} V_{1,s} e^{-\int_0^s \alpha_x (z, v_{0,z}) \dd z }\dd B_s \Bigg \}. \label{v2bexp}
\end{align} 

Then we obtain  
\begin{align*}
\hat C_{xy} &= 2 \Ex  \partial_{xx y} \tilde P_{\text{BS}}  \int_0^T  \left ( \int_0^t \rho_s^2 v_{0,s} V_{1, s} \dd s \right ) v_{0,t} V_{1, t} \dd t \\ &\quad +
2\Ex \partial_y \tilde P_{\text{BS}} \int_0^T v_{0,t}  e^{\int_0^t \alpha_x(z, v_{0,z}) \dd z } \left ( \int_0^t \lambda_s v_{0,s}^{\mu-1}   V_{1,s} e^{- \int_0^s \alpha_x(z, v_{0,z}) \dd z } \dd B_s \right )  \dd t.
\end{align*}
Hence
\begin{align*}
C_{xy} &= 2\Ex \partial_{xy} \tilde P_{\text{BS}}  \left ( \int_0^T \left (\int_0^t (1 - \rho_s^2) v_{0,s} V_{1,s} \dd s \right ) \left ( \rho_t V_{1,t} \dd B_t - \rho_t^2 v_{0,t} V_{1,t} \dd t\right )\right )  \\
&\quad -2 \Ex  \partial_{xy} \tilde P_{\text{BS}}   \int_0^T  \left ( \int_0^t \rho_s^2 v_{0,s} V_{1,s} \dd s  \right ) v_{0,t} V_{1,t} \dd t  \\
&\quad +  2 \Ex  \partial_{xx y} \tilde P_{\text{BS}}  \int_0^T  \left ( \int_0^t \rho_s^2 v_{0,s} V_{1, s} \dd s \right ) v_{0,t} V_{1, t} \dd t 
\\ &\quad +2\Ex \partial_y \tilde P_{\text{BS}} \int_0^T v_{0,t}  e^{\int_0^t \alpha_x(z, v_{0,z}) \dd z } \left ( \int_0^t \lambda_s v_{0,s}^{\mu-1}   V_{1,s} e^{- \int_0^s \alpha_x(z, v_{0,z}) \dd z } \dd B_s \right )  \dd t\\
&\quad -  2 \Ex \partial_{xy} \tilde P_{\text{BS}} \left (\int_0^T \left (\int_0^t \rho_s V_{1,s} \dd B_s - \int_0^t \rho_s^2 v_{0,s} V_{1,s} \dd s \right ) \rho_t^2 v_{0,t} V_{1,t} \dd t \right ).
\end{align*}


\subsection{$C_{yy}$}
\label{sec:Cyy}
$C_{yy}$ is given by \Cref{rem:sintbyparts} (stochastic integration by parts) as 
\begin{align*}
4 \Ex \partial_{yy} \tilde P_{\text{BS}} \left ( \int_0^T \left \{ \int_0^t (1-\rho_s^2 ) v_{0,s} V_{1,s} \dd s \right \} (1-\rho_t^2) v_{0,t} V_{1,t} \dd t \right ).
\end{align*}
\subsection{Adding $C_x, C_y, C_{xx}, C_{xy}$ and $C_{yy}$}
\label{sec:summing}
Now we add up all the terms after manipulation from \Crefrange{sec:C0}{sec:Cyy}:
\begin{align*} 
&(C_x + C_y + C_{xx}) + C_{xy} + C_{yy} \\
&=  \Ex \partial_y \tilde P_{\text{BS}} \left ( \int_0^T 2v_{0,t} V_{1,t} + V_{1,t}^2 + v_{0,t} V_{2,t} \dd t \right ) \\
&\quad +2\Ex \partial_y \tilde P_{\text{BS}} \int_0^T v_{0,t}  e^{\int_0^t \alpha_x(z, v_{0,z}) \dd z } \left ( \int_0^t \lambda_s v_{0,s}^{\mu-1}   V_{1,s} e^{- \int_0^s \alpha_x(z, v_{0,z}) \dd z } \dd B_s \right )  \dd t \\
&\quad + 2\Ex \partial_{xy} \tilde P_{\text{BS}}  \left ( \int_0^T \left (\int_0^t (1 - \rho_s^2) v_{0,s} V_{1,s} \dd s \right ) \left ( \rho_t V_{1,t} \dd B_t - \rho_t^2 v_{0,t} V_{1,t} \dd t\right )\right ) \\
&\quad - 2 \Ex  \partial_{xy} \tilde P_{\text{BS}}   \int_0^T  \left ( \int_0^t \rho_s^2 v_{0,s} V_{1,s} \dd s  \right ) v_{0,t} V_{1,t} \dd t  \\
&\quad +  2 \Ex  \partial_{xx y} \tilde P_{\text{BS}}  \int_0^T  \left ( \int_0^t \rho_s^2 v_{0,s} V_{1, s} \dd s \right ) v_{0,t} V_{1, t} \dd t  
\\ &\quad +
4 \Ex \partial_{yy} \tilde P_{\text{BS}} \left ( \int_0^T \left \{ \int_0^t (1-\rho_s^2 ) v_{0,s} V_{1,s} \dd s \right \} (1-\rho_t^2) v_{0,t} V_{1,t} \dd t \right ).
\end{align*}
Then
\begin{align*} 
&C_x + C_y + C_{xx} + C_{xy} + C_{yy} \\
	&=  \Ex \partial_y \tilde P_{\text{BS}} \left ( \int_0^T 2v_{0,t} V_{1,t} + V_{1,t}^2 +  v_{0,t} V_{2,t} \dd t \right )
\\&\quad +2\Ex \partial_y \tilde P_{\text{BS}} \int_0^T v_{0,t}  e^{\int_0^t \alpha_x(z, v_{0,z}) \dd z } \left ( \int_0^t \lambda_s v_{0,s}^{\mu-1}   V_{1,s} e^{- \int_0^s \alpha_x(z, v_{0,z}) \dd z } \dd B_s \right )  \dd t \\
&\quad + 2\Ex \partial_{xy} \tilde P_{\text{BS}}  \left ( \int_0^T \left (\int_0^t (1 - \rho_s^2) v_{0,s} V_{1,s} \dd s \right ) \left ( \rho_t V_{1,t} \dd B_t - \rho_t^2 v_{0,t} V_{1,t} \dd t\right )\right ) \\
&\quad +  4 \Ex  \partial_{yy} \tilde P_{\text{BS}}  \int_0^T  \left ( \int_0^t \rho_s^2 v_{0,s} V_{1, s} \dd s \right ) v_{0,t} V_{1, t} \dd t \\ 
&\quad + 4 \Ex \partial_{yy} \tilde P_{\text{BS}} \left ( \int_0^T \left \{ \int_0^t (1-\rho_s^2 ) v_{0,s} V_{1,s} \dd s \right \} (1-\rho_t^2) v_{0,t} V_{1,t} \dd t \right )  \\
	&= \Ex \partial_y \tilde P_{\text{BS}} \left ( \int_0^T 2v_{0,t} V_{1,t} + V_{1,t}^2 + v_{0,t} V_{2,t} \dd t \right )
\\&\quad + 2\Ex \partial_y \tilde P_{\text{BS}} \int_0^T v_{0,t}  e^{\int_0^t \alpha_x(z, v_{0,z}) \dd z } \left ( \int_0^t \lambda_s v_{0,s}^{\mu-1}   V_{1,s} e^{- \int_0^s \alpha_x(z, v_{0,z}) \dd z } \dd B_s \right )  \dd t \\
&\quad + 4\Ex \partial_{yy} \tilde P_{\text{BS}}  \left ( \int_0^T \left (\int_0^t (1 - \rho_s^2) v_{0,s} V_{1,s} \dd s \right ) \rho_t^2 v_{0,t} V_{1,t} \dd t \right ) \\
&\quad +  4 \Ex  \partial_{yy} \tilde P_{\text{BS}}  \int_0^T  \left ( \int_0^t \rho_s^2 v_{0,s} V_{1, s} \dd s \right ) v_{0,t} V_{1, t} \dd t \\ &\quad +
 4 \Ex \partial_{yy} \tilde P_{\text{BS}} \left ( \int_0^T \left \{ \int_0^t (1-\rho_s^2 ) v_{0,s} V_{1,s} \dd s \right \} (1-\rho_t^2) v_{0,t} V_{1,t} \dd t \right ),
\end{align*}
where we have used \Cref{prop:partial} ($P_{\text{BS}}$ partial derivative relationship), then \Cref{prop:duality} (Malliavin duality relationship) plus \Cref{prop:partial} ($P_{\text{BS}}$ partial derivative relationship) for the first and second equalities respectively. Combining and then simplifying the preceding $\partial_{yy}$ terms yields 
\begin{align*} 
&C_x + C_y + C_{xx} + C_{xy} + C_{yy} \\
	&= \Ex \partial_y \tilde P_{\text{BS}} \left ( \int_0^T 2v_{0,t} V_{1,t} + V_{1,t}^2 +  v_{0,t} V_{2,t} \dd t \right )
\\&\quad +2\Ex \partial_y \tilde P_{\text{BS}} \int_0^T v_{0,t}  e^{\int_0^t \alpha_x(z, v_{0,z}) \dd z } \left ( \int_0^t \lambda_s v_{0,s}^{\mu-1}   V_{1,s} e^{- \int_0^s \alpha_x(z, v_{0,z}) \dd z } \dd B_s \right )  \dd t\\
&\quad + 4\Ex \partial_{yy} \tilde P_{\text{BS}}  \left ( \int_0^T \left (\int_0^t v_{0,s} V_{1,s} \dd s \right ) v_{0,t} V_{1,t} \dd t \right ).
\end{align*}

 Lastly, notice by \Cref{rem:sintbyparts} (stochastic integration by parts) that
\begin{align*} 
 	2\left ( \int_0^T \left (\int_0^t v_{0,s} V_{1,s} \dd s \right ) v_{0,t} V_{1,t} \dd t \right ) = \left (\int_0^T v_{0,t} V_{1,t} \dd t \right )^2.
\end{align*} 

\begin{proposition}  
 \label{pricesimp}
In view of the calculations in \Crefrange{sec:C0}{sec:summing}, and under \Cref{assumpbeta}, we obtain the simpler form of the second-order approximation $\text{Put}_{\text {G}}^{(2)}$ from \Cref{thm:largeresult} as
\begin{align}
\label{eqn:pricesimp}
\begin{split}
	\text{Put}^{(2)}_{\text{G}} &= P_{\text{BS}} \left ( x_0, \int_0^T v_{0,t}^2  \dd t \right )  \\ &\quad + \Ex \partial_y \tilde P_{\text{BS}} \left ( \int_0^T 2v_{0,t} V_{1,t} + V_{1,t}^2 + v_{0,t} V_{2,t} \dd t \right ) 
	\\&\quad +2\Ex \partial_y \tilde P_{\text{BS}} \left ( \int_0^T v_{0,t} e^{\int_0^t \alpha_x(z, v_{0,z}) \dd z } \left (\int_0^t \lambda_s v_{0,s}^{\mu -1} V_{1,s} e^{-\int_0^s \alpha_x(z, v_{0,z}) \dd z } \dd B_s \right  ) \dd t \right ) \\
	&\quad + 2\Ex \partial_{yy} \tilde P_{\text{BS}}  \left ( \int_0^T v_{0,t} V_{1,t} \dd t \right )^2.
\end{split}
\end{align}
\end{proposition}

The last step is to reduce these remaining expectations in \Cref{pricesimp} down by eliminating the stochastic processes $(V_{1, t})$ and $(V_{2,t})$. To do so, we will require one more lemma, which is a consequence of \Cref{lem:detfunc}, and the new forms of $V_{1,t}$ and $V_{2,t}$ under \Cref{assumpbeta}, namely \cref{v2bexp,v1bexp}.

\begin{lemma}
\label{lem:Vred}
Consider a generic function $\ell: \reals \to \reals$ and denote by $\ell^{(k)}$ its $k$-th derivative whenever it exists. Moreover, let $\xi$ be a Lebesgue integrable deterministic function. Then the following equalities hold:
\begin{align}
	\Ex \left [\ell\!\left (\int_0^T \rho_t v_{0,t} \dd B_t \right ) \int_0^T \xi_t V_{1, t} \dd t \right ] &= \omega_{0,T}^{(-\alpha_x, \rho \lambda v_{0,\cdot}^{\mu+1}) ,(\alpha_x, \xi )}  \Ex \left [ \ell^{(1)}\!\left (\int_0^T \rho_t v_{0,t} \dd B_t \right ) \right ], \label{Vred1}
\end{align}
\begin{align}
\begin{split}
\Ex \left [\ell\!\left (\int_0^T \rho_t v_{0,t} \dd B_t \right ) \int_0^T \xi_t V_{1, t}^2 \dd t \right ] &=  2\omega_{0,T}^{(-\alpha_x, \rho \lambda v_{0, \cdot}^{\mu + 1} ), (-\alpha_x, \rho \lambda v_{0, \cdot}^{\mu + 1} ), ( 2\alpha_x, \xi ) } \Ex \left [ \ell^{(2)}\!\left (\int_0^T \rho_t v_{0,t} \dd B_t \right ) \right ] 
	\\&\quad + \omega_{0,T}^{(-2 \alpha_x, \lambda^2 v_{0, \cdot }^{2 \mu}), (2 \alpha_x, \xi) } \Ex \left [\ell\!\left (\int_0^T \rho_t v_{0,t} \dd B_t \right ) \right ], \label{Vred2}
\end{split}
\end{align}
\begin{align}
\begin{split}	
&\Ex \left [\ell\!\left (\int_0^T \rho_t v_{0,t} \dd B_t \right ) \int_0^T \xi_t V_{2,t} \dd t \right ] = \omega_{0,T}^{(- 2\alpha_x, \lambda^2 v_{0, \cdot}^{2 \mu}),( \alpha_x, \alpha_{xx}) , (\alpha_x, \xi) } \Ex \left [\ell\!\left (  \int_0^T \rho_t v_{0,t} \dd t\right ) \right ]  \\
&+ \Bigg \{ 2\omega_{0,T}^{(-\alpha_x, \rho \lambda v_{0, \cdot}^{\mu + 1}), ( -\alpha_x, \rho \lambda v_{0, \cdot}^{\mu + 1} ) , ( \alpha_x, \alpha_{xx}), (\alpha_x, \xi) } + 2\mu \omega_{0,T}^{(-\alpha_x, \rho \lambda v_{0, \cdot}^{\mu + 1} ), (0, \rho \lambda v_{0, \cdot}^{2 \mu - 1 }), (\alpha_x, \xi ) } \Bigg \} \\ &\quad \cdot \Ex \left [ \ell^{(2)}\!\left (  \int_0^T \rho_t v_{0,t} \dd B_t \right ) \right ],  \label{Vred3}
\end{split}
\end{align}
\begin{align}	
\begin{split}
\Ex \left [\ell\!\left (\int_0^T \rho_t v_{0,t} \dd B_t \right ) \left \{ \int_0^T \xi_t V_{1,t} \dd t \right \}^2 \right ] &= 2 \omega_{0,T}^{(- 2 \alpha_x, \lambda^2 v_{0,\cdot}^{2\mu}),(\alpha_x, \xi), (\alpha_x, \xi) } \Ex \left [ \ell\!\left ( \int_0^T \rho_t v_{0,t} \dd B_t \right ) \right ]  \\
&\quad +\left ( \omega_{0,T}^{(-\alpha_x, \rho \lambda v_{0, \cdot}^{\mu + 1}),(\alpha_x, \xi ) } \right )^2 \Ex \left [ \ell^{(2)}\!\left ( \int_0^T \rho_t v_{0,t} \dd B_t \right ) \right ]. \label{Vred4}
\end{split}
\end{align}
Here we write $\alpha_x := \alpha_x(\cdot, v_{0, \cdot})$ and $\alpha_{xx} := \alpha_{xx} ( \cdot, v_{0, \cdot})$ for readability.
\end{lemma}
\begin{proof}
We will only show how to obtain \cref{Vred1}. \Crefrange{Vred2}{Vred4} can be obtained in a similar way. First, we replace $V_{1,t}$ with its explicit form from \cref{v1bexp}. Thus we can write the left hand side of \cref{Vred1} as
\begin{align*}
	\Ex \left [\ell\!\left (\int_0^T \rho_t v_{0,t} \dd B_t \right ) \int_0^T \xi_t e^{\int_0^t \alpha_x(z, v_{0, z}) \dd z } \left ( \int_0^t \lambda_s v_{0,s}^\mu e^{-\int_0^s \alpha_x(z, v_{0,z}) \dd z } \dd B_s \right ) \dd t \right ].
\end{align*}
Using \Cref{lem:detfunc} with $f_t = \xi_t e^{\int_0^t \alpha_x(z, v_{0, z}) \dd z }$ and $Z_t = \int_0^t \lambda_s v_{0,s}^\mu e^{-\int_0^s \alpha_x(z, v_{0,z}) \dd z } \dd B_s$, we get
\begin{align*}
		&\Ex \left [\ell\!\left (\int_0^T \rho_t v_{0,t} \dd B_t \right ) \int_0^T \xi_t e^{\int_0^t \alpha_x(z, v_{0, z}) \dd z } \left ( \int_0^t \lambda_s v_{0,s}^\mu e^{-\int_0^s \alpha_x(z, v_{0,z}) \dd z } \dd B_s \right ) \dd t \right ]\\
		&= \Ex \left [\ell\!\left (\int_0^T \rho_t v_{0,t} \dd B_t \right ) \int_0^T \omega_{t, T}^{(\alpha_x, \xi)} \lambda_t v_{0,t}^\mu e^{-\int_0^t \alpha_x(z, v_{0,z}) \dd z } \dd B_t  \right ].
\end{align*}
Lastly, appealing to the Malliavin integration by parts \Cref{lem:mintbyparts} we obtain 
\begin{align*}
	&\Ex \left [\ell ^{(1)}\!\left (\int_0^T \rho_t v_{0,t} \dd B_t \right ) \int_0^T \omega_{t, T}^{(\alpha_x, \xi)} \rho_t  \lambda_t v_{0,t}^{\mu+1} e^{-\int_0^t \alpha_x(z, v_{0,z}) \dd z } \dd t  \right ] \\
	&=  \left (\int_0^T \omega_{t, T}^{(\alpha_x, \xi)} \rho_t  \lambda_t v_{0,t}^{\mu+1} e^{-\int_0^t \alpha_x(z, v_{0,z}) \dd z } \dd t \right ) \Ex \left [\ell ^{(1)}\!\left (\int_0^T \rho_t v_{0,t} \dd B_t \right ) \right ] \\
	&= \omega_{0,T}^{(-\alpha_x, \rho \lambda v_{0,\cdot}^{\mu+1}) ,(\alpha_x, \xi )}  \Ex \left [ \ell^{(1)}\!\left (\int_0^T \rho_t v_{0,t} \dd B_t \right ) \right ].
\end{align*}
In addition, to obtain \cref{Vred4}, notice the following integral property holds: 
\begin{align*}
	\left ( \omega_{0, T}^{(k^{(2)}, \ell^{(2)} ), (k^{(1)}, \ell^{(1)} ) } \right )^2 &= 2  \omega_{0, T}^{(k^{(2)}, \ell^{(2)} ), (k^{(1)}, \ell^{(1)} ) ,(k^{(2)}, \ell^{(2)} ), (k^{(1)}, \ell^{(1)} ) } \\
						&\quad + 4  \omega_{0, T}^{(k^{(2)}, \ell^{(2)} ), (k^{(2)}, \ell^{(2)} ) ,(k^{(1)}, \ell^{(1)} ), (k^{(1)}, \ell^{(1)} ) }.
\end{align*}
\end{proof}

With \Cref{lem:Vred} in our arsenal, we can eliminate the processes $V_{1,t}$ and $V_{2,t}$ from \Cref{pricesimp}. Using \Cref{lem:mintbyparts} (Malliavin integration by parts) and \cref{Vred1} yields
\begin{align*}
&2\Ex \partial_y \tilde P_{\text{BS}} \left ( \int_0^T v_{0,t} e^{\int_0^t \alpha_x(z, v_{0,z}) \dd z } \left (\int_0^t \lambda_s v_{0,s}^{\mu -1} V_{1,s} e^{-\int_0^s \alpha_x(z, v_{0,z}) \dd z } \dd B_s \right  ) \dd t \right ) \\
&= 2 \omega_{0,T}^{(- \alpha_x, \rho \lambda v_{0,\cdot}^{\mu + 1}), (0, \rho \lambda v_{0,\cdot}^\mu), (\alpha_x, v_{0, \cdot }) } \Ex \partial_{xxy} \tilde P_{\text{BS}}.
\end{align*}
Finally, using \Cref{lem:Vred} in a similar way on the rest of the terms in \cref{eqn:pricesimp}, we obtain the explicit second-order price given in \Cref{thm:expprice}.

\end{document}

%% file: FastcalibrationGeneralmodelexpmalliavin_JCAM.tex
\section{Fast calibration procedure}
\label{sec:fastcal3}
In this section, we present a fast calibration scheme for the Stochastic Verhulst model with time-dependent parameters, defined by
\begin{align}
	\begin{split}
	\dd S_t &= (r_t^d - r_t^f) S_t \dd t + V_t S_t \dd W_t, \quad S_0 > 0, \\
	\dd V_t &= \kappa_t (\theta_t - V_t) V_t \dd t + \lambda_t V_t \dd B_t, \quad V_0 = v_0, \label{eqn:xgbm}\\
	\dd \langle W, B \rangle_t &= \rho_t \dd t,
	\end{split}
\end{align}
where $S$ is the spot price and $V$ is the volatility. In the following, we simply refer to this model as the Verhulst model. The deterministic, time-dependent parameters $(\kappa_t)_{0 \leq t \leq T}, (\theta_t)_{0 \leq t \leq T}$ and $(\lambda_t)_{0 \leq t \leq T}$ are all assumed to be positive for all $t \in [0, T]$ and bounded. By \Cref{prop:existencemeasurexgbm}, a domestic equivalent martingale measure exists if $\rho_t \lambda_t - \kappa_t < 0  $ for all $t \in [0,T]$.

\begin{remark}[Stochastic Verhulst model heuristics]
The process $V$ from equation~\cref{eqn:xgbm} extends the deterministic Verhulst/Logistic model, which most famously arises in population growth models.\footnote{The deterministic Verhulst/Logistic model was first introduced by Verhulst in 1838 \citep{verhulst1, verhulst2, verhulst3}, then rediscovered and revived by Pearl and Reed in 1920 \citep{pearlreed1, pearlreed2}.} The Verhulst process behaves intuitively in the following way. In equation~\cref{eqn:xgbm}, the drift term of the volatility is $\kappa (\theta - V) V$, which can be interpreted as a mean-reversion to level $\theta$ at a speed of $\kappa V$. In other words, the mean reversion speed of $V$ depends on $V$ itself, and is thus stochastic. This contrasts with the regular linear mean reversion drift $\kappa (\theta - V)$, for which the mean reversion speed $\kappa$ is constant and does not depend on $V$. \citep{bakshi2006estimation} provides empirical justification for models with non-linear drift. For an in-depth discussion of the Verhulst model for option pricing, we refer the reader to the articles \citep{lewis2019exact} (where it is called the XGBM model), \citep{carr2019lognormal} (where it is called the Stochastic Logistic Model) and \citep{sepp2023log} (where it is called the Log-normal Beta stochastic volatility model with quadratic drift). 
\end{remark}

Notice that the Verhulst process $V$ from \cref{eqn:xgbm} does not satisfy \cref{assregA1}, as its drift coefficient is only locally Lipschitz continuous. However, this is not a problem as its diffusion coefficient is Lipschitz continuous. Hence, we can appeal to the usual It\^o style results on existence and uniqueness for solutions to SDEs.
\begin{proposition}
\label{prop:Verhulstprocesssolution}
Suppose $Y$ solves the SDE
\begin{align}
	\dd Y_t = a_t(b_t - Y_t) Y_t \dd t + c_t Y_t \dd B_t, \quad Y_0 = y_0>0, \label{eqn:verhulstsdesolve}
\end{align}
where $(a_t)_{0\leq t \leq T}$, $(b_t)_{0\leq t \leq T}$ and $(c_t)_{0\leq t \leq T}$ are strictly positive and bounded on $[0,T]$. Then the explicit pathwise unique strong solution is given by
\begin{align}
\begin{split}
 	Y_t &= F_t \left ( y_0^{-1} + \int_0^t a_u F_u \dd u \right )^{-1}, \\
 	F_t &= \exp \left (\int_0^t \left (a_u b_u - \frac{1}{2} c_u^2 \right ) \dd u + \int_0^t c_u \dd B_u \right ). \label{eqn:verhulstexpsolve}
\end{split}
\end{align}
\end{proposition}
\begin{proof}
The proof is similar to that of \citep[][Proposition 2.1]{carr2019lognormal}. Both the drift and diffusion coefficients in the SDE \cref{eqn:verhulstsdesolve} are locally Lipschitz, uniformly in $t \in [0,T]$. Clearly the diffusion coefficient obeys the linear growth condition, $(c_t x)^2  \leq K (1 + |x|^2)$ uniformly in $t$, for some constant $K > 0$. In addition, we have that $x [a_t(b_t -x )x ]\leq K (1 + |x|^2)$ uniformly in $t \in [0,T]$, and thus any potential of explosion in finite time is mitigated (this somewhat non-standard restriction on the growth of the drift is given in \citep[][Section 4.5, page 135]{kloeden2013numerical}). It remains to be seen that the solution is indeed given by \cref{eqn:verhulstexpsolve}. Using It\^o's formula on $Y_t$ with $f(x) = x^{-1}$ yields a linear SDE, which results in the explicit solution \cref{eqn:verhulstexpsolve}. Clearly this solution remains strictly positive in finite time.
\end{proof}

The following is the explicit second-order put option price in the Verhulst model. It is a corollary of \Cref{thm:expprice}.
\begin{corollary}[Verhulst model explicit second-order put option price] 
\label{lem:Verhulstexpprice}Under the Verhulst model \cref{eqn:xgbm}, the explicit second-order price of a put option is given by
\begin{align*}
\text{Put}^{(2)}_{\text{Verhulst}} &= P_{\text{BS}} \left ( x_0, \int_0^T v_{0,t}^2  \dd t \right )  \\
&\quad + 2 \omega_{0,T} ^{(- (\kappa \theta - 2\kappa v_{0, \cdot}), \rho \lambda v_{0, \cdot}^{2}), (\kappa \theta - 2\kappa v_{0, \cdot}, v_{0, \cdot}) } \partial_{xy} P_{\text{BS}} \left ( x_0, \int_0^T v_{0,t}^2  \dd t \right ) \\
&\quad + \omega_{0, T}^{(-2 (\kappa \theta - 2\kappa v_{0, \cdot}), \lambda^2 v_{0, \cdot}^{2  } ), ( 2(\kappa \theta - 2\kappa v_{0, \cdot}), 1)} \partial_{y} P_{\text{BS}} \left ( x_0, \int_0^T v_{0,t}^2  \dd t \right ) \\
&\quad + 2\omega_{0,T}^{(-(\kappa \theta - 2\kappa v_{0, \cdot}), \rho \lambda v_{0, \cdot }^{2}), (- (\kappa \theta - 2\kappa v_{0, \cdot}), \rho \lambda v_{0, \cdot}^{2}), (2 (\kappa \theta - 2\kappa v_{0, \cdot}),1 ) } \partial_{xxy} P_{\text{BS}} \left ( x_0, \int_0^T v_{0,t}^2  \dd t \right )  \\
&\quad + \omega_{0,T}^{(- 2(\kappa \theta - 2\kappa v_{0, \cdot}), \lambda^2 v_{0, \cdot}^{2}),( \kappa \theta - 2\kappa v_{0, \cdot}, -2\kappa) , (\kappa \theta - 2\kappa v_{0, \cdot}, v_{0, \cdot} ) }  \partial_y P_{\text{BS}} \left ( x_0, \int_0^T v_{0,t}^2  \dd t \right ) \\
&\quad + \Bigg \{ 2\omega_{0,T}^{(-(\kappa \theta - 2\kappa v_{0, \cdot}), \rho \lambda v_{0, \cdot}^{2}), ( -(\kappa \theta - 2\kappa v_{0, \cdot}), \rho \lambda v_{0, \cdot}^{2} ) , ( \kappa \theta - 2\kappa v_{0, \cdot}, -2 \kappa), (\kappa \theta - 2\kappa v_{0, \cdot}, v_{0, \cdot}) } \\&\qquad + 2  \omega_{0,T}^{(-(\kappa \theta - 2\kappa v_{0, \cdot}), \rho \lambda v_{0, \cdot}^{2} ), (0, \rho \lambda v_{0, \cdot}), (\kappa \theta - 2\kappa v_{0, \cdot} ,v_{0,\cdot}) } \Bigg \}  \partial_{xxy} P_{\text{BS}} \left ( x_0, \int_0^T v_{0,t}^2  \dd t \right ) \\
&\quad +2 \omega_{0,T}^{(- (\kappa \theta - 2\kappa v_{0, \cdot}), \rho \lambda v_{0,\cdot}^{2}), (0, \rho \lambda v_{0,\cdot}), (\kappa \theta - 2\kappa v_{0, \cdot}, v_{0, \cdot }) } \partial_{xxy} P_{\text{BS}} \left ( x_0, \int_0^T v_{0,t}^2  \dd t \right ) \\
&\quad + 4 \omega_{0,T}^{(- 2 (\kappa \theta - 2\kappa v_{0, \cdot}), \lambda^2 v_{0,\cdot}^{2}),(\kappa \theta - 2\kappa v_{0, \cdot}, v_{0,\cdot}), (\kappa \theta - 2\kappa v_{0, \cdot}, v_{0,\cdot}) }  \partial_{yy} P_{\text{BS}} \left ( x_0, \int_0^T v_{0,t}^2  \dd t \right ) \\
&\quad +2\left ( \omega_{0,T}^{(-(\kappa \theta - 2\kappa v_{0, \cdot}), \rho \lambda v_{0, \cdot}^{2}),(\kappa \theta - 2\kappa v_{0, \cdot}, v_{0,\cdot}) } \right )^2   \partial_{xxyy} P_{\text{BS}} \left ( x_0, \int_0^T v_{0,t}^2  \dd t \right ).
\end{align*}
\end{corollary}

For convenience we restate the integral operator from \Cref{defn:integraloperator},
\begin{align}
	\omega_{t, T}^{(k, \ell )} = \int_t^T \ell_u e^{\int_0^u k_z \dd z } \dd u, \label{int}
\end{align}
and its $n$-fold iterated extension 
\begin{align}
	\omega_{t,T}^{(k^{(n)}, \ell^{(n)}), (k^{(n-1)}, \ell^{(n-1)}), \dots , (k^{(1)}, \ell^{(1)})}  = \omega_{t,T}^{\big (k^{(n)}, \ell^{(n)} w_{\cdot,T} ^{(k^{(n-1)}, \ell^{(n-1)}) ,\dots, (k^{(1)}, \ell^{(1)})}\big ) }, \quad n \in \mathbb{N}. \label{intit}
\end{align}
We will refer to the functions $k, \ell, k^{(n)}, \ell^{(n)}$ as dummy functions. The rest of this section is devoted to establishing a fast calibration scheme for the Verhulst model. To do this, we recognise that the approximation of the put option price from \Cref{lem:Verhulstexpprice} is expressed in terms of iterated integral operators \cref{int,intit}. Our goal is to show that when parameters are assumed to be piecewise-constant, these iterated integral operators are closed-form, and obey a convenient recursive property.

Let $\mathcal{T} = \{ 0 = T_0, T_1, \dots, T_{N-1}, T_N = T \}$, where $T_i < T_{i+1}$ is a collection of maturity dates on $[0,T]$, with $\Delta T_i := T_{i+1} - T_i$ and $\Delta T_0 \equiv T_1$. When the dummy functions are piecewise-constant, that is, $\ell^{(n)}_t = \ell^{(n)}_i$ for $ t \in [T_{i}, T_{i+1})$ and similarly for $k^{(n)}$, we can recursively calculate the integral operators \cref{int,intit}. Consider the ODE for $(v_{0,t})$ in the Verhulst model \cref{eqn:xgbm}, 
\begin{align}
	\dd v_{0,t} &= \kappa_t ( \theta_t - v_{0,t})v_{0,t} \dd t, \quad v_{0,0} = v_0. \label{verhulstODE}
\end{align}
It is true that an explicit solution exists for this ODE \cref{verhulstODE}, namely
\begin{align*}
	v_{0,t} &= e^{\int_s^t \kappa_z \theta_z \dd z } \left ( \frac{v_{0,s}}{1 + v_{0,s} \int_s^t \kappa_u e^{\int_0^u \kappa_z \theta_z \dd z} \dd u }\right).
\end{align*}
However, the solution is a quotient, and unfortunately we cannot utilise it to make the following recursive formulas simpler, unlike in \citep{IGa}. Instead, we will compute values of $(v_{0,t})$ on a grid to approximate well any integrals involving them. We will not compute $(v_{0, t})$ over the maturity grid $\mathcal{T}$, as in practice it is quite coarse. Instead, we compute $(v_{0, t})$ on a finer grid $\mathcal{\tilde T} = \{0, \tilde T_1, \dots, \tilde T_{\tilde N - 1}, T \}$, where $\tilde T_i < \tilde T_{i+1}$ such that $\mathcal{ \tilde T} \supseteq \mathcal{T}$. Similarly define $\Delta \tilde T_i :=\tilde T_{i+1} - \tilde T_i$ with $\Delta \tilde T_0 \equiv \tilde T_1$. We now have two grids, $\mathcal{T}$ which contains the maturity dates, and $\mathcal{ \tilde T}$ which contains $\mathcal{T}$ and is where $(v_{0, t})$ is computed over. Define 
\begin{align*}
	e_t^{(k^{(n)}, \dots, k^{(1)})} &:= e^{\int_0^t \sum^n_{j=1} k_z^{(j)} \dd z}, \\
	e_{v, t}^{(h^{(n)}, \dots, h^{(1)})} &:= e^{\int_0^t v_{0,z} \sum^n_{j=1} h_z^{(j)} \dd z}, \\
	\varphi_{t, T_{i+1}}^{(k, h, p)} &:= \int_t^{T_{i+1}}\gamma_i^p(u) e^{\int_{T_i}^u k_z + h_z v_{0,z}  \dd z } \dd u,
\end{align*}
where $\gamma_i(u) :=( u - T_i)/\Delta T_i$ and $p \in \mathbb{N} \cup \{ 0 \}$. In addition, define the $n$-fold extension of $\varphi_{\cdot, \cdot}^{(\cdot, \cdot, \cdot)}$ as:
\begin{align*}
	\varphi_{t,T_{i+1}}^{(k^{(n)}, h^{(n)}, p_n), \dots,( k^{(1)}, h^{(1)}, p_1)} := \int_{t}^{T_{i+1}} &\gamma_i^{p_n} (u) e^{\int_{T_i}^u k_z^{(n)} + h_z^{(n)} v_{0,z} \dd z } \\ &\cdot \varphi_{{u}, T_{i+1}}^{(k^{(n-1)}, h^{(n-1)}, p_{n-1}), \dots, (k^{(2)}, h^{(2)}, p_2), ( k^{(1)}, h^{(1)}, p_1)} \dd u,
\end{align*}
where $p_n \in \mathbb{N} \cup \{0 \}$.

We now assume that the dummy functions are piecewise-constant on $\mathcal{T}$. However, since $(v_{0, t})$ is computed over the finer grid $\mathcal{\tilde T}$, to make this recursion simpler, we will assume that we are working on the finer grid $ \mathcal{ \tilde T}$ rather than $\mathcal{T}$, so that $v_{0, t}$ can be approximated by $v_{0, \tilde T_{\tilde i}}$ over $[\tilde T_{\tilde i}, \tilde T_{\tilde i + 1})$. Moreover, since the dummy functions are piecewise-constant on $\mathcal{T}$, then there exists an equivalent parameterisation on $\mathcal{\tilde T}$. For example, let $k_i$ be the constant value of $k$ on $[T_i, T_{i+1})$. Then there exist $\tilde T_{\tilde i}, \tilde T_{\tilde i +1}, \dots, \tilde T_{\tilde j}$ such that $\tilde T_{\tilde i} = T_i$ and $\tilde T_{\tilde j} = T_{i+1}$. Then let $\tilde k_{m} := k_i$ for $m = \tilde i, \dots , \tilde j$. Thus, without loss of generality, we can assume that we are working on $\mathcal{\tilde  T}$ and we will suppress the tilde from now on. With the assumption that the dummy functions are piecewise-constant, we can obtain the integral operator at time $T_{i+1}$ expressed by terms at $T_i$.
{
\scriptsize
\begin{align*} 
&\omega_{0,T_{i+1}}^{(k^{(1)} + h^{(1)} v_{0, \cdot}, \ell^{(1)} v_{0,\cdot}^{q_1})} \\&= \omega_{0,T_{i}}^{(k^{(1)} + h^{(1)} v_{0, \cdot}, \ell^{(1)}v_{0,\cdot}^{q_1})} + \ell_i^{(1)}  v^{q_1}_{0, T_i} e_{T_i}^{(k^{(1)})}  e_{v, T_i}^{(h^{(1)})}\varphi_{T_i, T_{i+1}}^{(k^{(1)}, h^{(1)}, 0)}, \\[.3cm]
&\omega_{0,T_{i+1}}^{(k^{(2)} + h^{(2)} v_{0, \cdot}, \ell^{(2)}v_{0,\cdot}^{q_2}), (k^{(1)} + h^{(1)} v_{0, \cdot}, \ell^{(1)}v_{0,\cdot}^{q_1})} \\&= \omega_{0,T_{i}}^{(k^{(2)} + h^{(2)} v_{0, \cdot}, \ell^{(2)}v_{0,\cdot}^{q_2}), (k^{(1)} + h^{(1)} v_{0, \cdot}, \ell^{(1)}v_{0,\cdot}^{q_1})} \\ &\quad +  \ell_i^{(1)} v^{q_1}_{0, T_i} e_{T_i}^{(k^{(1)})} e_{v, T_i}^{(h^{(1)})}\varphi_{T_i, T_{i+1}} ^{(k^{(1)}, h^{(1)}, 0)} \omega_{0,T_{i}}^{(k^{(2)} + h^{(2)} v_{0, \cdot}, \ell^{(2)}v_{0,\cdot}^{q_2})} \\&\quad + \ell_i^{(2)} \ell_i^{(1)} v^{q_2 + q_1}_{0, T_i} e_{T_i}^{(k^{(2)}, k^{(1)})} e_{v, T_i}^{( h^{(2)},h^{(1)} )} \varphi_{T_i, T_{i+1}}^{(k^{(2)}, h^{(2)}, 0) ,( k^{(1)}, h^{(1)}, 0)},  \\[.3cm]
&\omega_{0,T_{i+1}}^{(k^{(3)}+ h^{(3)}v_{0,\cdot}, \ell^{(3)}v_{0,\cdot}^{q_3}),\dots, (k^{(1)} + h^{(1)} v_{0,\cdot}, \ell^{(1)}v_{0,\cdot}^{q_1})} \\&=  \omega_{0,T_{i}}^{(k^{(3)}+ h^{(3)}v_{0,\cdot}, \ell^{(3)}v_{0,\cdot}^{q_3}),\dots, (k^{(1)} + h^{(1)} v_{0,\cdot}, \ell^{(1)}v_{0,\cdot}^{q_1})} \\&\quad + \ell_i^{(1)} v^{q_1}_{0, T_i} e_{T_i}^{(k^{(1)})} e_{v,T_i}^{(h^{(1)})} \varphi_{T_i, T_{i+1}}^{(k^{(1)}, h^{(1)}, 0)} \omega_{0,T_i} ^{(k^{(3)} + h^{(3)} v_{0, \cdot} , \ell^{(3)}v_{0,\cdot}^{q_3}), (k^{(2)} + h^{(2)} v_{0, \cdot}, \ell^{(2)}v_{0,\cdot}^{q_2})}  \\
&\quad + \ell_i^{(2)} \ell_i^{(1)} v^{q_2 + q_1}_{0, T_i} e_{T_i}^{(k^{(2)}, k^{(1)})}  e_{v, T_i}^{( h^{(2)}, h^{(1)})} \varphi_{T_i, T_{i+1}}^{(k^{(2)}, h^{(2)}, 0), ( k^{(1)}, h^{(1)}, 0)} \omega_{0,T_i}^{(k^{(3)} + h^{(3)} v_{0, \cdot}, \ell^{(3)}v_{0,\cdot}^{q_3})} \\&\quad + \ell_i^{(3)} \ell_i^{(2)} \ell_i^{(1)}  v^{q_3 + q_2 + q_1}_{0, T_i}  e_{T_i}^{(k^{(3)}, k^{(2)}, k^{(1)})}  e_{v,T_i}^{( h^{(3)}, h^{(2)}, h^{(1)})}\varphi_{T_i, T_{i+1}}^{(k^{(3)}, h^{(3)}, 0), ( k^{(2)}, h^{(2)},  0), ( k^{(1)}, h^{(1)},  0)}, \\[.3cm]
&\omega_{0,T_{i+1}}^{(k^{(4)} + h^{(4)} v_{0, \cdot}, \ell^{(4)}v_{0,\cdot}^{q_4}), \dots , (k^{(1)} + h^{(1)} v_{0, \cdot}, \ell^{(1)}v_{0,\cdot}^{q_1})} \\&=  \omega_{0,T_{i}}^{(k^{(4)} + h^{(4)} v_{0, \cdot}, \ell^{(4)}v_{0,\cdot}^{q_4}),\dots, (k^{(1)} + h^{(1)} v_{0, \cdot} , \ell^{(1)}v_{0,\cdot}^{q_1})} \\&\quad + \ell_i^{(1)}  v^{q_1}_{0, T_i}  e_{T_i}^{(k^{(1)})}  e_{v,T_i}^{(h^{(1)})}\varphi_{T_i, T_{i+1}}^{(k^{(1)}, h^{(1)},0)} \omega_{0,T_i} ^{(k^{(4)} + h^{(4)} v_{0, \cdot}, \ell^{(4)}v_{0,\cdot}^{q_4}), (k^{(3)} + h^{(3)} v_{0, \cdot} , \ell^{(3)}v_{0,\cdot}^{q_3}), (k^{(2)} + h^{(2)} v_{0, \cdot}, \ell^{(2)}v_{0,\cdot}^{q_2})} \\
&\quad + \ell_i^{(2)} \ell_i^{(1)} v^{q_2 + q_1}_{0, T_i} e_{T_i}^{(k^{(2)}, k^{(1)})} e_{v,T_i}^{(h^{(2)}, h^{(1)})}  \varphi_{T_i, T_{i+1}}^{(k^{(2)}, h^{(2)}, 0), ( k^{(1)}, h^{(1)}, 0)} \omega_{0,T_i}^{(k^{(4)} + h^{(4)} v_{0, \cdot}, \ell^{(4)}v_{0,\cdot}^{q_4}),(k^{(3)} + h^{(3)} v_{0, \cdot}, \ell^{(3)}v_{0,\cdot}^{q_3})} \\&\quad + \ell_i^{(3)} \ell_i^{(2)} \ell_i^{(1)}  v^{q_3 + q_2 + q_1}_{0, T_i} e_{T_i}^{(k^{(3)}, k^{(2)}, k^{(1)})} e_{v,T_i}^{( h^{(3)}, h^{(2)}, h^{(1)})}  \varphi_{T_i, T_{i+1}}^{(k^{(3)}, h^{(3)}, 0), (k^{(2)}, h^{(2)}, 0), ( k^{(1)}, h^{(1)}, 0)} \omega_{0,T_i}^{(k^{(4)} + h^{(4)} v_{0, \cdot} , \ell^{(4)}v_{0,\cdot}^{q_4})} \\&\quad + \ell_i^{(4)} \ell_i^{(3)} \ell_i^{(2)} \ell_i^{(1)} v^{q_4 + q_3 + q_2 + q_1}_{0, T_i} e_{T_i}^{(k^{(4)}, k^{(3)}, k^{(2)}, k^{(1)})}  e_{v,T_i}^{( h^{(4)}, h^{(3)}, h^{(2)}, h^{(1)})}\varphi_{T_i, T_{i+1}}^{(k^{(4)}, h^{(4)}, 0),(k^{(3)}, h^{(3)}, 0), (k^{(2)}, h^{(2)}, 0),( k^{(1)}, h^{(1)}, 0)}.
\end{align*}
}

The only terms here that are not closed-form are the functions $e_\cdot^{(\cdot, \dots , \cdot)}, e_{v, \cdot}^{(\cdot, \dots , \cdot)}$ and $\varphi_{t, T_{i+1}}^{(\cdot, \cdot, \cdot), \dots , (\cdot, \cdot, \cdot)}$. \\For $t \in (T_i, T_{i+1}]$, we can derive the following: 
\begin{align*}
e_t^{(k^{(n)}, \dots , k^{(1)})} &= 
	e_{T_i}^{(k^{(n)}, \dots , k^{(1)})} e^{ \Delta T_i \gamma_i(t) \sum_{j=1}^n k_i^{(j)}} = e^{\sum_{m=0}^{i-1}\Delta T_m \sum_{j=1}^n k_m^{(j)} } e^{ \Delta T_i \gamma_i(t) \sum_{j=1}^n k_i^{(j)}}, \\
	e_{v,t}^{(h^{(n)}, \dots , h^{(1)})} &= 
	e_{v,T_i}^{(h^{(n)}, \dots , h^{(1)})} e^{ \Delta T_i \gamma_i(t) v_{0, T_i} \sum_{j=1}^n h_i^{(j)}}= e^{\sum_{m=0}^{i-1}\Delta T_m  v_{0, T_m}\sum_{j=1}^n h_m^{(j)} } e^{ \Delta T_i \gamma_i(t) v_{0, T_i} \sum_{j=1}^n h_i^{(j)}},
\end{align*}
where $e_0^{(k^{(n)}, \dots, k^{(1)})} = 1$ and $e_{v, 0}^{(h^{(n)}, \dots, h^{(1)})} = 1$.

Let $\tilde k_i := k_i + h_i v_{0, T_i}$ and $\tilde k_i^{(n)} := k_i^{(n)} + h_i^{(n)}v_{0, T_i} $. Then
\begin{align*}
 \varphi_{t, T_{i+1}}^{(k, h, p)} &= 
 	\begin{cases} 
	\frac{1}{\tilde k_i} \left (e^{\tilde k_i \Delta T_i } - \gamma_i^p(t) e^{\tilde k_i \Delta T_i \gamma_i(t) } - \frac{p}{\Delta T_i} \varphi_{t, T_{i+1}}^{(k,h, p-1)} \right ), \quad &\tilde k_i \neq 0, p \geq 1,\\[.25cm] 
 	\frac{1}{\tilde k_i} \left ( e^{ \tilde k_i \Delta T_i } -e^{\tilde k_i \Delta T_i \gamma_i(t) }  \right ), \quad &\tilde k_i \neq 0, p =0, \\[.25cm] 
 	\frac{1}{p+1} \Delta T_i \left( 1-  \gamma_i^{p+1}(t) \right )  &\tilde k_i = 0, p \geq 0.
 	\end{cases}
\end{align*}
In addition, for $n \geq 2$,
\begin{align*}
&\varphi_{t, T_{i+1}}^{(k^{(n)}, h^{(n)}, p_n),  \dots, (k^{(1)}, h^{(1)}, p_1)} = \\
	&\begin{cases} 
	\frac{1}{\tilde k_i^{(n)}} \Big ( \varphi_{t, T_{i+1}}^{(k^{(n)} + k^{(n-1)}, h^{(n)} + h^{(n-1)}, p_n + p_{n-1}),  (k^{(n-2)}, h^{(n-2)}, p_{n-2} ), \dots, (k^{(1)}, h^{(1)}, p_1)}\\-  \frac{p_n}{\Delta T_i} \varphi_{t, T_{i+1}}^{(k^{(n)}, h^{(n)}, p_n-1), (k^{(n-1)}, h^{(n-1)},  p_{n-1}), \dots, (k^{(1)}, h^{(1)}, p_1)}  \\- \gamma_i^{p_n}(t) e^{ \tilde k_i^{(n)} \Delta T_i \gamma_i(t)} \varphi_{t, T_{i+1}}^{(k^{(n-1)}, h^{(n-1)}, p_{n-1}), \dots,( k^{(1)}, h^{(1)}, p_1)} \Big ), &\tilde k_i^{(n)} \neq 0 , p_n \geq 1, \\[.5cm]
	\frac{1}{\tilde k_i^{(n)}} \Big (\varphi_{t, T_{i+1}}^{( k^{(n)} + k^{(n-1)}, h^{(n)} + h^{(n-1)}, p_{n-1}), (k^{(n-2)}, h^{(n-2)}, p_{n-2}), \dots, (k^{(1)}, h^{(1)}, p_1)} \\- e^{ \tilde k_i^{(n)} \Delta T_i \gamma_i(t)} \varphi_{t, T_{i+1}}^{(k^{(n-1)}, h^{(n-1)}, p_{n-1}), \dots,( k^{(1)}, h^{(1)}, p_1)}    \Big ), 		&\tilde k_i^{(n)}, \neq 0, p_n =0, \\[.5cm]
	\frac{\Delta T_i }{p_n + 1} \Big (\varphi_{t, T_{i+1}} ^{(k^{(n-1)}, h^{(n-1)}, p_n + p_{n-1} + 1), (k^{(n-2)}, h^{(n-2)},p_{n-2}), \dots , (k^{(1)}, h^{(1)}, p_1)} \\-  \gamma_i^{p_n +1 } (t)  \varphi_{t, T_{i+1}}^{(k^{(n-1)}, h^{(n-1)}, p_{n-1}), \dots,( k^{(1)}, h^{(1)}, p_1)}\Big ), & \tilde k_i^{(n)} = 0, p_n \geq 0.
\end{cases}
\end{align*}

\begin{remark}[Fast calibration scheme]
These formulas make it possible to calibrate the piecewise-constant parameters of the stochastic volatility model to market data in an efficient way, as described in the following algorithm. Let $(\mu_t) \equiv \mu = (\mu^{(1)}, \mu^{(2)}, \dots, \mu^{(n)})$ be an arbitrary set of parameters and denote by $\omega_t$ an arbitrary integral operator.
\begin{itemize}
\item Calibrate $\mu$ over $[0,T_1)$ to obtain $\mu_0$. This involves computing $\omega_{T_1}$.
\item Calibrate $\mu$ over $[T_1, T_2)$ to obtain $\mu_1$. This involves computing $\omega_{T_2}$ which is expressed in terms of $\omega_{T_1}$, the latter already being computed in the previous step.
\item Repeat until time $T_N$.
\end{itemize}
\end{remark}

%% file: NumericsGeneralmodelexpmalliavin_JCAM.tex
\section{Numerical tests and sensitivity analysis}
\label{sec:numerical3}
\begin{remark}
The language and methodology in this section is not dissimilar to that of \citep{das2018closedform}. This is because the numerical sensitivity analysis methodology carried out here is essentially the same, and the main difference is that we utilise our approximation method established in this article rather than theirs.
\end{remark}

In this section, we will numerically investigate the accuracy of our closed-form approximation formula in the Stochastic Verhulst model (\Cref{lem:Verhulstexpprice}). For an arbitrary set of piecewise-constant parameters $(\kappa_t, \theta_t, \lambda_t , \rho_t) \equiv (\kappa, \theta, \lambda, \rho) =:\mu $, we will perform a sensitivity analysis by varying one of the parameters $\kappa$, $\theta$, $\lambda$, $\rho$ at a time while keeping the rest fixed. Then, we will compute the difference (signed error) in implied volatilities between our approximation formula and a Monte-Carlo estimator for maturity times $T \in \{1/12, 3/12, 6/12, 1\} \equiv \{1\text{M}, 3\text{M}, 6\text{M}, 1\text{Y}\}$ and strikes corresponding to Put 10, 25, and ATM deltas. Thus, the signed error of the implied volatility for a given parameter set $\mu$, maturity $T$, and strike $K$ is
\begin{align*}
	\text{Error}(\mu, T, K) = \sigma_{\text{IM-Approx}}(\mu, T, K)  - \sigma_{\text{IM-Monte}}(\mu, T, K).
\end{align*}
Additionally, in this section we also compare the computational run time of our closed-form approximation formula with the Monte-Carlo simulation approach.

The Monte-Carlo simulation takes advantage of the mixing solution methodology from \Cref{appen:mixingsol}. Using this relationship for Monte-Carlo simulation, there is no need to simulate $S$, one need only simulate $V$. This reduces the run time as well as the standard error of the procedure.

We will employ the classical Euler-Maruyama method to simulate the volatility process. To reduce the Monte-Carlo and discretisation errors sufficiently well, we use 10,000,000 Monte Carlo paths and 24 time steps per day in all our tests, where a year is comprised of 252 trading days.

\begin{remark}
The code utilised to obtain the numerical results in this section is available on GitHub \citep{das2023}. In particular, what is provided is: 
\begin{itemize}
\item A routine which computes our closed-form approximation of put option prices for the Stochastic Verhulst model with piecewise-constant parameter inputs.
\item A routine which implements the Monte-Carlo simulation via the mixing solution methodology for the pricing of put option prices in the Stochastic Verhulst model with piecewise-constant parameter inputs.
\item A routine which compares the accuracy and runtimes of the aforementioned methods.
\end{itemize}
\end{remark}

\subsection{Stochastic Verhulst model sensitivity analysis}

\begin{definition}
A piecewise-constant parameter which is piecewise-constant over $n$ intervals will be called an $n$-piece parameter or a parameter with $n$ pieces.
\end{definition}

\label{verhulstsense} We start from a `safe' parameter set, given by:

\begin{center}
\begin{tabular}{llll} 
\toprule
$S_0$ & \ $v_0$ & $r^d$ & $r^f$ \\
\midrule
$100$ & \ $18\%$ &  2\% & 0 \\
 \bottomrule \\
\end{tabular}
\end{center}
with


\begin{center}
\begin{tabular}{l c cccc}
\toprule

 $T$	 && 		$\kappa$ & $\theta$ & $\lambda$ & $\rho$   \\
\midrule
1M   && 		$5.00$   & 1.70\%   & 0.414  & -0.391  \\
  
3M   && 		$5.00$  & 1.70\%   & 0.414  & -0.391 \\

6M   &&	 	$5.00$   & 1.70\%   & 0.414  & -0.391  \\

1Y   && 		$5.00$   & 1.70\%   & 0.414  & -0.391  \\
\bottomrule
\end{tabular}
\end{center}


However, we would like to consider piecewise-constant parameter inputs, as our closed-form approximation method has been designed to take advantage of them. To do so, we will make our parameters piecewise-constant over three time intervals, the length of the first, second and third time interval having proportions $1/4$, $1/4$, $1/2$ of the maturity time $T$ respectively. For example, if $T = 1/12$ then a 3-piece parameter is piecewise-constant on the time intervals $[0, 1/48), [1/48, 2/48), [2/48, 4/48)$. To choose the `safe' values of the 3-piece parameters, we simply perturb the value of each safe parameter over each time interval. This yields the following table of `safe' 3-piece parameter sets:

\begin{center}
\begin{tabular}{lcc c cccc cc}
\toprule
$T$	 & Piece & Proportion	&& 		$\kappa$ & $\theta$ & $\lambda$ & $\rho$ & 	$r^d$ & $r^f$   \\

\midrule
1M   		& 1 & 1/4 		&&	 	4.80   & 1.70\%  & 0.394 & -0.371 & 			1\% & 0 \\
   		& 2 & 1/4 	   	&&	 	5.20   & 2.10\%  & 0.434 & -0.411 &			3\% & 0 \\
   		& 3 & 1/2 		&&	 	5.00   & 1.90\%  & 0.414 & -0.391 & 			2\% & 0 \\

\midrule
3M   		& 1 & 1/4  		&&	 	4.80   & 1.70\%  & 0.394 & -0.371 &			1\% & 0   \\
   		& 2 & 1/4	   	&&	 	5.20   & 2.10\%  & 0.434 & -0.411 & 			3\% & 0   \\
   		& 3 & 1/2   		&&	 	5.00   & 1.90\%  & 0.414 & -0.391 & 			2\% & 0  \\

\midrule
6M   		& 1 & 1/4   		&&	 	4.80   & 1.70\%  & 0.394 & -0.371 & 			1\% & 0	\\
   		& 2 & 1/4	   	&&	 	5.20   & 2.10\%  & 0.434 & -0.411 & 			3\% & 0 	 \\
   		& 3 & 1/2 		&&	 	5.00   & 1.90\%  & 0.414 & -0.391 & 			2\% & 0 	 \\
    
\midrule
1Y   		& 1 & 1/4  		&&	 	4.80   & 1.70\%  & 0.394 & -0.371 &  			1\% & 0 	\\
   		& 2 & 1/4	   	&&	 	5.20   & 2.10\%  & 0.434 & -0.411 & 			3\% & 0	 \\
   		& 3 & 1/2   		&&	 	5.00   & 1.90\%  & 0.414 & -0.391 & 			2\% & 0	 \\

\bottomrule
\end{tabular}
\end{center}

When computing our closed-form approximation formula, we utilise 27 points in the grid for approximating integrals involving $(v_{0,t})$. In other words, the set $\tilde{\mathcal{T}}$ has 27 elements. Therefore, despite using 3-piece parameters, they end up being 26-piece parameters. We have done this to reduce discretisation error as much as possible whilst retaining a reasonable run time. Indeed, 27 happened to be the value such that, beyond this, our standard windows machine started to struggle. Coarser grids may be acceptable for the purposes of application, however our intention here is to showcase the power of our approximation formula.

In our numerical analysis, we vary one of the 3-piece parameters $(\kappa, \theta, \lambda, \rho)$ with the rest fixed, and then compute implied volatilities via both the closed-form approximation formula as well as the Monte-Carlo method as described above. Specifically, we select a 3-piece parameter from $(\kappa, \theta, \lambda, \rho)$, and start at 40\% of its safe 3-piece parameter value, then increase the value of each piece in increments of 20\%, all the way up to 160\%, whilst keeping the other three 3-piece parameters fixed at their safe value. We then repeat this process with each of the other three 3-piece parameters. The relevant tables are \Cref{table:varyingkappaverhulst}, \Cref{table:varyingthetaverhulst}, \Cref{table:varyinglambdaverhulst}, and \Cref{table:varyingrhoverhulst}, for the analysis of $\kappa$, $\theta$, $\lambda$, and $\rho$ respectively. Note that the $\rho$ table values are reversed, as $\rho$ is negative, and thus 160\% of the safe 3-piece $\rho$ is approximately $(-0.594, -0.658, -0.626)$. This ensures that the parameter values are increasing for all tables.

The sensitivity analysis fares well, with signed errors in implied volatility less than 15bp for reasonable parameter values (namely the 100\% columns). Indeed such errors are acceptable for use within applications. For more extreme parameter values, such as the 160\% columns, the errors are larger but do not exceed 30bp. For example, investigating the $\lambda$ table (\Cref{table:varyinglambdaverhulst}) which corresponds to vol-of-vol, we see signed errors in implied volatility increasing with respect to it. In the worst case setting, this is approximately 30bp for a 1Y Put 10 put option. For the analysis of the $\rho$ table (\Cref{table:varyingrhoverhulst}) which corresponds to correlation, the signed errors are sometimes increasing, but sometimes decreasing. This is not unusual, as it suggests that eventually the absolute error will increase. In general, we notice that signed errors behave monotonically with respect to all parameter adjustments, suggesting that absolute errors eventually increase as we get further away from the `safe' parameters and towards more extreme parameter values.

\Cref{fig:verhulsterror} indicates the signed error in implied volatility versus maturity and moneyness, where we keep the 3-piece parameters fixed at their `safe' value. With these `safe' parameters, the signed errors in implied volatility do not exceed 18bp. The plots demonstrate that signed errors in implied volatility monotonically increase or decrease as the put option becomes more out of the money. This is intuitively expected and suggests that absolute errors will eventually increase with respect to these parameters. A peculiarity occurs regarding the behaviour with respect to maturity. From first glance, it seems that signed errors in implied volatility do not necessarily increase or decrease monotonically with respect to maturity. However, we have to take into account that for each different maturity value, the parameters being utilised are not constant, but piecewise-constant over the respective 3-piece. More precisely, parameters are piecewise-constant over 3-pieces with lengths $[0, T/4), [T/4, 2T/4), [2T/4, T)$ for any maturity $T$. This explains the apparent non-monoticity of the signed error in implied volatility with respect to maturity --- the model parameters being inputted are essentially different for each $T$ value.

\begin{figure}[h!]
    \centering
    \begin{subfigure}{0.48\textwidth}
        \centering
        \includegraphics[width=\linewidth]{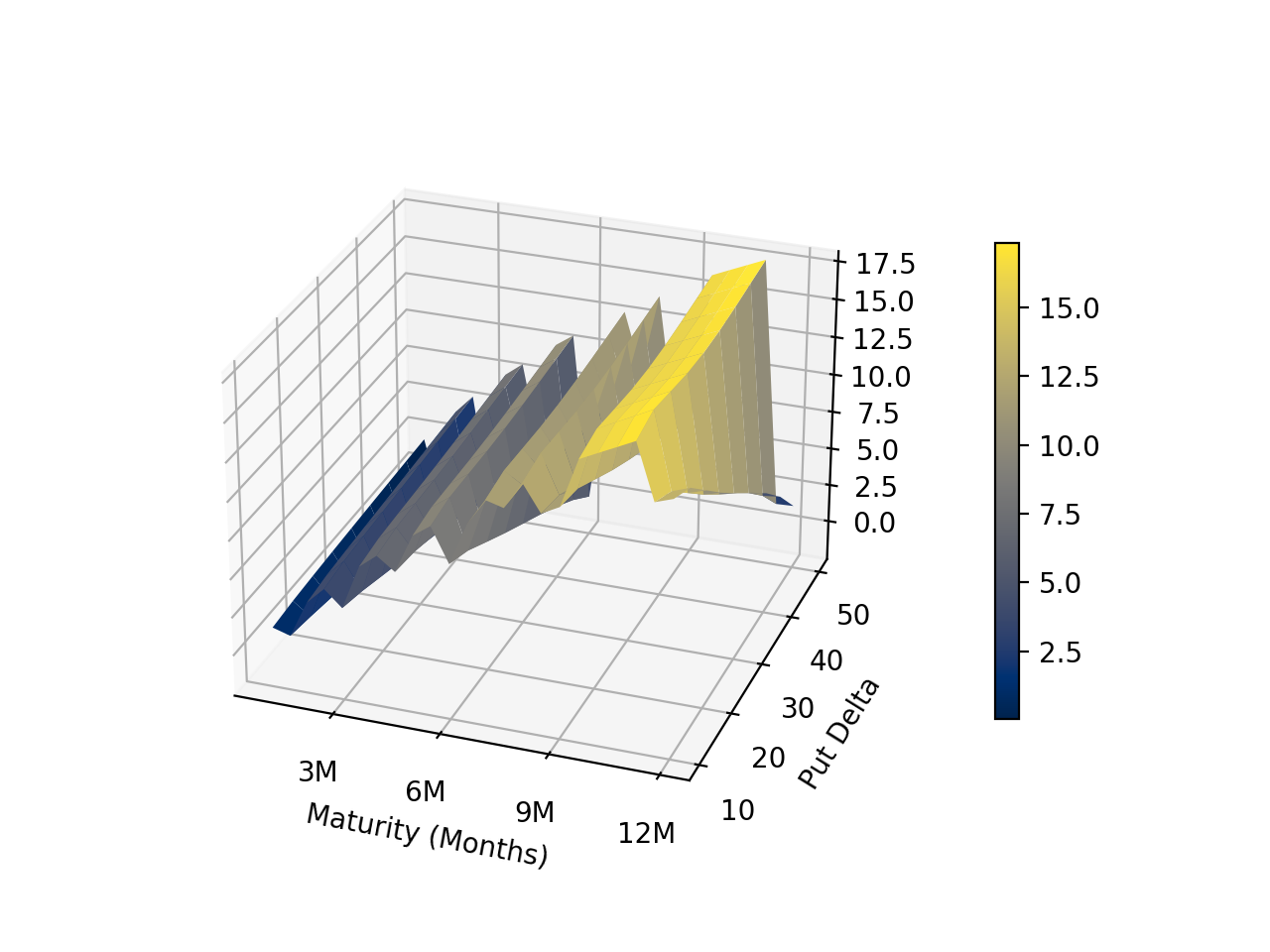}
        \caption{Surface plot of signed error in implied volatility.}
        \label{fig:verhulsterror-surf}
    \end{subfigure}
    \hfill
    \begin{subfigure}{0.45\textwidth}
        \centering
        \includegraphics[width=\linewidth]{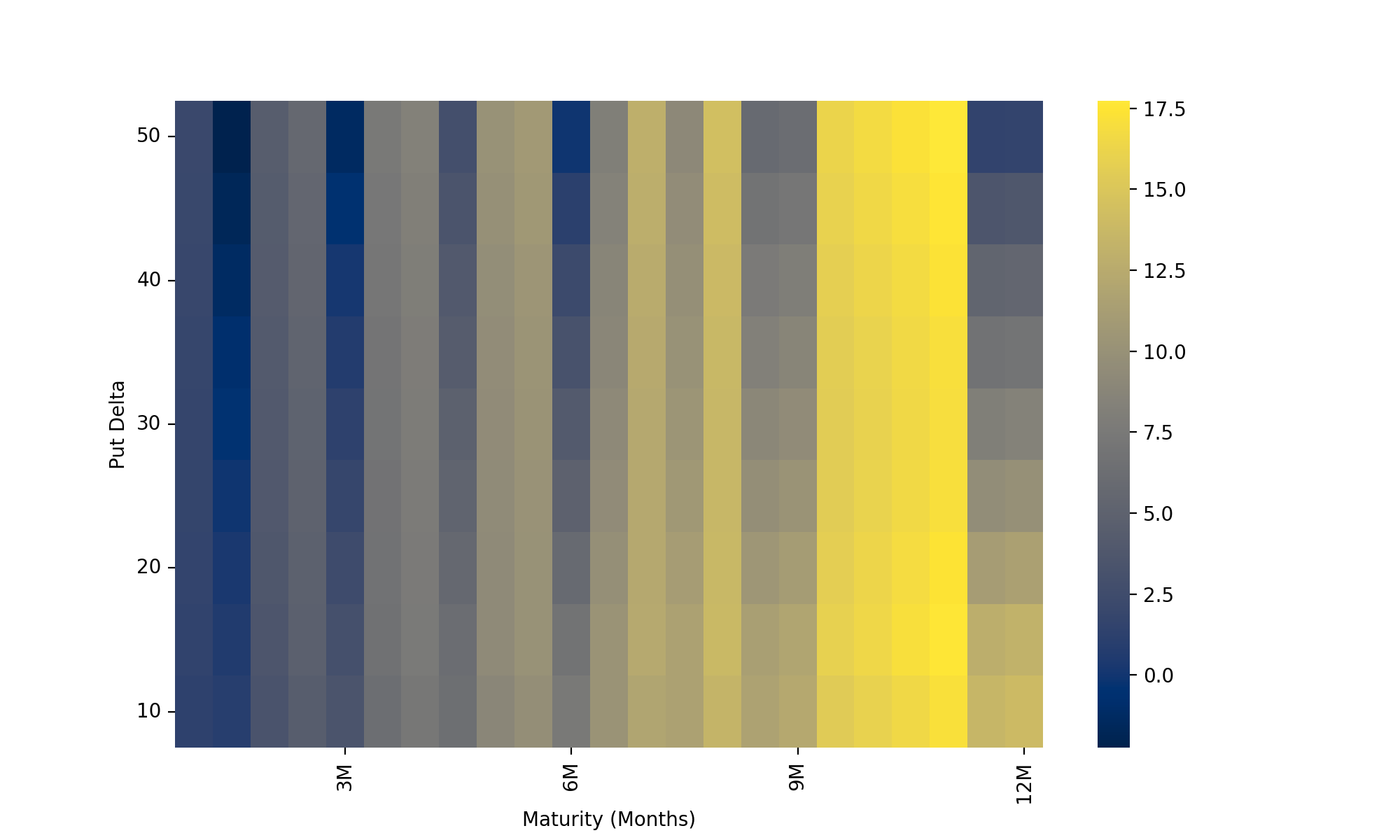}
        \caption{Heatmap of signed error in implied volatility.}
        \label{fig:verhulsterror-heat}
    \end{subfigure}
    \caption{Plots of signed error in implied volatility versus Maturity and Put Delta in the Stochastic Verhulst model, where the parameters are kept at their `safe' 3-piece values.}
    \label{fig:verhulsterror}
\end{figure}

\subsection{Run time comparison against Monte-Carlo simulation}
In the previous subsection, we investigated how our closed-form approximation formula behaves with respect to parameter adjustments. We accomplished this by comparing it against the benchmark price, which was obtained via a Monte-Carlo simulation with 10,000,000 paths and therefore has an extremely low standard error. However, it is not productive to compare the performance of these two methods, as their run times are completely different --- on average 1 second for the closed-form approximation, and $10^3$ seconds for the benchmark Monte-Carlo. 

It is natural to ask whether our closed-form approximation formula is faster than a Monte-Carlo simulation whose standard error is comparable to the signed error of the closed-form approximation formula. To answer this, we need to deduce the number of paths $N$ for a Monte-Carlo simulation so that its standard error corresponds to the signed error of the closed-form approximation formula. We will call a Monte-Carlo simulation with such number of paths $N$ a `fair Monte-Carlo simulation', to distinguish it from the very accurate but very slow benchmark Monte-Carlo simulation. 

We choose the number of paths for the fair Monte-Carlo simulation as
\begin{align}
\label{eqn:NvalueMC}
    N = \frac{10^7 (\text{StdE})^2}{(\text{SE})^2}\times 1.96^2,
\end{align}
where StdE refers to the benchmark Monte-Carlo simulation price standard error, and SE refers to the signed error in the closed-form approximation formula with the benchmark Monte-Carlo price. We stress that these error quantities do not correspond to implied volatility, unlike in \Crefrange{table:varyingkappaverhulst}{table:varyingrhoverhulst} from the sensitivity analysis. The number $10^7$ refers to the fact that the benchmark Monte-Carlo simulation is implemented with that many paths. Some simple statistical analysis demonstrates that this choice of $N$ ensures that $1.96 \; \times$ the Monte-Carlo standard error is smaller than the absolute error of the closed-form approximation formula 95\% of the time.

Plugging in the relevant StdE and SE into \cref{eqn:NvalueMC}, we obtain the number of paths which will be used in the fair Monte-Carlo simulation. Then, we compute the standard errors of the fair Monte-Carlo simulations, their run times and the run times of our corresponding closed-form approximation formula. We also display the signed error in price of our closed-form approximation formula with respect to the benchmark Monte-Carlo simulation. We focus on using the `safe' parameters, as these are reasonable parameters for application, and our interest here is in run times. The relevant table is \Cref{table:runtimeverhulst}.

The numerical results in \Cref{table:runtimeverhulst} demonstrate that our closed-form approximation formula outperforms the fair Monte-Carlo simulation in all cases bar the 1Y Delta 10 put option. This is as expected, as the error in our closed-form approximation formula is larger for out of the money put options with large maturities. Nonetheless, for more reasonable choices of Delta and maturity, our closed-form approximation formula outperforms the fair Monte-Carlo simulation regarding run time. We should reiterate that our numerical experiments are performed when the function $(v_{0, t})$ is discretised over a grid with $27$ points --- in other words, the grid $\mathcal{\tilde T}$ has 26 points. We found this choice to be the best in terms of tradeoff between run time and error. 

\input{TablesGeneralmodelexpmalliavin_JCAM_revision}

%% file: TablesGeneralmodelexpmalliavin_JCAM_revision.tex

\label{sec:varyingkappaverhulst}
\;

\begin{table}[H]
\caption{Approximated Implied Volatilities in \%, signed error of implied volatilities (Approximation - Monte Carlo) in basis points, and standard errors of the benchmark Monte-Carlo simulation (IV(\%), SE(bp), StdE(bp) respectively), computed in the Stochastic Verhulst model with $\kappa$ varying from 40\% to 160\% of its `safe' 3-piece value.}
\label{table:varyingkappaverhulst}
\begin{center}
\begin{tabular}{lc l rrrrrrr}
\toprule
& & & \multicolumn{7}{c}{scaling factor for $\kappa$ safe value} \\
\cmidrule{1-2}\cmidrule{4-10} 
Moneyness&	$T$ & IV/SE/StdE	&  \ 40\% & 60\%  & 80\%  & 100\% & 120\% & 140\% & 160\% \\
\cmidrule{1-2}\cmidrule{4-10}  

 ATM	   & {1M} 	&IV{\footnotesize\,(\%)}	   & 17.76 		&17.65 		&17.54  		&17.44 	&17.33 	&17.23 	&17.13 \\
		& 	 	 &SE{\footnotesize\,(bp)}		& \red{0.64}		& \red{1.07}		&\red{1.50} 		&\red{1.91}	&\red{2.32}	&\red{2.72}	&\red{3.11} \\
        & 	 	 &StdE{\footnotesize\,(bp)}		& \blue{0.32}		&\blue{0.32}		&\blue{0.31} 		&\blue{0.31}	&\blue{0.31}	&\blue{0.31}	&\blue{0.31} \\
 		& {3M}   &IV{\footnotesize\,(\%)}		& 17.24		&16.93		&16.64		&16.37	&16.11	&15.86	&15.63 \\
 		& 	  	 &SE{\footnotesize\,(bp)}		& \red{-4.64}		&\red{-3.61}		&\red{-2.63}		&\red{-1.71}	&\red{-0.83}	&\red{0.00}	&\red{0.80} \\
        & 	 	 &StdE{\footnotesize\,(bp)}		& \blue{0.34}		&\blue{0.33}		&\blue{0.32} 		&\blue{0.32}	&\blue{0.31}	&\blue{0.30}	&\blue{0.30} \\
 		& {6M}   &IV{\footnotesize\,(\%)}		& 16.58		&16.04		&15.56		&15.12	&14.72	&14.35	&14.02 \\
 		& 	  	 &SE{\footnotesize\,(bp)}		& \red{-5.74}		&\red{-3.95}		&\red{-2.33}		&\red{-0.85}	&\red{0.51}	&\red{1.77}	&\red{2.95} \\
        & 	 	 &StdE{\footnotesize\,(bp)}		& \blue{0.35}		&\blue{0.33}		&\blue{0.32} 		&\blue{0.31}	&\blue{0.30}	&\blue{0.29}	&\blue{0.28} \\
 		& {1Y}   &IV{\footnotesize\,(\%)}		& 15.45		&14.60		&13.89		&13.29	&12.77	&12.31	&11.90 \\
 		& 	  	 &SE{\footnotesize\,(bp)}		& \red{-6.76}		&\red{-4.07}		&\red{-1.73}		&\red{0.35}	&\red{2.22}	&\red{3.92}	&\red{5.49} \\
        & 	 	 &StdE{\footnotesize\,(bp)}		& \blue{0.35}		&\blue{0.33}		&\blue{0.31} 		&\blue{0.29}	&\blue{0.28}	&\blue{0.27}	&\blue{0.26} \\
 \midrule

 Put 25 & {1M}  	&IV{\footnotesize\,(\%)}	   & 18.05		&17.94		&17.83		&17.72		&17.61		&17.51		&17.40 \\
		& 	  	 &SE{\footnotesize\,(bp)}		& \red{0.53}		&\red{0.97}		&\red{1.39}		&\red{1.79}		&\red{2.19}		&\red{2.58}		&\red{2.96} \\
        & 	 	 &StdE{\footnotesize\,(bp)}		& \blue{0.22}		&\blue{0.22}		&\blue{0.22} 		&\blue{0.22}	&\blue{0.22}	&\blue{0.21}	&\blue{0.21} \\
 		& {3M}   &IV{\footnotesize\,(\%)}		& 17.78		& 17.46		&17.16		&16.88		&16.61		&16.35		&16.11 \\
 		& 	  	 &SE{\footnotesize\,(bp)}		& \red{-0.99}		&\red{0.03}		&\red{0.99}		&\red{1.89}		&\red{2.75}		&\red{3.56}		&\red{4.33} \\
        & 	 	 &StdE{\footnotesize\,(bp)}		& \blue{0.25}		&\blue{0.24}		&\blue{0.23} 		&\blue{0.23}	&\blue{0.22}	&\blue{0.22}	&\blue{0.21} \\
		& {6M}   &IV{\footnotesize\,(\%)}		& 17.33		&16.77		&16.26		&15.80		&15.38		&14.99		&14.64 \\
		& 	  	 &SE{\footnotesize\,(bp)}		& \red{0.14}		&\red{1.85}		&\red{3.38}		&\red{4.78}		&\red{6.06}		&\red{7.26}		&\red{8.37} \\
        & 	 	 &StdE{\footnotesize\,(bp)}		& \blue{0.26}		&\blue{0.25}		&\blue{0.24} 		&\blue{0.22}	&\blue{0.22}	&\blue{0.21}	&\blue{0.20} \\
		& {1Y}   &IV{\footnotesize\,(\%)}		& 16.46		&15.55		&14.80		&14.15		&13.59		&13.10		&12.66 \\
		& 	   	 &SE{\footnotesize\,(bp)}		& \red{3.32}		&\red{5.64}		&\red{7.67}		&\red{9.48}		&\red{11.13}		&\red{12.64}		&\red{14.04} \\
        & 	 	 &StdE{\footnotesize\,(bp)}		& \blue{0.27}		&\blue{0.24}		&\blue{0.23} 		&\blue{0.21}	&\blue{0.20}	&\blue{0.18}	&\blue{0.17} \\
 \midrule

 Put 10	& {1M} &IV{\footnotesize\,(\%)}		   & 18.34		&18.22		&18.11		&17.99		&17.89		&17.78		&17.67 \\
 		& 	  	&SE{\footnotesize\,(bp)}		& \red{0.39}		&\red{0.84}		&\red{1.25}		&\red{1.65}		&\red{2.04}		&\red{2.42}		&\red{2.79} \\
        & 	 	&StdE{\footnotesize\,(bp)}		& \blue{0.19}		&\blue{0.19}		&\blue{0.19} 		&\blue{0.18}	&\blue{0.18}	&\blue{0.18}	&\blue{0.18} \\
 		& {3M}  &IV{\footnotesize\,(\%)}		& 18.31		&17.98		&17.67		&17.37		&17.10		&16.83		&16.58 \\
 		& 	  	&SE{\footnotesize\,(bp)}		& \red{1.23}		&\red{2.19}		&\red{3.08}		&\red{3.91}		&\red{4.71}		&\red{5.45}		&\red{6.17} \\
        & 	 	&StdE{\footnotesize\,(bp)}		& \blue{0.22}		&\blue{0.21}		&\blue{0.21} 		&\blue{0.20}	&\blue{0.20}	&\blue{0.19}	&\blue{0.19} \\
 		& {6M}  &IV{\footnotesize\,(\%)}		& 18.09		&17.49		&16.96		&16.47		&16.03		&15.62		&15.25 \\
 		& 	  	&SE{\footnotesize\,(bp)}		& \red{4.08}		&\red{5.53}		&\red{6.82}		&\red{7.99}		&\red{9.05}		&\red{10.02}		&\red{10.93} \\
        & 	    &StdE{\footnotesize\,(bp)}		& \blue{0.24}		&\blue{0.23}		&\blue{0.22} 		&\blue{0.21}	&\blue{0.20}	&\blue{0.19}	&\blue{0.18} \\
 		& {1Y}  &IV{\footnotesize\,(\%)}		& 17.50		&16.53		&15.71		&15.01		&14.41		&13.87		&13.40 \\
 		& 	   	&SE{\footnotesize\,(bp)}		& \red{10.89}		&\red{12.23}		&\red{13.36}		&\red{14.36}		&\red{15.27}		&\red{16.13}		&\red{16.93} \\
        & 	 	&StdE{\footnotesize\,(bp)}		& \blue{0.26}		&\blue{0.24}		&\blue{0.22} 		&\blue{0.20}	&\blue{0.19}	&\blue{0.18}	&\blue{0.17} \\
 \bottomrule
\end{tabular}
\end{center}
\end{table}

\label{sec:varyingthetaverhulst}
\;

\begin{table}[H]
\caption{Approximated Implied volatilities in \%, signed error of implied volatilities (Approximation - Monte Carlo) in basis points, and standard errors of the benchmark Monte-Carlo simulation in basis points (IV(\%), SE(bp), StdE(bp) respectively), computed in the Stochastic Verhulst model with $\theta$ varying from 40\% to 160\% of its `safe' 3-piece value.}
\label{table:varyingthetaverhulst}

\begin{center}
\begin{tabular}{lc l rrrrrrr}
\toprule
& & & \multicolumn{7}{c}{scaling factor for $\theta$ safe value} \\
\cmidrule{1-2}\cmidrule{4-10} 
Moneyness&	$T$ & IV/SE/StdE	& 40\% & 60\%  & 80\%  & 100\% & 120\% & 140\% & 160\% \\
\cmidrule{1-2}\cmidrule{4-10}  

ATM	    & {1M}   &IV{\footnotesize\,(\%)}	    &17.40 &17.41 &17.42 &17.44 &17.45 &17.46 &17.47 \\
		& 	 	 &SE{\footnotesize\,(bp)}		&\red{2.10} &\red{2.03} &\red{1.97} &\red{1.91} &\red{1.85} &\red{1.79} &\red{1.74} \\
		& 	 	 &StdE{\footnotesize\,(bp)}		&\blue{0.31} &\blue{0.31} &\blue{0.31} &\blue{0.31} &\blue{0.31} &\blue{0.31} &\blue{0.31} \\
 		& {3M}   &IV{\footnotesize\,(\%)}		&16.28 &16.31 &16.34 &16.37 &16.40 &16.43 &16.46 \\
 		& 	  	 &SE{\footnotesize\,(bp)}		&\red{-1.48} &\red{-1.56} &\red{-1.64} &\red{-1.71} &\red{-1.78} &\red{-1.85} &\red{-1.92} \\
		& 	 	 &StdE{\footnotesize\,(bp)}		&\blue{0.31} &\blue{0.32} &\blue{0.32} &\blue{0.32} &\blue{0.32} &\blue{0.32} &\blue{0.32} \\
 		& {6M}   &IV{\footnotesize\,(\%)}		&14.98 &15.02 &15.07 &15.12 &15.17 &15.22 &15.27 \\
 		& 	  	 &SE{\footnotesize\,(bp)}		&\red{-0.45} &\red{-0.59} &\red{-0.72} &\red{-0.85} &\red{-0.98} &\red{-1.11} &\red{-1.24} \\
        & 	 	 &StdE{\footnotesize\,(bp)}		&\blue{0.31} &\blue{0.31} &\blue{0.31} &\blue{0.31} &\blue{0.31} &\blue{0.31} &\blue{0.31} \\
 		& {1Y}   &IV{\footnotesize\,(\%)}		&13.08 &13.15 &13.22 &13.29 &13.36 &13.44 &13.51 \\
 		& 	  	 &SE{\footnotesize\,(bp)}		&\red{1.06} &\red{0.82} &\red{0.58} &\red{0.35} &\red{0.12} &\red{-0.11} &\red{-0.34} \\
        & 	 	 &StdE{\footnotesize\,(bp)}		&\blue{0.29} &\blue{0.29} &\blue{0.29} &\blue{0.29} &\blue{0.29} &\blue{0.29} &\blue{0.30} \\

 \midrule
 
Put 25 	& {1M}  &IV{\footnotesize\,(\%)}		&17.68 &17.70 &17.71 &17.72 &17.73 &17.74 &17.75 \\
		& 	  	&SE{\footnotesize\,(bp)}		&\red{1.96} &\red{1.91} &\red{1.85} &\red{1.79} &\red{1.74} &\red{1.68} &\red{1.62} \\
        & 	 	&StdE{\footnotesize\,(bp)}		&\blue{0.22} &\blue{0.22} &\blue{0.22} &\blue{0.22} &\blue{0.22} &\blue{0.22} &\blue{0.22} \\
 		& {3M}  &IV{\footnotesize\,(\%)}		&16.79 &16.82 &16.85 &16.88 &16.91 &16.94 & 16.97 \\
 		& 	  	&SE{\footnotesize\,(bp)}		&\red{2.10} &\red{2.03} &\red{1.96} &\red{1.89} &\red{1.82} &\red{1.75} &\red{1.68} \\
        & 	 	&StdE{\footnotesize\,(bp)}		&\blue{0.23} &\blue{0.23} &\blue{0.23} &\blue{0.23} &\blue{0.23} &\blue{0.23} &\blue{0.23} \\
		& {6M}  &IV{\footnotesize\,(\%)}		&15.65 &15.70 &15.75 &15.80 &15.85 &15.90 &15.95 \\
		& 	  	&SE{\footnotesize\,(bp)}		&\red{5.15} &\red{5.02} &\red{4.90} &\red{4.78} &\red{4.66} &\red{4.53} &\red{4.41} \\
        & 	 	&StdE{\footnotesize\,(bp)}		&\blue{0.22} &\blue{0.22} &\blue{0.22} &\blue{0.22} &\blue{0.23} &\blue{0.23} &\blue{0.23} \\
		& {1Y}  &IV{\footnotesize\,(\%)}		&13.94 &14.01 &14.08 &14.15 &14.23 &14.30 &14.37 \\
		& 	   	&SE{\footnotesize\,(bp)}		&\red{10.12} &\red{9.91} &\red{9.69} &\red{9.48} &\red{9.27} &\red{9.06} &\red{8.85} \\
        & 	 	&StdE{\footnotesize\,(bp)}		&\blue{0.21} &\blue{0.21} &\blue{0.21} &\blue{0.21} &\blue{0.21} &\blue{0.21} &\blue{0.21} \\
 \midrule
 
Put 10	& {1M}  &IV{\footnotesize\,(\%)}		&17.96 &17.97 &17.98 &17.99 &18.01 &18.02 &18.03 \\
 		& 	  	&SE{\footnotesize\,(bp)}		&\red{1.80} &\red{1.76} &\red{1.70} &\red{1.65} &\red{1.59} &\red{1.53} &\red{1.48} \\
        & 	 	&StdE{\footnotesize\,(bp)}		&\blue{0.18} &\blue{0.18} &\blue{0.18} &\blue{0.18} &\blue{0.18} &\blue{0.18} &\blue{0.18} \\
 		& {3M}  &IV{\footnotesize\,(\%)}		&17.28 &17.31 &17.34 &17.37 &17.40 &17.43 &17.47 \\
 		& 	  	&SE{\footnotesize\,(bp)}		&\red{4.08} &\red{4.03} &\red{3.97} &\red{3.91} &\red{3.86} &\red{3.80} &\red{3.74} \\
        & 	 	&StdE{\footnotesize\,(bp)}		&\blue{0.20} &\blue{0.20} &\blue{0.20} &\blue{0.20} &\blue{0.20} &\blue{0.20} &\blue{0.20} \\
 		& {6M}  &IV{\footnotesize\,(\%)}		&16.32 &16.37 &16.42 &16.47 &16.52 &16.57 &16.62 \\
 		& 	  	&SE{\footnotesize\,(bp)}		&\red{8.23} &\red{8.15} &\red{8.07} &\red{7.99} &\red{7.90} &\red{7.82} &\red{7.74} \\
        & 	 	&StdE{\footnotesize\,(bp)}		&\blue{0.20} &\blue{0.21} &\blue{0.21} &\blue{0.21} &\blue{0.21} &\blue{0.21} &\blue{0.21} \\
 		& {1Y}  &IV{\footnotesize\,(\%)}		&14.79 &14.86 &14.94 &15.01 &15.09 &15.16 &15.24 \\
 		& 	   	&SE{\footnotesize\,(bp)}		&\red{14.54} &\red{14.48} &\red{14.42} &\red{14.36} &\red{14.30} &\red{14.23} &\red{14.17} \\
        & 	 	&StdE{\footnotesize\,(bp)}		&\blue{0.20} &\blue{0.20} &\blue{0.20} &\blue{0.20} &\blue{0.20} &\blue{0.20} &\blue{0.20} \\
\bottomrule
\end{tabular}
\end{center}
\end{table}

\label{sec:varyinglambdaverhulst}
\;

\begin{table}[H]
\caption{Approximated Implied volatilities in \%, signed error of implied volatilities (Approximation - Monte Carlo) in basis points, and standard errors of the benchmark Monte-Carlo simulation in basis points (IV(\%), SE(bp), StdE(bp) respectively), computed in the Stochastic Verhulst model with $\lambda$ varying from 40\% to 160\% of its `safe' 3-piece value.}
\label{table:varyinglambdaverhulst}
\begin{center}
\begin{tabular}{lc l rrrrrrr}
\toprule
& & & \multicolumn{7}{c}{scaling factor for $\lambda$ safe value} \\
\cmidrule{1-2}\cmidrule{4-10} 
Moneyness & $T$ & IV/SE/StdE & 40\% & 60\% & 80\% & 100\% & 120\% & 140\% & 160\% \\
\cmidrule{1-2}\cmidrule{4-10}  

ATM        & {1M}  & IV{\footnotesize\,(\%)}        & 17.44 & 17.43 & 17.43 & 17.44 & 17.44 & 17.44 & 17.44 \\
           &       & SE{\footnotesize\,(bp)}        & \red{1.91} & \red{1.91} & \red{1.91} & \red{1.91} & \red{1.92} & \red{1.93} & \red{1.95} \\
           &       & StdE{\footnotesize\,(bp)}      & \blue{0.29} & \blue{0.30} & \blue{0.30} & \blue{0.31} & \blue{0.32} & \blue{0.33} & \blue{0.34} \\
           & {3M}  & IV{\footnotesize\,(\%)}        & 16.38 & 16.37 & 16.37 & 16.37 & 16.37 & 16.37 & 16.37 \\
           &       & SE{\footnotesize\,(bp)}        & \red{-1.87} & \red{-1.83} & \red{-1.78} & \red{-1.71} & \red{-1.62} & \red{-1.51} & \red{-1.36} \\
           &       & StdE{\footnotesize\,(bp)}      & \blue{0.28} & \blue{0.29} & \blue{0.30} & \blue{0.32} & \blue{0.33} & \blue{0.34} & \blue{0.35} \\
           & {6M}  & IV{\footnotesize\,(\%)}        & 15.16 & 15.14 & 15.13 & 15.12 & 15.11 & 15.11 & 15.10 \\
           &       & SE{\footnotesize\,(bp)}        & \red{-1.13} & \red{-1.06} & \red{-0.97} & \red{-0.85} & \red{-0.69} & \red{-0.46} & \red{-0.15} \\
           &       & StdE{\footnotesize\,(bp)}      & \blue{0.27} & \blue{0.28} & \blue{0.30} & \blue{0.31} & \blue{0.32} & \blue{0.34} & \blue{0.35} \\
           & {1Y}  & IV{\footnotesize\,(\%)}        & 13.37 & 13.35 & 13.32 & 13.29 & 13.26 & 13.24 & 13.21 \\
           &       & SE{\footnotesize\,(bp)}        & \red{-0.15} & \red{-0.01} & \red{0.15} & \red{0.35} & \red{0.61} & \red{0.97} & \red{1.44} \\
           &       & StdE{\footnotesize\,(bp)}      & \blue{0.25} & \blue{0.26} & \blue{0.28} & \blue{0.29} & \blue{0.31} & \blue{0.32} & \blue{0.33} \\
\midrule

Put 25     & {1M}  & IV{\footnotesize\,(\%)}        & 17.55 & 17.60 & 17.66 & 17.72 & 17.78 & 17.84 & 17.91 \\
           &       & SE{\footnotesize\,(bp)}        & \red{1.88} & \red{1.84} & \red{1.81} & \red{1.79} & \red{1.81} & \red{1.85} & \red{1.92} \\
           &       & StdE{\footnotesize\,(bp)}      & \blue{0.19} & \blue{0.20} & \blue{0.21} & \blue{0.22} & \blue{0.23} & \blue{0.24} & \blue{0.24} \\
           & {3M}  & IV{\footnotesize\,(\%)}        & 16.59 & 16.68 & 16.78 & 16.88 & 16.98 & 17.09 & 17.20 \\
           &       & SE{\footnotesize\,(bp)}        & \red{1.19} & \red{1.34} & \red{1.56} & \red{1.89} & \red{2.34} & \red{2.95} & \red{3.74} \\
           &       & StdE{\footnotesize\,(bp)}      & \blue{0.19} & \blue{0.20} & \blue{0.21} & \blue{0.23} & \blue{0.24} & \blue{0.25} & \blue{0.27} \\
           & {6M}  & IV{\footnotesize\,(\%)}        & 15.44 & 15.56 & 15.67 & 15.80 & 15.93 & 16.06 & 16.21 \\
           &       & SE{\footnotesize\,(bp)}        & \red{3.50} & \red{3.72} & \red{4.13} & \red{4.78} & \red{5.73} & \red{7.04} & \red{8.78} \\
           &       & StdE{\footnotesize\,(bp)}      & \blue{0.18} & \blue{0.19} & \blue{0.21} & \blue{0.22} & \blue{0.24} & \blue{0.26} & \blue{0.28} \\
           & {1Y}  & IV{\footnotesize\,(\%)}        & 13.81 & 13.94 & 14.07 & 14.22 & 14.37 & 14.53 & 14.69 \\
           &       & SE{\footnotesize\,(bp)}        & \red{7.24} & \red{7.59} & \red{8.30} & \red{9.48} & \red{11.26} & \red{13.73} & \red{17.01} \\
           &       & StdE{\footnotesize\,(bp)}      & \blue{0.16} & \blue{0.17} & \blue{0.19} & \blue{0.21} & \blue{0.23} & \blue{0.25} & \blue{0.26} \\
\midrule

Put 10     & {1M}  & IV{\footnotesize\,(\%)}        & 17.65 & 17.76 & 17.88 & 17.99 & 18.12 & 18.25 & 18.38 \\
           &       & SE{\footnotesize\,(bp)}        & \red{1.81} & \red{1.73} & \red{1.67} & \red{1.65} & \red{1.68} & \red{1.77} & \red{1.94} \\
           &       & StdE{\footnotesize\,(bp)}      & \blue{0.15} & \blue{0.16} & \blue{0.17} & \blue{0.18} & \blue{0.19} & \blue{0.20} & \blue{0.22} \\
           & {3M}  & IV{\footnotesize\,(\%)}        & 16.78 & 16.97 & 17.17 & 17.37 & 17.59 & 17.82 & 18.05 \\
           &       & SE{\footnotesize\,(bp)}        & \red{2.57} & \red{2.83} & \red{3.26} & \red{3.91} & \red{4.84} & \red{6.09} & \red{7.71} \\
           &       & StdE{\footnotesize\,(bp)}      & \blue{0.15} & \blue{0.17} & \blue{0.18} & \blue{0.20} & \blue{0.22} & \blue{0.24} & \blue{0.25} \\
           & {6M}  & IV{\footnotesize\,(\%)}        & 15.69 & 15.94 & 16.20 & 16.47 & 16.76 & 17.06 & 17.36 \\
           &       & SE{\footnotesize\,(bp)}        & \red{5.41} & \red{5.83} & \red{6.65} & \red{7.99} & \red{9.98} & \red{12.73} & \red{16.37} \\
           &       & StdE{\footnotesize\,(bp)}      & \blue{0.15} & \blue{0.17} & \blue{0.19} & \blue{0.21} & \blue{0.23} & \blue{0.25} & \blue{0.27} \\
           & {1Y}  & IV{\footnotesize\,(\%)}        & 14.06 & 14.36 & 14.68 & 15.01 & 15.36 & 15.71 & 16.08 \\
           &       & SE{\footnotesize\,(bp)}        & \red{9.91} & \red{10.52} & \red{11.92} & \red{14.36} & \red{18.06} & \red{23.22} & \red{29.99} \\
           &       & StdE{\footnotesize\,(bp)}      & \blue{0.13} & \blue{0.16} & \blue{0.18} & \blue{0.20} & \blue{0.23} & \blue{0.25} & \blue{0.27} \\
\bottomrule
\end{tabular}
\end{center}
\end{table}

\label{sec:varyingrhoverhulst}
\;

\begin{table}[H]
\caption{Approximated Implied volatilities in \%, signed error of implied volatilities (Approximation - Monte Carlo) in basis points, and standard errors of the benchmark Monte-Carlo simulation in basis points (IV(\%), SE(bp), StdE(bp) respectively), computed in the Stochastic Verhulst model with $\rho$ varying from 160\% to 40\% of its `safe' 3-piece value.}
\label{table:varyingrhoverhulst}
\begin{center}
\begin{tabular}{lc l rrrrrrr}
\toprule
& & & \multicolumn{7}{c}{scaling factor for $\rho$ safe value} \\
\cmidrule{1-2}\cmidrule{4-10} 
Moneyness & $T$ & IV/SE/StdE & \ 160\% & 140\% & 120\% & 100\% & 80\% & 60\% & 40\% \\
\cmidrule{1-2}\cmidrule{4-10}  

 ATM    & {1M}   & IV{\footnotesize\,(\%)}        & 17.42  & 17.42  & 17.43  & 17.44  & 17.44  & 17.45  & 17.45 \\
        &        & SE{\footnotesize\,(bp)}        & \red{1.73} & \red{1.79} & \red{1.85} & \red{1.91} & \red{1.97} & \red{2.03} & \red{2.08} \\
        &        & StdE{\footnotesize\,(bp)}      & \blue{0.50} & \blue{0.43} & \blue{0.37} & \blue{0.31} & \blue{0.25} & \blue{0.20} & \blue{0.14} \\
        & {3M}   & IV{\footnotesize\,(\%)}        & 16.32  & 16.34  & 16.35  & 16.37  & 16.38  & 16.40  & 16.41 \\
        &        & SE{\footnotesize\,(bp)}        & \red{-1.71} & \red{-1.71} & \red{-1.71} & \red{-1.71} & \red{-1.71} & \red{-1.72} & \red{-1.74} \\
        &        & StdE{\footnotesize\,(bp)}      & \blue{0.49} & \blue{0.43} & \blue{0.37} & \blue{0.32} & \blue{0.26} & \blue{0.21} & \blue{0.15} \\
        & {6M}   & IV{\footnotesize\,(\%)}        & 15.04  & 15.07  & 15.10  & 15.12  & 15.15  & 15.17  & 15.19 \\
        &        & SE{\footnotesize\,(bp)}        & \red{-1.06} & \red{-0.98} & \red{-0.91} & \red{-0.85} & \red{-0.81} & \red{-0.79} & \red{-0.79} \\
        &        & StdE{\footnotesize\,(bp)}      & \blue{0.47} & \blue{0.42} & \blue{0.36} & \blue{0.31} & \blue{0.26} & \blue{0.21} & \blue{0.16} \\
        & {1Y}   & IV{\footnotesize\,(\%)}        & 13.17  & 13.21  & 13.25  & 13.29  & 13.33 & 13.36  & 13.40 \\
        &        & SE{\footnotesize\,(bp)}        & \red{0.10} & \red{0.20} & \red{0.29} & \red{0.35} & \red{0.39} & \red{0.41} & \red{0.41} \\
        &        & StdE{\footnotesize\,(bp)}      & \blue{0.44} & \blue{0.39} & \blue{0.34} & \blue{0.29} & \blue{0.24} & \blue{0.20} & \blue{0.15} \\
\midrule

 Put 25 & {1M}   & IV{\footnotesize\,(\%)}        & 17.86  & 17.81  & 17.76  & 17.72  & 17.67  & 17.62  & 17.58 \\
        &        & SE{\footnotesize\,(bp)}        & \red{1.51} & \red{1.62} & \red{1.71} & \red{1.79} & \red{1.87} & \red{1.94} & \red{2.01} \\
        &        & StdE{\footnotesize\,(bp)}      & \blue{0.36} & \blue{0.31} & \blue{0.26} & \blue{0.22} & \blue{0.18} & \blue{0.14} & \blue{0.10} \\
        & {3M}   & IV{\footnotesize\,(\%)}        & 17.08  & 17.01  & 16.95  & 16.88  & 16.81  & 16.73  & 16.66 \\
        &        & SE{\footnotesize\,(bp)}        & \red{2.13} & \red{2.07} & \red{1.99} & \red{1.89} & \red{1.77} & \red{1.64} & \red{1.50} \\
        &        & StdE{\footnotesize\,(bp)}      & \blue{0.36} & \blue{0.31} & \blue{0.27} & \blue{0.23} & \blue{0.19} & \blue{0.15} & \blue{0.12} \\
        & {6M}   & IV{\footnotesize\,(\%)}        & 16.03  & 15.96  & 15.88  & 15.80  & 15.71  & 15.63  & 15.54 \\
        &        & SE{\footnotesize\,(bp)}        & \red{4.87} & \red{4.89} & \red{4.86} & \red{4.78} & \red{4.66} & \red{4.49} & \red{4.29} \\
        &        & StdE{\footnotesize\,(bp)}      & \blue{0.35} & \blue{0.30} & \blue{0.26} & \blue{0.22} & \blue{0.19} & \blue{0.15} & \blue{0.12} \\
        & {1Y}   & IV{\footnotesize\,(\%)}        & 14.41  & 14.33  & 14.24  & 14.15  & 14.06  & 13.96  & 13.86 \\
        &        & SE{\footnotesize\,(bp)}        & \red{9.63} & \red{9.66} & \red{9.61} & \red{9.48} & \red{9.28} & \red{9.02} & \red{8.70} \\
        &        & StdE{\footnotesize\,(bp)}      & \blue{0.32} & \blue{0.28} & \blue{0.24} & \blue{0.21} & \blue{0.18} & \blue{0.15} & \blue{0.12} \\
\midrule

 Put 10 & {1M}   & IV{\footnotesize\,(\%)}        & 18.26  & 18.17  & 18.09  & 17.99  & 17.90  & 17.81  & 17.72 \\
        &        & SE{\footnotesize\,(bp)}        & \red{1.38} & \red{1.49} & \red{1.57} & \red{1.65} & \red{1.72} & \red{1.80} & \red{1.88} \\
        &        & StdE{\footnotesize\,(bp)}      & \blue{0.32} & \blue{0.27} & \blue{0.22} & \blue{0.18} & \blue{0.15} & \blue{0.12} & \blue{0.09} \\
        & {3M}   & IV{\footnotesize\,(\%)}        & 17.78  & 17.64  & 17.51  & 17.37  & 17.23  & 17.09  & 16.94 \\
        &        & SE{\footnotesize\,(bp)}        & \red{4.82} & \red{4.52} & \red{4.21} & \red{3.91} & \red{3.63} & \red{3.34} & \red{3.06} \\
        &        & StdE{\footnotesize\,(bp)}      & \blue{0.33} & \blue{0.28} & \blue{0.24} & \blue{0.20} & \blue{0.17} & \blue{0.13} & \blue{0.10} \\
        & {6M}   & IV{\footnotesize\,(\%)}        & 16.95  & 16.80  & 16.64  & 16.47  & 16.30  & 16.13  & 15.95 \\
        &        & SE{\footnotesize\,(bp)}        & \red{9.46} & \red{8.95} & \red{8.45} & \red{7.99} & \red{7.54} & \red{7.10} & \red{6.68} \\
        &        & StdE{\footnotesize\,(bp)}      & \blue{0.33} & \blue{0.29} & \blue{0.24} & \blue{0.21} & \blue{0.17} & \blue{0.14} & \blue{0.11} \\
        & {1Y}   & IV{\footnotesize\,(\%)}        & 15.55  & 15.38  & 15.20  & 15.01  & 14.82  & 14.62  & 14.42 \\
        &        & SE{\footnotesize\,(bp)}        & \red{17.52} & \red{16.35} & \red{15.30} & \red{14.36} & \red{13.53} & \red{12.80} & \red{12.16} \\
        &        & StdE{\footnotesize\,(bp)}      & \blue{0.32} & \blue{0.28} & \blue{0.24} & \blue{0.20} & \blue{0.17} & \blue{0.14} & \blue{0.12} \\
\bottomrule
\end{tabular}
\end{center}
\end{table}

\begin{table}[H]
\caption{Run times (Time(s)) in the Stochastic Verhulst model for the closed-form approximation formula and fair Monte-Carlo simulation with `safe' 3-piece parameters in seconds. Additionally, their corresponding Put option prices (Price), signed error (SE(\%))/standard error (StdE(\%)) are displayed. \# Paths is the number of paths in the fair Monte-Carlo simulation rounded down to the nearest 100.}
\label{table:runtimeverhulst}
\begin{center}
\begin{tabular}{lc c rrr c rrrr}
\toprule
& & & \multicolumn{3}{c}{Closed-form formula} & & \multicolumn{4}{c}{Fair Monte-Carlo} \\
\cmidrule{1-2}\cmidrule{4-6} \cmidrule{8-11} 
Moneyness &	$T$  & &	Time(s) & Price & SE(\%)  & & Time(s)  & Price & StdE(\%) & \# Paths \\
\cmidrule{1-2}\cmidrule{4-6} \cmidrule{8-11}

 ATM	& {1M} 	&  &1.36  	&2.077  & 0.22  	&	&29.60 	&2.075  	&0.112  	&1,023,400 \\
     & {3M} 	&  &1.36  	&3.478	&-0.34  	&	&112.50 &3.482  	&0.174  	&1,317,500 \\
     & {6M} 	&  &1.33  	&4.700	&-0.24  	&	&884.20 &4.706  	&0.122  	&5,074,000 \\
     & {1Y} 	&  &1.40  	&6.198	&0.14  	    &	&9016.45&6.208      &0.071	    &26,434,000 \\
\midrule

 Put 25 & {1M}  &  &1.37    &0.770  &0.163   &	 &16.65    &0.768   &0.083      &567,500 \\
        & {3M} 	&  &1.38  	&1.229	&0.292  	&	  &48.25 	&1.227   &0.149      &553,100\\
        & {6M} 	&  &1.40  	&1.544	&1.015  	&	  &10.35 	&1.533   &0.515      &84,700 \\
        & {1Y} 	&  &1.36  	&1.779	&2.710  	&	  &3.46 	&1.750   &1.372      &18,600 \\
\midrule

 Put 10 & {1M}  &  &1.39   &0.251  &0.083  &	     &13.89    &0.251  	   &0.043     &474,800 \\
        & {3M} 	&  &1.40    &0.389	&0.324  	&	  &5.61 	&0.385  	&0.165  	&100,500\\
        & {6M} 	&  &1.38  	&0.459	&0.855  	&	  &2.24 	&0.450  	&0.434  	&25,500 \\
        & {1Y} 	&  &1.38  	&0.463	&1.831  	&	  &1.33 	&0.430  	&0.902  	&7,400 \\

 \bottomrule
\end{tabular}
\end{center}
\end{table}

%% file: SABRGeneralmodelexpmalliavin.tex
\section{SABR-$\mu$ model}
\label{appen:sabr}
In this section, we consider the problem of pricing a put option in the following model, which we dub the SABR-$\mu$ model:
\begin{align}
\label{eqn:SABRmu}
\begin{split}
    \dd F_t &= F_t V_t \dd W_t, \quad F_0 = S_0 > 0,\\ 
    \dd V_t &= \lambda_t V_t^\mu \dd B_t, \quad V_0 = v_0, \\
    \dd \langle W, B \rangle_t &= \rho_t \dd t,
\end{split}
\end{align}
where $F_t = S_t e^{\int_0^t r_u^f \dd u}/e^{\int_0^t r_u^d \dd u}$ is the forward spot. The classical SABR model (with skewness parameter equal to $1$) \citep{sabr} can be recovered by setting $\mu = 1$ in \cref{eqn:SABRmu}. The purpose of this section is to illustrate that our general closed-form approximation formula also works well, as expected, in familiar classical models such as the SABR model.

The formula for the second-order price of a put option in the SABR-$\mu$ model \cref{eqn:SABRmu} can be deduced by substituting $\alpha(t, x) \equiv 0$, and $v_{0, t} = v_0$ into the formula given in \Cref{thm:expprice}. This leads to the following corollary:
\begin{corollary}[SABR-$\mu$ model explicit second-order put option price]
\label{cor:SABRprice}
Under the SABR-$\mu$ model \cref{eqn:SABRmu}, the explicit second-order price of a put option is given by
\begin{align*}
\text{Put}^{(2)}_{\text{SABR}} &= P_{\text{BS}} \left ( x_0, v_0^2 T \right )  \\
&\quad + 2 \omega_{0,T} ^{(0, \rho \lambda v_{0}^{\mu + 1}), (0, v_{0}) } \partial_{xy} P_{\text{BS}} \left ( x_0, v_0^2 T \right ) \\
&\quad + \omega_{0, T}^{(0, \lambda^2 v_{0}^{2\mu } ), ( 0, 1)} \partial_{y} P_{\text{BS}} \left ( x_0, v_0^2 T \right ) \\
&\quad + 2\omega_{0,T}^{(0, \rho \lambda v_{0}^{\mu + 1}), (0, \rho \lambda v_{0}^{\mu + 1}), (0,1 ) } \partial_{xxy} P_{\text{BS}} \left ( x_0, v_0^2 T \right )  \\
&\quad + 2\mu \omega_{0,T}^{(0, \rho \lambda v_{0}^{\mu + 1} ), (0, \rho \lambda v_{0}^{2 \mu - 1 }), (0, v_{0} ) }  \partial_{xxy} P_{\text{BS}} \left ( x_0, v_0^2 T \right ) \\
&\quad +2 \omega_{0,T}^{(0, \rho \lambda v_{0}^{\mu + 1}), (0, \rho \lambda v_{0}^\mu), (0, v_{0}) } \partial_{xxy} P_{\text{BS}} \left ( x_0, v_0^2 T \right ) \\
&\quad + 4 \omega_{0,T}^{(0, \lambda^2 v_{0}^{2\mu}),(0, v_{0}), (0, v_{0}) }  \partial_{yy} P_{\text{BS}} \left ( x_0, v_0^2 T  \right ) \\
&\quad +2\left ( \omega_{0,T}^{(0, \rho \lambda v_{0}^{\mu + 1}),(0, v_{0}) } \right )^2   \partial_{xxyy} P_{\text{BS}} \left ( x_0, v_0^2 T \right )
\end{align*}
where $x_0 = \ln(S_0)$, the partial derivatives of $P_{\text{BS}}$ are given in \Cref{appen:partialderivativesPBS}, and the definitions of the integral operators $\omega_{0, T}^{(\cdot, \cdot), \dots, (\cdot, \cdot)}$ are given in \Cref{defn:integraloperator}.
\end{corollary}

Using the formula in \Cref{cor:SABRprice}, we can rapidly compute approximate put option prices in the SABR-$\mu$ model. Moreover, the classical SABR model is a special case of the Verhulst model \cref{eqn:xgbm}, where we take the mean reversion speed to be $0$. As a consequence, the function $(v_{0,t})$ (which satisfies an ODE similar to the SDE of $V$ but without any diffusion component) is simply a constant, and we do not need to numerically compute it over a fine grid. This makes the computation of our closed-form approximation formula in this setting extremely quick.

The upcoming two subsections include our numerical experiments in the classical SABR model with piecewise-constant parameters. For more in-depth explanations regarding technical points, terminology, and so on, we refer the reader to \Cref{sec:numerical3}, where we tackled the more complex Stochastic Verhulst model in detail.

\begin{remark}
The code utilised to obtain the numerical results in this section is available on GitHub \citep{das2023}. In particular, what is provided is: 
\begin{itemize}
\item A routine which computes our closed-form approximation of put option prices for the SABR-$\mu$ model with piecewise-constant parameter inputs.
\item A routine which implements the Monte-Carlo simulation via the mixing solution methodology for the pricing of put option prices in the SABR-$\mu$ model with piecewise-constant parameter inputs.
\item A routine which compares the accuracy and runtimes of the aforementioned methods.
\end{itemize}
\end{remark}

First of all, we start off with the following `safe' constant parameters: 
\begin{center}
\begin{tabular}{llllcc} 
\toprule
$S_0$ & \ $v_0$ & $r^d$ & $r^f$ & $\lambda$ & $\rho$ \\
\midrule
$100$ & \ $18\%$ &  2\% & 0 & 0.414 & -0.391 \\
 \bottomrule \\
\end{tabular}
\end{center}
However, our closed-form approximation formula is built to take advantage of piecewise-constant parameter inputs. Therefore, we generate `safe' piecewise-constant parameters by perturbing the safe constant parameters in the following way:

\begin{center}
\begin{tabular}{lcc c cc cc}
\toprule
$T$	 & Piece & Proportion	&& 	 $\lambda$ & $\rho$ & 	$r^d$ & $r^f$   \\

\midrule
1M   	& 1 & 1/4 		&&	 	 0.394 & -0.371 & 			1\% & 0 \\
   		& 2 & 1/4 	   	&&	 	 0.434 & -0.411 &			3\% & 0 \\
   		& 3 & 1/2 		&&	 	 0.414 & -0.391 & 			2\% & 0 \\

\midrule
3M   	& 1 & 1/4  		&&	 	 0.394 & -0.371 &			1\% & 0   \\
   		& 2 & 1/4	   	&&	 	 0.434 & -0.411 & 			3\% & 0   \\
   		& 3 & 1/2   	&&	 	 0.414 & -0.391 & 			2\% & 0  \\

\midrule
6M   	& 1 & 1/4   	&&	 	 0.394 & -0.371 & 			1\% & 0	\\
   		& 2 & 1/4	   	&&	 	 0.434 & -0.411 & 			3\% & 0 	 \\
   		& 3 & 1/2 		&&	 	 0.414 & -0.391 & 			2\% & 0 	 \\
    
\midrule
1Y   	& 1 & 1/4  		&&	 	 0.394 & -0.371 &  			1\% & 0 	\\
   		& 2 & 1/4	   	&&	 	 0.434 & -0.411 & 			3\% & 0	 \\
   		& 3 & 1/2   	&&	 	 0.414 & -0.391 & 			2\% & 0	 \\

\bottomrule
\end{tabular}
\end{center}


\subsection{Sensitivity analysis}
We perform a sensitivity analysis similar to the one from \Cref{sec:numerical3}, varying parameters away from their `safe' values and studying the error in the closed-form approximation formula with the benchmark Monte-Carlo simulation. The relevant tables are \Cref{table:varyinglambdaSABR} and \Cref{table:varyingrhoSABR} for the sensitivity analysis of $\lambda$ and $\rho$ respectively. The sensitivity analysis for $\lambda$ and $\rho$ demonstrates that our closed-form approximation formula fares well with parameter adjustments. As expected, and as our error analysis from \Cref{sec:erroranalysis3} suggests, the error increases as maturity increases, as well as for out-of-the-money put options.

\subsection{Run time comparison against Monte-Carlo simulation}

Since $(v_{0,t})$ is a constant in the SABR model, the grid $\tilde{\mathcal{T}}$ has one point. Therefore, the closed-form approximation formula in the classical SABR model is almost instantaneous, in the realm of $10^{-3}$ seconds. This will clearly outperform any Monte-Carlo simulation whose error is similar to the closed-form approximation formula error. This is in contrast to the Stochastic Verhulst model, where it is roughly 1 second --- this is because we have $27$ points in the $\tilde{\mathcal{T}}$ grid.

Rather than providing extensive numerical results, consider the following. A Monte-Carlo simulation with $1,000$ paths and maturity $T = 1/12$ with $24$ time steps a day takes $0.015$ seconds. In contrast, our closed-form approximation formula takes only $0.001$ seconds. Of course, this is not a sensible comparison, as the standard error for a Monte-Carlo simulation with $1,000$ paths is very high. Rather, this is illustrating a beyond worst case scenario. We would need many more paths to obtain a fair Monte-Carlo simulation. Therefore, there is no possibility that a fair Monte-Carlo simulation can outperform our closed-form approximation formula in this setting.

\input{TablesGeneralmodelexpmalliavin_SABRmodel}

%% file: TablesGeneralmodelexpmalliavin_SABRmodel.tex


\label{sec:varyinglambdaSABR}
\;

\begin{table}[H]
\caption{Implied volatilities in \%, signed error of implied volatilities (Approximation - Monte Carlo) in basis points, and standard errors of the benchmark Monte-Carlo simulation in basis points (IV(\%), SE(bp), StdE(bp) respectively), computed in the classical SABR model with $\lambda$ varying from 40\% to 160\% of its `safe' 3-piece value.}
\label{table:varyinglambdaSABR}
\begin{center}
\begin{tabular}{lc l rrrrrrr}
\toprule
& & & \multicolumn{7}{c}{scaling factor for $\lambda$ safe value} \\
\cmidrule{1-2}\cmidrule{4-10} 
Moneyness & $T$ & IV/SE/StdE & 40\% & 60\% & 80\% & 100\% & 120\% & 140\% & 160\% \\
\cmidrule{1-2}\cmidrule{4-10}  

ATM        & {1M}  & IV{\footnotesize\,(\%)}        & 17.99 & 17.99 & 17.99 & 18.00 & 18.00 & 18.00 & 18.01 \\
           &       & SE{\footnotesize\,(bp)}        & \red{-0.27} & \red{-0.27} & \red{-0.27} & \red{-0.26} & \red{-0.26} & \red{-0.25} & \red{-0.23} \\
           &       & StdE{\footnotesize\,(bp)}      & \blue{0.30} & \blue{0.31} & \blue{0.32} & \blue{0.32} & \blue{0.33} & \blue{0.34} & \blue{0.35} \\
           & {3M}  & IV{\footnotesize\,(\%)}        & 17.98 & 17.98 & 17.98 & 17.99 & 17.99 & 18.01 & 18.03 \\
           &       & SE{\footnotesize\,(bp)}        & \red{0.03} & \red{0.06} & \red{0.10} & \red{0.16} & \red{0.24} & \red{0.34} & \red{0.47} \\
           &       & StdE{\footnotesize\,(bp)}      & \blue{0.31} & \blue{0.32} & \blue{0.34} & \blue{0.35} & \blue{0.37} & \blue{0.38} & \blue{0.40} \\
           & {6M}  & IV{\footnotesize\,(\%)}        & 17.96 & 17.96 & 17.96 & 17.97 & 17.99 & 18.02 & 18.05 \\
           &       & SE{\footnotesize\,(bp)}        & \red{-0.08} & \red{-0.02} & \red{0.08} & \red{0.23} & \red{0.46} & \red{0.77} & \red{1.20} \\
           &       & StdE{\footnotesize\,(bp)}      & \blue{0.32} & \blue{0.34} & \blue{0.36} & \blue{0.38} & \blue{0.40} & \blue{0.42} & \blue{0.45} \\
           & {1Y}  & IV{\footnotesize\,(\%)}        & 17.93 & 17.92 & 17.92 & 17.94 & 17.98 & 18.03 & 18.10 \\
           &       & SE{\footnotesize\,(bp)}        & \red{-0.15} & \red{0.07} & \red{0.44} & \red{1.04} & \red{1.94} & \red{3.24} & \red{5.08} \\
           &       & StdE{\footnotesize\,(bp)}      & \blue{0.33} & \blue{0.36} & \blue{0.39} & \blue{0.42} & \blue{0.45} & \blue{0.48} & \blue{0.52} \\
\midrule

Put 25     & {1M}  & IV{\footnotesize\,(\%)}        & 18.11 & 18.17 & 18.23 & 18.29 & 18.35 & 18.42 & 18.49 \\
           &       & SE{\footnotesize\,(bp)}        & \red{-0.28} & \red{-0.32} & \red{-0.34} & \red{-0.35} & \red{-0.33} & \red{-0.28} & \red{-0.19} \\
           &       & StdE{\footnotesize\,(bp)}      & \blue{0.20} & \blue{0.21} & \blue{0.22} & \blue{0.23} & \blue{0.24} & \blue{0.25} & \blue{0.26} \\
           & {3M}  & IV{\footnotesize\,(\%)}        & 18.18 & 18.28 & 18.39 & 18.51 & 18.63 & 18.76 & 18.90 \\
           &       & SE{\footnotesize\,(bp)}        & \red{-0.13} & \red{-0.13} & \red{-0.05} & \red{0.14} & \red{0.46} & \red{0.94} & \red{1.61} \\
           &       & StdE{\footnotesize\,(bp)}      & \blue{0.21} & \blue{0.23} & \blue{0.24} & \blue{0.26} & \blue{0.28} & \blue{0.30} & \blue{0.31} \\
           & {6M}  & IV{\footnotesize\,(\%)}        & 18.25 & 18.40 & 18.55 & 18.73 & 18.91 & 19.11 & 19.32 \\
           &       & SE{\footnotesize\,(bp)}        & \red{-0.28} & \red{-0.17} & \red{0.17} & \red{0.82} & \red{1.85} & \red{3.35} & \red{5.37} \\
           &       & StdE{\footnotesize\,(bp)}      & \blue{0.23} & \blue{0.25} & \blue{0.27} & \blue{0.29} & \blue{0.32} & \blue{0.35} & \blue{0.37} \\
           & {1Y}  & IV{\footnotesize\,(\%)}        & 18.34 & 18.54 & 18.78 & 19.03 & 19.32 & 19.62 & 19.95 \\
           &       & SE{\footnotesize\,(bp)}        & \red{-0.21} & \red{0.35} & \red{1.60} & \red{3.75} & \red{7.03} & \red{11.67} & \red{17.90} \\
           &       & StdE{\footnotesize\,(bp)}      & \blue{0.24} & \blue{0.27} & \blue{0.31} & \blue{0.34} & \blue{0.38} & \blue{0.42} & \blue{0.46} \\
\midrule

Put 10     & {1M}  & IV{\footnotesize\,(\%)}        & 18.22 & 18.33 & 18.45 & 18.58 & 18.71 & 18.84 & 18.98 \\
           &       & SE{\footnotesize\,(bp)}        & \red{-0.34} & \red{-0.40} & \red{-0.45} & \red{-0.45} & \red{-0.40} & \red{-0.29} & \red{-0.10} \\
           &       & StdE{\footnotesize\,(bp)}      & \blue{0.16} & \blue{0.17} & \blue{0.18} & \blue{0.19} & \blue{0.20} & \blue{0.22} & \blue{0.23} \\
           & {3M}  & IV{\footnotesize\,(\%)}        & 18.37 & 18.58 & 18.81 & 19.04 & 19.29 & 19.55 & 19.83 \\
           &       & SE{\footnotesize\,(bp)}        & \red{-0.27} & \red{-0.27} & \red{-0.10} & \red{0.31} & \red{1.02} & \red{2.08} & \red{3.58} \\
           &       & StdE{\footnotesize\,(bp)}      & \blue{0.17} & \blue{0.19} & \blue{0.21} & \blue{0.23} & \blue{0.25} & \blue{0.28} & \blue{0.30} \\
           & {6M}  & IV{\footnotesize\,(\%)}        & 18.53 & 18.84 & 19.17 & 19.53 & 19.91 & 20.31 & 20.73 \\
           &       & SE{\footnotesize\,(bp)}        & \red{-0.45} & \red{-0.19} & \red{0.60} & \red{2.07} & \red{4.38} & \red{7.66} & \red{12.05} \\
           &       & StdE{\footnotesize\,(bp)}      & \blue{0.19} & \blue{0.22} & \blue{0.24} & \blue{0.27} & \blue{0.30} & \blue{0.34} & \blue{0.37} \\
           & {1Y}  & IV{\footnotesize\,(\%)}        & 18.75 & 19.20 & 19.70 & 20.24 & 20.83 & 21.45 & 22.10 \\
           &       & SE{\footnotesize\,(bp)}        & \red{-0.17} & \red{1.12} & \red{3.94} & \red{8.74} & \red{15.89} & \red{25.75} & \red{38.62} \\
           &       & StdE{\footnotesize\,(bp)}      & \blue{0.21} & \blue{0.25} & \blue{0.29} & \blue{0.33} & \blue{0.38} & \blue{0.42} & \blue{0.47} \\
\bottomrule
\end{tabular}
\end{center}
\end{table}

\label{sec:varyingrhoSABR}
\;

\begin{table}[H]
\caption{Implied volatilities in \%, signed error of implied volatilities (Approximation - Monte Carlo) in basis points, and standard errors of the benchmark Monte-Carlo simulation in basis points (IV(\%), SE(bp), StdE(bp) respectively), computed in the classical SABR model with $\rho$ varying from 160\% to 40\% of its `safe' 3-piece value.}
\label{table:varyingrhoSABR}
\begin{center}
\begin{tabular}{lc l rrrrrrr}
\toprule
& & & \multicolumn{7}{c}{scaling factor for $\rho$ safe value} \\
\cmidrule{1-2}\cmidrule{4-10} 
Moneyness & $T$ & IV/SE/StdE & \ 160\% & 140\% & 120\% & 100\% & 80\% & 60\% & 40\% \\
\cmidrule{1-2}\cmidrule{4-10}  

 ATM    & {1M}   & IV{\footnotesize\,(\%)}        & 18.07  & 18.06  & 18.05  & 18.04  & 18.03  & 18.02  & 18.01 \\
        &        & SE{\footnotesize\,(bp)}        & \red{-0.20} & \red{-0.20} & \red{-0.21} & \red{-0.22} & \red{-0.23} & \red{-0.24} & \red{-0.25} \\
        &        & StdE{\footnotesize\,(bp)}      & \blue{0.23} & \blue{0.24} & \blue{0.24} & \blue{0.25} & \blue{0.26} & \blue{0.26} & \blue{0.27} \\
        & {3M}   & IV{\footnotesize\,(\%)}        & 18.21  & 18.17  & 18.14  & 18.11  & 18.08  & 18.05  & 18.03 \\
        &        & SE{\footnotesize\,(bp)}        & \red{0.01} & \red{0.01} & \red{0.01} & \red{0.01} & \red{0.00} & \red{0.00} & \red{0.00} \\
        &        & StdE{\footnotesize\,(bp)}      & \blue{0.20} & \blue{0.21} & \blue{0.22} & \blue{0.23} & \blue{0.24} & \blue{0.25} & \blue{0.26} \\
        & {6M}   & IV{\footnotesize\,(\%)}        & 18.43  & 18.35  & 18.28  & 18.21  & 18.16  & 18.11  & 18.06 \\
        &        & SE{\footnotesize\,(bp)}        & \red{-0.31} & \red{-0.23} & \red{-0.17} & \red{-0.14} & \red{-0.12} & \red{-0.11} & \red{-0.11} \\
        &        & StdE{\footnotesize\,(bp)}      & \blue{0.17} & \blue{0.18} & \blue{0.19} & \blue{0.21} & \blue{0.22} & \blue{0.24} & \blue{0.25} \\
        & {1Y}   & IV{\footnotesize\,(\%)}        & 18.86  & 18.70  & 18.55  & 18.43  & 18.31 & 18.21  & 18.13 \\
        &        & SE{\footnotesize\,(bp)}        & \red{-1.96} & \red{-1.44} & \red{-1.06} & \red{-0.79} & \red{-0.61} & \red{-0.50} & \red{-0.43} \\
        &        & StdE{\footnotesize\,(bp)}      & \blue{0.14} & \blue{0.15} & \blue{0.16} & \blue{0.18} & \blue{0.20} & \blue{0.22} & \blue{0.24} \\
\midrule

 Put 25 & {1M}   & IV{\footnotesize\,(\%)}        & 17.68  & 17.71  & 17.75  & 17.78  & 17.82  & 17.86  & 17.91 \\
        &        & SE{\footnotesize\,(bp)}        & \red{-0.09} & \red{-0.05} & \red{-0.01} & \red{0.00} & \red{-0.01} & \red{-0.04} & \red{-0.08} \\
        &        & StdE{\footnotesize\,(bp)}      & \blue{0.13} & \blue{0.13} & \blue{0.14} & \blue{0.15} & \blue{0.16} & \blue{0.16} & \blue{0.17} \\
        & {3M}   & IV{\footnotesize\,(\%)}        & 17.57  & 17.60  & 17.64  & 17.68  & 17.73  & 17.79  & 17.85 \\
        &        & SE{\footnotesize\,(bp)}        & \red{-1.06} & \red{-0.58} & \red{-0.23} & \red{0.00} & \red{0.13} & \red{0.17} & \red{0.16} \\
        &        & StdE{\footnotesize\,(bp)}      & \blue{0.10} & \blue{0.11} & \blue{0.11} & \blue{0.12} & \blue{0.14} & \blue{0.15} & \blue{0.16} \\
        & {6M}   & IV{\footnotesize\,(\%)}        & 17.55 & 17.56  & 17.59  & 17.63  & 17.68  & 17.74  & 17.81 \\
        &        & SE{\footnotesize\,(bp)}        & \red{-4.04} & \red{-2.59} & \red{-1.50} & \red{-0.73} & \red{-0.25} & \red{0.00} & \red{0.07} \\
        &        & StdE{\footnotesize\,(bp)}      & \blue{0.08} & \blue{0.09} & \blue{0.09} & \blue{0.11} & \blue{0.12} & \blue{0.13} & \blue{0.15} \\
        & {1Y}   & IV{\footnotesize\,(\%)}        & 17.68  & 17.63  & 17.61  & 17.62  & 17.65  & 17.70  & 17.78 \\
        &        & SE{\footnotesize\,(bp)}        & \red{-12.65} & \red{-8.54} & \red{-5.36} & \red{-3.03} & \red{-1.47} & \red{-0.55} & \red{-0.14} \\
        &        & StdE{\footnotesize\,(bp)}      & \blue{0.11} & \blue{0.09} & \blue{0.08} & \blue{0.09} & \blue{0.10} & \blue{0.12} & \blue{0.14} \\
\midrule

 Put 10 & {1M}   & IV{\footnotesize\,(\%)}        & 17.39  & 17.45  & 17.51  & 17.58  & 17.65  & 17.73  & 17.82 \\
        &        & SE{\footnotesize\,(bp)}        & \red{-0.24} & \red{-0.09} & \red{0.03} & \red{0.10} & \red{0.11} & \red{0.09} & \red{0.04} \\
        &        & StdE{\footnotesize\,(bp)}      & \blue{0.08} & \blue{0.09} & \blue{0.09} & \blue{0.10} & \blue{0.11} & \blue{0.12} & \blue{0.12} \\
        & {3M}   & IV{\footnotesize\,(\%)}        & 17.20  & 17.24  & 17.30  & 17.38  & 17.47  & 17.58  & 17.70 \\
        &        & SE{\footnotesize\,(bp)}        & \red{-2.88} & \red{-1.79} & \red{-0.94} & \red{-0.34} & \red{0.04} & \red{0.23} & \red{0.26} \\
        &        & StdE{\footnotesize\,(bp)}      & \blue{0.06} & \blue{0.06} & \blue{0.07} & \blue{0.08} & \blue{0.09} & \blue{0.10} & \blue{0.11} \\
        & {6M}   & IV{\footnotesize\,(\%)}        & 17.20  & 17.19  & 17.21  & 17.26  & 17.34  & 17.46  & 17.61 \\
        &        & SE{\footnotesize\,(bp)}        & \red{-9.22} & \red{-6.31} & \red{-3.92} & \red{-2.11} & \red{-0.87} & \red{-0.14} & \red{0.14} \\
        &        & StdE{\footnotesize\,(bp)}      & \blue{0.07} & \blue{0.06} & \blue{0.06} & \blue{0.06} & \blue{0.07} & \blue{0.09} & \blue{0.10} \\
        & {1Y}   & IV{\footnotesize\,(\%)}        & 17.50  & 17.35  & 17.25  & 17.22  & 17.26  & 17.35  & 17.51 \\
        &        & SE{\footnotesize\,(bp)}        & \red{-25.42} & \red{-18.44} & \red{-12.35} & \red{-7.40} & \red{-3.74} & \red{-1.38} & \red{-0.19} \\
        &        & StdE{\footnotesize\,(bp)}      & \blue{0.13} & \blue{0.10} & \blue{0.08} & \blue{0.06} & \blue{0.06} & \blue{0.07} & \blue{0.09} \\
\bottomrule
\end{tabular}
\end{center}
\end{table}